\documentclass[11pt, a4paper]{article}

\usepackage{microtype}

\usepackage{thmtools, thm-restate}

\usepackage[utf8]{inputenc} % Character encoding
\usepackage[T1]{fontenc} % Font encoding
\usepackage{geometry} % JHEP page layout
\usepackage{tikz}
\usetikzlibrary{arrows.meta,decorations.pathmorphing,fadings, calc}

\usepackage[%
    colorlinks      = true,
  linktocpage     = true,
  linkcolor       = blue,
  citecolor       = blue,
  urlcolor        = blue,
  breaklinks =true,
  hypertexnames=false,
]{hyperref}

\usepackage{setspace}
\usepackage[affil-it]{authblk}

\usepackage{amsmath, amsthm, amsfonts, amssymb, amsbsy, bigstrut, graphicx, enumerate,  upref, longtable, comment, booktabs, array, caption, subcaption}

\usepackage[numbers, sort&compress]{natbib}

\newcommand{\ep}{\epsilon}

\newcommand{\mc}{\mathcal{C}}

\newtheorem{thm}{Theorem}[section]
\newtheorem{lmm}[thm]{Lemma}
\newtheorem{cor}[thm]{Corollary}

\theoremstyle{definition}

\newcommand{\bx}{\mathbf{x}}

\newcommand{\md}{\mathcal{D}}

\newcommand{\sign}{\operatorname{sign}}

\newcommand{\Res}{\operatorname{Res}}

\numberwithin{equation}{section}

\newcommand{\balpha}{{\boldsymbol \alpha}}
\renewcommand{\bx}{{\boldsymbol x}}

\renewcommand{\Re}{\operatorname{Re}}

\usepackage[utf8]{inputenc}
\usepackage{amsmath}
\usepackage{amsfonts}
\usepackage{bbm}
\usepackage{mathtools}
\usepackage{tikz}
\usetikzlibrary{decorations.pathreplacing}
\usepackage{amsthm}
\usepackage[english]{babel}
\usepackage{fancyhdr}
\usepackage{amssymb}
\usepackage{enumitem}
\usepackage{qtree}
\usepackage{mathrsfs}

\newcommand{\Z}{\mathbb{Z}}

\newcommand{\R}{\mathbb{R}}

\newcommand{\C}{\mathbb{C}}

\newcommand{\E}{\mathbb{E}}

\renewcommand{\S}{\mathbb{S}}

\renewcommand{\tilde}{\widetilde}
\renewcommand{\hat}{\widehat}

\newcommand{\normord}[1]{:\mathrel{#1}:}
\renewcommand{\i}{\mathrm{i}}
\renewcommand{\Im}{\operatorname{Im}}
\begin{document}
%\title{Some exact calculations for timelike Liouville field theory}
\title{Exact calculations beyond charge neutrality in timelike Liouville field theory}
%\title{An exactly solvable case of timelike Liouville field theory}
%[Recent developments in measures of association]
%\title{xxx}
\author{Sourav Chatterjee\thanks{Department of Statistics, Stanford University, USA. Email: \href{mailto:souravc@stanford.edu}{\tt souravc@stanford.edu}. 
}}
\affil{Stanford University}
%The author was partially supported by NSF grants DMS-2113242 and DMS-2153654. The author thanks xxx for helpful comments.

%\address{Departments of mathematics and statistics, Stanford University}
%\email{souravc@stanford.edu}
%\dedicatory
%Research partially supported by NSF grants DMS-1855484 and DMS-2113242}
%\keywords{}
%\subjclass[2020]{}

\maketitle

%\begin{center}
%{\it \small In honor of friend and teacher Prof.~Rajeeva L.~Karandikar on the occasion of his 65$^{th}$ birthday.}
%\end{center}

\begin{abstract}
Timelike Liouville field theory (also known as imaginary Liouville theory or imaginary
Gaussian multiplicative chaos) is expected to describe two-dimensional quantum gravity
in a positive-curvature regime, but its path integral is not a probability measure and
rigorous exact computations are currently available only in the charge-neutral
(integer screening) case.  In this paper we show that at the special coupling
$b=\frac{1}{\sqrt2}$, the Coulomb-gas expansion of the timelike path integral becomes explicitly
computable beyond charge neutrality.  The reason is that the $n$-fold integrals generated
by the interaction acquire a Vandermonde/determinantal structure at $b=\frac{1}{\sqrt2}$, which
allows exact evaluation in terms of classical special functions.

We derive Mellin--Barnes type representations (involving the Barnes $G$-function and, in
a three-point case, Gauss hypergeometric functions) for the zero- and one-point functions,
for an antipodal two-point function, and for a three-point function with a resonant
insertion $\alpha_2=b$.  We then address the subtle zero-mode integration: after a Gaussian
regularization we obtain an explicit renormalized partition function
$C(\frac{1}{\sqrt2},\mu)=e(4\pi\sqrt2 \mu)^{-1}$, identify distributional limits in the physically
relevant regime $\alpha_j=\frac{1}{2}Q+\mathrm{i} P_j$, and compare with the Hankel-contour prescription
recently proposed in the physics literature.  These results provide the first 
rigorously controlled family of exact calculations in timelike Liouville theory outside
charge neutrality.
%Timelike Liouville field theory is a model of 2D quantum gravity that has attracted renewed interest in recent years. It is known as imaginary Liouville quantum gravity or imaginary Gaussian multiplicative chaos in probability theory.  Exact calculations for this model have been rigorously verified under the so-called `charge neutrality condition', but dropping this condition remains a major challenge. This paper shows that for the special value $\frac{1}{\sqrt{2}}$ of the Liouville coupling constant $b$, the path integral for the model admits some exact calculations even in the absence of charge neutrality. Since the nature of this model is not yet fully understood, these calculations may yield useful insights.
\end{abstract}

\setcounter{tocdepth}{1}
\tableofcontents

%\section{A new coefficient of correlation}

\section{Introduction}
\addtocontents{toc}{\protect\setcounter{tocdepth}{2}}
\subsection{Background}

Liouville field theory is a central object in two-dimensional conformal field theory and
in Polyakov's formulation of 2D quantum gravity: the Liouville field plays the role of a
random conformal factor and formally produces a random geometry.
The \emph{spacelike} theory admits a rigorous probabilistic realization via Gaussian
multiplicative chaos and Liouville quantum gravity, and an extensive mathematical
literature has developed around it; see, for instance, \cite{chatterjeewitten25, berestyckipowell24,chatterjee25} for surveys.

The \emph{timelike} (or \emph{imaginary}) Liouville theory is, by contrast, much less
understood mathematically.  In the probability literature it is often referred to as
\emph{imaginary Gaussian multiplicative chaos} \cite{bonnefontetal25, lacoinetal15} or 
as \emph{imaginary Liouville theory} \cite{guillarmouetal23, xiaoxie25}.  It has gained renewed attention
recently in physics because of new proposals using it to build models of 2D gravity
\cite{alexandrovetal24, anninosetal21, anninosetal25, blommaertetal24, collieretal24, collieretal25, roussillontsiares25, muhlmannetal25, rangamanizheng25}.  One motivation is that timelike Liouville theory is expected to
capture a positive-curvature semiclassical regime \cite{chatterjee25, polchinski89} akin to Einsteinian general relativity, in contrast with the
spacelike theory, which leads to negative-curvature geometries \cite{lacoinetal22}.

In timelike Liouville theory on the unit sphere $\S^2$, the (formal) action of a field
$\phi:\S^2\to\mathbb{R}$ is given by
\[
I(\phi) := \frac{1}{4\pi}\int_{\S^2}  (\phi(x)\Delta_{\S^2}\phi(x) + 2Q \phi(x) + 4\pi \mu \normord{e^{2b\phi(x)}})da(x),
%I(\phi) :=\frac1{4\pi}\int_{\S^2}\biggl(\phi(x)\Delta_{\S^2}\phi(x) + 2Q\phi(x)+ 4\pi\mu \normord{e^{2b\phi(x)}}\biggr)da(x),
\]
where $\Delta_{\S^2}$ is the spherical Laplacian, $a$ is the area measure, $b>0$ is the
Liouville coupling, $\mu>0$ is the cosmological constant, and $Q:=b-\frac1b$.
The normal-ordered exponential $\normord{e^{2b\phi(x)}}$ is a renormalized version of $e^{2b\phi(x)}$;
it can be viewed heuristically as $\exp(2b\phi(x)+2b^2G_{\S^2}(x,x))$, where the Green's function  $G_{\S^2}$ is the
inverse of $-\frac1{2\pi}\Delta_{\S^2}$ on mean-zero functions.  Explicitly,
\begin{align}\label{gform}
G_{\S^2}(x,y) = -\ln\|x-y\|-\frac{1}{2}+\ln 2,
\end{align}
where $\|\cdot\|$ is the Euclidean norm in $\mathbb{R}^3$ \cite[Lemma~3.2.6]{chatterjee25}.

The basic difficulty is that the timelike kinetic term has the ``wrong sign'', so the
path integral weight $e^{-I(\phi)}$ is not associated with a probability measure and the
usual probabilistic tools available for the spacelike theory are not directly applicable.
A major theme in recent work \cite{chatterjee25, guillarmouetal23} is to nonetheless extract rigorous information
from the path integral by interpreting the Gaussian part through an analytic continuation
(or, equivalently, through an \emph{imaginary} Gaussian free field).

\subsection{Correlation functions and the charge-neutral barrier}

For $\alpha_1,\dots,\alpha_k\in\mathbb{C}$ and distinct points $x_1,\dots,x_k\in \S^2$, the
$k$-point correlation function is formally given by the path integral
\begin{align}\label{favg2}
C(\balpha, \bx, b, \mu) :=\int \biggl(\prod_{j=1}^k \normord{e^{2\alpha_j \phi(x_j)}} \biggr)e^{-I(\phi)}\mathcal{D}\phi,
\end{align}
%(where $\int \ldots \mathcal{D}\phi$ denotes integration with respect to `Lebesgue measure' on the space of real-valued function on $\S^2$, $\balpha = (\alpha_1,\ldots,\alpha_k)$, and $\bx = (x_1,\ldots,x_k)$) 
with the analogous definition of the partition function (the ``zero-point function'')
\[
C(b,\mu):=\int e^{-I(\phi)}\md \phi.
\]
On the physics side, a widely accepted formula for the three-point function is the
\emph{timelike DOZZ formula} derived in  \cite{schomerus03, zamolodchikov05, kostovpetkova06, kostovpetkova07, kostovpetkova07a}.  However, these derivations 
do not compute the correlation functions directly from the path integral; instead they rely on
bootstrap/recursion relations known in the spacelike theory (notably Teschner's relations~\cite{teschner95}) and assume an analytic continuation to the timelike regime.  At rational central charge, the bootstrap picture has recently been put on a
more analytic footing: \citet{roussillontsiares25} derive the relevant Virasoro fusion and modular
kernels and, in particular, demonstrate crossing symmetry and modular covariance for timelike
Liouville at rational $b^2$.
% Optionally, if you want to explicitly nod to earlier numerical bootstrap literature:
% This can be viewed as an analytic counterpart to earlier numerical studies at $c\le 1$, e.g.\ \citet{ribaultsantachiara15}.
 \citet*{harlowetal11} emphasized the conceptual importance of more transparent path integral
computations and provided evidence that such computations should be possible.

The only family of exact calculations that has been rigorously verified to date from the
timelike path integral is essentially restricted to the \emph{charge-neutral} regime.
A key parameter~is
\[
w:=\frac{Q-\sum_{j=1}^k\alpha_j}{b}.
\]
When $w$ is a positive integer (the charge neutrality condition), Giribet \cite{giribet12}
showed how the timelike DOZZ formula can be derived from the path integral, and this was
later made rigorous (with clarifications) in \cite{chatterjee25}, following related work in a
nearby model by \citet*{guillarmouetal23}.

Outside charge neutrality, even defining and evaluating \eqref{favg2} becomes subtle.
This regime is also the one of primary physical interest: conformal field theory
considerations suggest choosing
\begin{align}\label{alphajform}
\alpha_j=\frac Q2+\i P_j,\qquad P_j\in\mathbb{R},
\end{align}
for which $w$ is typically non-integer (and already purely imaginary when $k=2$).
A recent proposal of \citet{usciatietal25} advocates a specific contour prescription
for integrating the zero mode, supported by numerics and exact computations in a
related circle model.  Whether that proposal can be treated fully rigorously remains an
open question.

\medskip
\noindent\textbf{Main message of this paper.}
We show that one can go beyond charge neutrality at the special coupling
\[
b=\frac1{\sqrt2}.
\]
More precisely, starting from the path integral we obtain explicit formulas for the
zero-, one-, and (antipodal) two-point functions, and for a particular three-point
function with a resonant insertion $\alpha_2=b$, in a parameter range that includes the
physically important choice $\alpha_j=\frac{1}{2}Q+\i P_j$.  These formulas are expressed through
Mellin--Barnes type integrals involving the Barnes $G$-function (and in the three-point
case also Gauss hypergeometric functions), and they allow us to analyze the delicate
integration over the zero mode.  In particular:
\begin{itemize}
\item we obtain an explicit fixed-zero-mode formula for the zero-point function and,
after Gaussian regularization of the zero-mode integral, an explicit renormalized
partition function (Theorems~\ref{zerocthm}--\ref{zerothm});
\item in the two-point case with $\alpha_j=\frac{1}{2}Q+\i P_j$, the regularized correlation function
has a distributional limit implementing momentum conservation (Theorem~\ref{distthm});
\item in the three-point setting with $\alpha_2=b$, the regularized correlation exhibits a
genuine pole, matching the expected singularity structure of the timelike DOZZ formula
(Theorem~\ref{threethm});
\item we also compute the effect of integrating the zero mode along the Hankel-type contour
advocated in \cite{usciatietal25} and compare it with the real-line prescription; already for the
partition function and two-point function the outcomes differ (Theorems~\ref{zerohankel} and~\ref{distthm2}).
\end{itemize}
It is perhaps also worth noting that the same coupling $b = \frac{1}{\sqrt{2}}$ (the free-fermion point) has recently been shown by Bauerschmidt and collaborators~\cite{bauerschmidtwebb24, bauerschmidtetal25} to admit strikingly explicit exact calculations in the massless sine–Gordon model, though any direct connection with timelike Liouville theory is at present unclear.

\subsection{Expanding the path integral and why $b=\frac{1}{\sqrt2}$ is special}

Write the field as $\phi=c+X$, where the \emph{zero mode} is
\[
c=c(\phi):=\frac1{4\pi}\int_{\S^2}\phi(x)\,da(x),
\]
and $X$ has mean zero.  Formally, this splits the path integral into an integral over
$c\in\mathbb{R}$ and an integral over mean-zero fields.  For equation~\eqref{favg2} this yields
\begin{equation}\label{cororiginal}
C(\balpha,\bx,b,\mu)=\int_{-\infty}^{\infty}e^{-2wbc}\,C(\balpha,\bx,b,\mu,c)dc,
\end{equation}
where the fixed-zero-mode object is
\begin{align*}
&C(\balpha,\bx,b,\mu,c):=\int\biggl(\prod_{j=1}^k \normord{e^{2\alpha_j X(x_j)}}\biggr)e^{-I_c(X)}\md X,\\
&I_c(X):=\frac{1}{4\pi}\int_{S^2}(X(x)\Delta_{\S^2}X(x)+4\pi\mu e^{2bc}\normord{e^{2bX(x)}})da(x).
\end{align*}
As explained in \cite{chatterjee25}, one can give a rigorous meaning to $C(\balpha,\bx,b,\mu,c)$ by
interpreting the Gaussian part through an imaginary Gaussian free field.  Concretely,
one obtains
\begin{align*}
C(\balpha, \bx, b, \mu,c) &:= \E\biggl[ \biggl(\prod_{j=1}^k \normord{e^{2\i \alpha_j Z(x_j)}} \biggr)\exp\biggl(-\mu e^{2bc} \int_{\S^2} \normord{e^{2\i b Z(x)}}da(x)\biggr)\biggr],
\end{align*}
where $Z$ is the mean-zero Gaussian free field on $\S^2$ with covariance $G_{\S^2}$.
Expanding the exponential produces a Coulomb-gas series
\begin{align}\label{zcor}
C(\balpha, \bx, b, \mu,c) &= \biggl(\prod_{1\le j<j'\le k} e^{-4\alpha_j\alpha_{j'}  G_{\S^2}(x_j,x_{j'})}\biggr)\biggl\{1+\sum_{n=1}^\infty \frac{(-\mu e^{2bc})^n}{n!}a_n\biggr\},
\end{align}
with coefficients 
\begin{align}
a_n&:= \int_{(\S^2)^n}\exp\biggl(-4b\sum_{j=1}^k\sum_{l=1}^n \alpha_j G_{\S^2}(x_j, y_l) \notag\\
&\hskip1in -4b^2 \sum_{1\le l < l'\le n} G_{\S^2}(y_l,y_{l'}) \biggr) da(y_1)\cdots da(y_n).\label{andef}
\end{align}
The following lemma records the basic analytic input needed to justify this expansion: it gives absolute convergence of the integrals defining $a_n$, absolute convergence of the resulting series for $C(\balpha, \bx, b, \mu,c)$, and the regularity in parameters that will be used repeatedly later.  It is proved in \textsection\ref{convergencepf}. 
\begin{lmm}\label{convergence}
The integral in equation \eqref{andef} and the series in equation \eqref{zcor} are absolutely convergent if 
\begin{align}\label{maincond}
b\in (0,1) \text{ and } \Re(\alpha_j) > -\frac{1}{2b} \text{ for each $j$.}
\end{align}
Moreover, if $b\in (0,1)$, then $C(\balpha, \bx, b, \mu,c)$ and $a_n$ are continuous in $(\balpha,\bx,c)$ and analytic in $(\balpha,c)$ in the region where $\Re(\alpha_j)>-\frac{1}{2b}$ for each $j$ and $x_1,\ldots,x_k$ are distinct. Lastly, the coefficient $a_n$ satisfies the bound
\begin{align}\label{wagner2}
a_n \le n^{b^2n} e^{Cn}
\end{align}
where $C$ depends only on $b$ and $\alpha_1,\ldots,\alpha_k$.
\end{lmm}

The special value $b=\frac{1}{\sqrt2}$ is precisely the point at which these $n$-fold integrals
become \emph{determinantal}: after stereographic projection the interaction term yields a
Vandermonde factor $\prod_{l<l'}|z_l-z_{l'}|^2$, allowing exact evaluation by
standard random-matrix/orthogonal-polynomial manipulations.  This makes it possible to
compute $a_n$ explicitly and then resum~\eqref{zcor} by a Mellin--Barnes transform,
leading to the explicit integral representations stated in Section~\ref{resultssec}. 

It is natural to ask whether other rational values of $b^2$ might admit comparably rigid or
``integrable'' features.
While the \emph{determinantal} simplification exploited here is specific to $b^2=\tfrac12$ at the level
of the Coulomb-gas integrals, recent bootstrap results from the physics literature indicate that for \emph{all} rational $b^2$ one can
construct consistent Virasoro modular and fusion kernels in the timelike domain, implying crossing
symmetry and modular covariance in that rational regime \cite{roussillontsiares25}.  Related progress for a  supersymmetric extension, including explicit three-point
data derived from the same general philosophy, appears in \citet{muhlmannetal25}; see also
\citet{rangamanizheng25}. It is, however, not clear how to extend the results and techniques of the present paper to obtain rigorous results for general rational $b^2$.

Finally, equation~\eqref{cororiginal} shows that the remaining analytic problem is the integration
over the zero mode, for which there is currently no universally accepted contour away from
charge neutrality.  In this paper we primarily integrate over the real line (with a
Gaussian regularization when necessary) and we also analyze the Hankel-type contour
prescription of \citet{usciatietal25} in the same solvable setting, providing a direct comparison.

\medskip
\noindent\textbf{Organization.}
\textsection\ref{resultssec} states the main results.  \textsection\ref{sketchsec} sketches the proofs of Theorems~\ref{zerocthm} and~\ref{zerothm} 
to highlight the main ideas, and \textsection\ref{proofsec} contains the complete proofs.  Auxiliary bounds
and analytic lemmas are collected in the Appendix.

\section{Results}\label{resultssec}
This section records the main explicit evaluations obtained in the paper at the “solvable” point
$b=\frac{1}{\sqrt{2}}$. For $k=0,1,2$ (and for a resonant $k=3$ case) we first compute correlation functions with the zero mode fixed at $c$, i.e., $C(\balpha,\bx,b,\mu,c)$. We then address the zero-mode integration in equation~\eqref{cororiginal}, which is typically not absolutely convergent: throughout we interpret it via a Gaussian regularization $e^{-\epsilon^2 c^{2}}$ and take $\epsilon\to 0$ (sometimes after an explicit
renormalization). Finally, in each case we also record what one obtains by replacing the real-line
integral over $c$ by the Hankel-type contour proposed in \cite{usciatietal25}.

\subsection{The zero-point function at $b=\frac{1}{\sqrt{2}}$}\label{zerosec}
We now give an explicit Mellin–Barnes type representation of $C(2^{-1/2},\mu,c)$; that is, the zero-point function with the zero mode fixed at $c$. The formula involves the Barnes $G$-function. Recall that $G$ is the entire function defined by
\begin{align}\label{gformula}
G(z+1) := (2\pi)^{z/2} \exp\biggl(-\frac{z+ z^2(1+\gamma_{\mathrm{E}})}{2}\biggr)\prod_{k=1}^\infty\biggl\{\biggl(1+\frac{z}{k}\biggr)^k\exp\biggl(\frac{z^2}{2k}-z\biggr)\biggr\},
\end{align}
where $\gamma_{\mathrm{E}}$ is Euler's constant. From equation \eqref{gformula} one readily checks that $G$ is entire, and has zeros at the nonpositive integers and nowhere else.

Let $\Gamma$ denote the classical Gamma function. Recall that $\Gamma$ has no zeros, and the only
poles are at the nonpositive integers. Thus, $\Gamma$ has a logarithm on the simply connected
domain $\mathbb{C}\setminus(-\infty,0]$, which we denote by $\Pi$. Specifically, $\Pi$ is an analytic function on this domain such that $e^{\Pi}=\Gamma$, rendered unique by the
condition that $\Pi$ is real-valued on $(0,\infty)$ (so $\Pi=\ln \Gamma$ there). Equivalently, $\Pi$ admits the representation
\begin{align}\label{pirep}
\Pi(z) = \int_C \psi(w) dw,
\end{align}
where $C$ is any contour from $1$ to $z$ that lies entirely in $\mathbb{C}\setminus(-\infty,0]$, and
$\psi=\Gamma'/\Gamma$ is the digamma function. We stress that $\Pi$ is an analytic logarithm of $\Gamma$ on $\mathbb{C}\setminus(-\infty,0]$, which is {\it not} obtained by composing $\Gamma$ with an analytic branch of $\log$. Having defined $\Pi$, we set, for $z\in\mathbb{C}\setminus(-\infty,0]$ and $w\in\mathbb{C}$,
$\Gamma(z)^{\,w}:=\exp(w\Pi(z))$. Separately, for $z\in\mathbb{C}\setminus(-\infty,0]$ we write $z^{w}:=\exp(w\log z)$ for the principal branch of $\log$ on the same domain; the notation $\Gamma(z)^{w}$ always refers to the definition via $\Pi$ above.

The following result gives the zero-point function at $b=\frac{1}{\sqrt{2}}$, when the path integral is
restricted to fields with zero mode equal to $c$.
\begin{thm}\label{zerocthm}
For any $c\in \R$ and $\mu >0$, we have
\begin{align*}
C(2^{-1/2}, \mu, c) = \frac{1}{2\pi}\int_{-\infty}^\infty \Gamma(1-\i y) f(-1+\i y) (\mu e^{\sqrt{2}c})^{-1+\i y} dy,%\Gamma(\i y)^{-\i y}(4\pi \mu e^{\sqrt{2}c})^{-1+\i y} e^{1-\frac{1}{2}y^2-\frac{3}{2}\i y}G(1+\i y)^2dy.
\end{align*}
where $f:\C \setminus(-\infty,-1]\to\C$ is the analytic function 
\begin{align*}%\label{fdef}
f(z) := \frac{(4\pi)^z e^{\frac{1}{2}z(z-1)}G(z+1)^2}{\Gamma(z+1)^{z-1}}.
\end{align*}
%where $c\in (-1,0)$ is arbitrary, $G$ denotes the Barnes $G$-function, $\Gamma$ is the usual Gamma function, and the denominator is defined using the analytic branch of $\log \Gamma$ in $\C \setminus (-\infty,0]$ that is real-valued on $(0,\infty)$.
\end{thm}
A sketch of the proof of Theorem \ref{zerocthm} is given in \textsection\ref{sketchsec}, and the complete proof is in \textsection\ref{zeroproofsec}. Next, we compute the (unrestricted) zero-point function by integrating over the zero mode, using the relation \eqref{cororiginal}. For the zero-point function, 
\[
w = \frac{Q}{b} = 1-\frac{1}{b^2}.
\]
With $b=\frac{1}{\sqrt{2}}$, we get $w=-1$. After multiplying $C(2^{-1/2}, \mu, c)$ by $e^{-2wbc}$ and integrating this over $c\in \R$, we get 
\begin{align}\label{chalf0}
C(2^{-1/2}, \mu) = \int_{-\infty}^\infty e^{\sqrt{2}c} C(2^{-1/2}, \mu,c) dc,
\end{align}
where $C(2^{-1/2}, \mu,c)$ is as in Theorem \ref{zerocthm}. The integral in equation \eqref{chalf0} is not absolutely convergent, so we interpret it as an oscillatory integral via a Gaussian cutoff. We regularize it by defining
\[
C_\ep(2^{-1/2}, \mu) := \int_{-\infty}^\infty e^{\sqrt{2}c-\ep^2 c^2} C(2^{-1/2}, \mu,c) dc,
\]
and then define
\[
C(2^{-1/2},\mu) := \lim_{\ep\to 0} C_\ep(2^{-1/2},\mu).
\]
The next theorem shows that this prescription yields a well-defined limit, with an explicit closed form.  The proof is in \textsection\ref{zerothmpf}, and a brief sketch of the proof is in \textsection\ref{sketchsec}. % at $b=\frac{1}{\sqrt{2}}$.
\begin{thm}\label{zerothm}
For any $\mu>0$,
\[
C(2^{-1/2}, \mu) = \frac{e}{4\pi\sqrt{2}\mu}.
\]
\end{thm}
The timelike Liouville zero-point function depends on the prescription used to define the
zero-mode integral. Outside the charge-neutral regime there is no universally accepted choice,
and different prescriptions lead to inequivalent answers.

One indication of this is that if one starts from the timelike DOZZ three-point function and
formally extracts a zero-point function from it (e.g., by an identity limit), the result vanishes
\cite[Section~4.1]{anninosetal21}. \citet*{anninosetal21} proposed to define the
zero-point function by analytic continuation of the \emph{spacelike} Liouville three-point
function. Their prescription gives
\[
C_{\mathrm{AC}}(b,\mu)
= \pm \i  (\pi \mu \gamma(-b^2))^{-\frac{1}{b^2} + 1}
\frac{1+b^2}{\pi^3 q\gamma(-b^2)\gamma(-b^{-2})}e^{q^2 - q^2 \ln 4},
\]
where $\gamma(z) := \frac{\Gamma(z)}{\Gamma(1-z)}$ and $q := \frac{1}{b} - b$.
At the value $b=2^{-1/2}$ we have $-b^{-2}=-2$, hence
\[
\gamma(-b^{-2})=\frac{\Gamma(-2)}{\Gamma(3)},
\]
which has a pole. Accordingly, the meromorphic prefactor $1/\gamma(-b^{-2})$ has a zero there,
and the analytic-continuation expression yields
\[
C_{\mathrm{AC}}(2^{-1/2},\mu)=0,
\]
at least at the level of direct substitution, i.e., without introducing an additional limiting
prescription at this pole. It is nevertheless noteworthy that the $\mu$-dependence in
$C_{\mathrm{AC}}(b,\mu)$ matches the scaling appearing in Theorem~\ref{zerothm}.

In contrast, Theorem~\ref{zerothm} uses the real-line prescription \eqref{cororiginal},
interpreted via Gaussian regularization, and produces a \emph{nonzero} value. Since both
prescriptions (and their corresponding behaviors) appear in the literature, we will keep track
of both cases below: real-line versus the Hankel-type contour suggested by
\citet{usciatietal25}. Concretely, their contour $\mathcal{C}$ runs along the horizontal ray
$b^{-1}\pi \i + [0,\infty)$ from $b^{-1}\pi \i+\infty$ to $b^{-1}\pi \i$, then down to $0$,
and finally out along $[0,\infty)$ (see Figure~\ref{usciatifig}). We refer to $\mathcal{C}$ as
a \emph{Hankel contour}, since it is a translate of the standard Hankel contour appearing in
integral representations of the Gamma function. Let $\tilde{C}(b,\mu)$ denote the zero-point
function defined with this contour. The following result, proved in
\textsection\ref{zerohankelpf}, shows that $\tilde{C}(2^{-1/2},\mu)=0$ for all~$\mu$.

% in your preamble:
% Preamble:
% \usepackage{tikz}

% Preamble:
% \usepackage{tikz}
% \usetikzlibrary{arrows.meta,fadings}

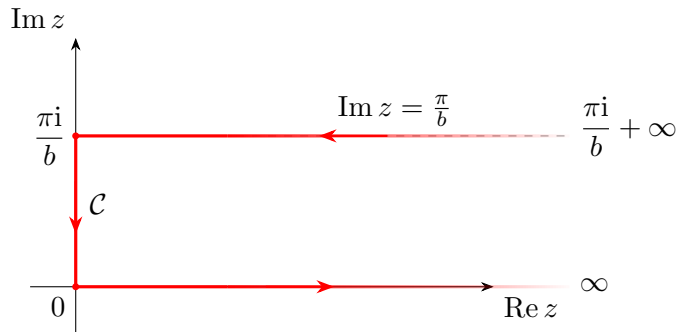
\begin{figure}[t]
\centering
\begin{tikzpicture}[>=Stealth, scale=1.0]

  % ---- parameters (visual only) ----
  \def\Y{2.0}      % plotted height for Im(z)=pi/b
  \def\Xmax{6.5}   % right extent of picture

  % ---- axes (black) ----
  \draw[black,->] (-0.6,0) -- (\Xmax-1,0) node[below right] {$\Re z$};
  \draw[black,->] (0,-0.6) -- (0,\Y+1.3) node[above left] {$\Im z$};

  % ---- guide line for Im(z)=pi/b ----
  \draw[black!40, dashed] (0,\Y) -- (\Xmax,\Y);
  \node[black, above] at (4.2,\Y) {$\Im z=\frac{\pi}{b}$};

  % ---- contour (colored) ----
  % finite (solid) parts
  \draw[very thick, red] (2.0,\Y) -- (0,\Y);
  \draw[very thick, red] (0,\Y) -- (0,0);
  \draw[very thick, red] (0,0) -- (2.0,0);

  % "to infinity" tails with true fading
  \draw[very thick, red, path fading=east]
    (2.0,\Y) -- (\Xmax,\Y);
  \draw[very thick, red, path fading=east]
    (2.0,0) -- (\Xmax,0);

  % ---- direction arrows ----
  \draw[red,->, thick] (4.1,\Y) -- (3.2,\Y);
  \draw[red,->, thick] (0,\Y*0.70) -- (0,\Y*0.35);
  \draw[red,->, thick] (2.5,0) -- (3.4,0);

  % ---- key points and labels (labels in black) ----
  \fill[red] (0,0) circle (1.3pt);
  \fill[red] (0,\Y) circle (1.3pt);
  \node[black, below left] at (0,0) {$0$};
  \node[black, left] at (0,\Y) {$\dfrac{\pi \i}{b}$};

  % ---- endpoint labels (black, positioned to avoid overlap) ----
  \node[black, above right] at (\Xmax,\Y-.4) {$\dfrac{\pi \i}{b}+\infty$};
  \node[black, below right] at (\Xmax,.25) {$\infty$};

  % ---- contour label (black) ----
  \node[black] at (0.3,0.55*\Y) {$\mathcal{C}$};

\end{tikzpicture}
\caption{Hankel-type contour $\mathcal{C}$ running from $\frac{\pi \i}{b}+\infty$ to $\frac{\pi \i}{b}$, then to $0$, and finally to $\infty$. Faded segments indicate the parts extending to infinity.}
\label{usciatifig}
\end{figure}

%In Subsection \ref{zerohankelpf}, we define $\tilde{C}(2^{-1/2},\mu)$ rigorously using Gaussian regularization, and prove the following result.
\begin{thm}\label{zerohankel}
Let $\tilde{C}(b,\mu)$ denote the zero-point function defined using the Hankel contour. Then for any $\mu>0$, $\tilde{C}(2^{-1/2},\mu)=0$.
\end{thm}

Note that this matches the prediction of \citet{anninosetal21} for the zero-point function at $b=\frac{1}{\sqrt{2}}$. However, as discussed in \textsection\ref{twosec} below, the same contour leads to a two-point function that differs from existing proposals in the physics literature; at present this leaves the contour choice unresolved. 

A useful way to interpret the discrepancy between the contours is via steepest-descent.
In the semiclassical analysis of \citet{anninosetal21}, the genus-zero amplitude can be organized as
a sum of contributions from two critical points of the timelike action: the real round-sphere saddle
and an additional (oscillatory) complex saddle. At the special value $b=2^{-1/2}$ these two
contributions can cancel, which is compatible with the vanishing of
$C_{\mathrm{AC}}(2^{-1/2},\mu)$ under analytic continuation, while retaining only the real saddle
yields a nonzero expression~\cite[Equation (5.6)]{anninosetal21}. This underscores that the choice
of integration contour for the timelike mode --- and, relatedly, the overall phase convention --- is an
essential part of the definition of the timelike path integral.

\subsection{The one-point function at $b=\frac{1}{\sqrt{2}}$}\label{onesec}
We now consider the case $k=1$, with a single $\alpha\in\mathbb{C}$ and $x\in \S^{2}$. By rotational symmetry the answer is independent of $x$; we fix $x=e_{3}=(0,0,1)$. 
We have the following result about the one-point function when the path integral is restricted to fields
with zero mode equal to $c$. The proof is in \textsection\ref{onecthmpf}.

\begin{thm}\label{onecthm}
Let $x = e_3$, and let $\alpha$ be any complex number with $\Re(\alpha)\in (-\frac{1}{\sqrt{2}}, 0]$. Let $w := -1-\sqrt{2}\alpha$, so that $\Re(w)\in [-1,0)$.
Then for any $c\in \R$ and $\mu >0$, we have
\begin{align*}
C(\alpha, x, 2^{-1/2}, \mu, c) = \frac{1}{2\pi } \int_{-\infty}^\infty \Gamma(-w-\i y) f(w+\i y)(\mu e^{\sqrt{2}c})^{w+\i y} dy, %-\frac{1}{2\pi}\int_{-\infty}^\infty \Gamma(1-\i y)\Gamma(\i y)^{-\i y}(4\pi \mu e^{\sqrt{2}c})^{-1+\i y} e^{1-\frac{1}{2}y^2-\frac{3}{2}\i y}G(1+\i y)^2dy.
\end{align*}
where 
\begin{align*}%\label{fdef}
f(z) := \frac{(4\pi)^z e^{\frac{1}{2}z(z-3-2w)}\Gamma(z+1)G(z+1)G(z-w)}{\Gamma(z-w)^{z}G(-w)}.%\frac{(4\pi)^z e^{\frac{1}{2}z(z-1)}G(z+1)^2}{\Gamma(z+1)^{z-1}},
\end{align*}
%where $c\in (-1,0)$ is arbitrary, $G$ denotes the Barnes $G$-function, $\Gamma$ is the usual Gamma function, and the denominator is defined using the analytic branch of $\log \Gamma$ in $\C \setminus (-\infty,0]$ that is real-valued on $(0,\infty)$.
\end{thm}
Next, just as for the zero-point function, we define the regularized one-point function
\[
C_\ep(\alpha, x, 2^{-1/2}, \mu) := \int_{-\infty}^\infty e^{-\sqrt{2}wc-\ep^2 c^2} C(\alpha, x, 2^{-1/2}, \mu,c) dc,
\]
where $x=e_{3}$ and $w=-1-\sqrt{2}\alpha$ as above. The next theorem identifies the leading asymptotics of $C_\epsilon$ as $\epsilon\to 0$. In particular:
(i) when $\alpha=0$ one recovers the zero-point function; (ii) when $\Re(\alpha)<0$ the unrenormalized quantity tends to $0$; and (iii) when $\Re(\alpha)=0$ with $\alpha\neq 0$, $C_\epsilon$ exhibits logarithmic oscillations in $\epsilon$, but becomes convergent after the explicit renormalization by $\epsilon^{\sqrt{2}\alpha}$.
\begin{thm}\label{onethm}
Let $x = e_3$ and $\alpha$ be a complex number with $\Re(\alpha)\in (-\frac{1}{\sqrt{2}}, 0]$. Let $w:= -1-\sqrt{2}\alpha$, so that $\Re(w)\in [-1,0)$. Then for any $\mu>0$,
\begin{align*}
&\lim_{\ep\to 0} \ep^{\sqrt{2}\alpha} C_\ep(\alpha, x, 2^{-1/2}, \mu) \\
&= \frac{ (4\sqrt{2}\pi\mu)^we^{-\frac{1}{2}w(w+3)}G(w+2)\cos(\frac{\pi}{2}(w+1)) \Gamma(-w)\Gamma(\frac{1}{2}w+1)}{\sqrt{\pi}G(-w)}.
\end{align*}
\end{thm}
Theorem \ref{onethm} is proved in \textsection\ref{onethmpf}. We are not aware of a widely accepted prediction for the one-point function outside charge neutrality, so we do not attempt a detailed comparison here. Let us also work out the one-point function defined via the Hankel contour. Proceeding as for the zero-point function, we get the formal expression
\begin{align*}
&\tilde{C}(\alpha,x,2^{-1/2},\mu) \\
&=\int_0^\infty e^{-\sqrt{2}wc} C(\alpha, x, 2^{-1/2}, \mu, c) dc  - \int_0^\infty e^{-\sqrt{2}w(c+\sqrt{2}\pi \i)} C(\alpha, x, 2^{-1/2}, \mu, c+\sqrt{2}\pi \i) dc \\
&\qquad -\i \int_0^{\sqrt{2}\pi} e^{-\sqrt{2}\i w t} C(\alpha,x, 2^{-1/2}, \mu, \i t) dt\\
&= (1-e^{-2\pi \i w})\int_0^\infty e^{-\sqrt{2}wc} C(\alpha,x,2^{-1/2}, \mu, c) dc -\i \int_0^{\sqrt{2}\pi} e^{-\sqrt{2}\i wt} C(\alpha,x,2^{-1/2}, \mu, \i t) dt.
\end{align*}
Here again the real integral is not absolutely convergent, so we introduce the Gaussian-regularized version
\begin{align*}
\tilde{C}_\ep(\alpha,x,2^{-1/2},\mu)&:= (1-e^{-2\pi \i w})\int_0^\infty e^{-\sqrt{2}wc-\ep^2 c^2} C(\alpha,x,2^{-1/2}, \mu, c) dc \\
&\qquad \qquad -\i \int_0^{\sqrt{2}\pi} e^{-\sqrt{2}\i wt} C(\alpha,x,2^{-1/2}, \mu, \i t) dt.
\end{align*}
The following theorem shows that, after renormalization, the Hankel-contour prescription differs from the real-line prescription by the expected multiplicative factor $1-e^{-2\pi \i w}$.
It is proved in  \textsection\ref{onethm2pf}.
\begin{thm}\label{onethm2}
Let $x = e_3$, and let $\alpha$ be any complex number with $\Re(\alpha)\in (-\frac{1}{\sqrt{2}}, 0]$. Let $w:= -1-\sqrt{2}\alpha$. Then %$\ep^{\sqrt{2}\alpha}\tilde{C}_\ep(\alpha,x,2^{-1/2},\mu)$ has the same limit as $\ep^{\sqrt{2}\alpha} (1-e^{-2\pi \i w})C_\ep(\alpha,x,2^{-1/2},\mu)$ as $\ep\to 0$.
\begin{align*}
&\lim_{\ep\to 0} \ep^{\sqrt{2}\alpha} \tilde{C}_\ep(\alpha, x, 2^{-1/2}, \mu) = (1-e^{-2\pi \i w})\lim_{\ep\to 0} \ep^{\sqrt{2}\alpha} C_\ep(\alpha, x, 2^{-1/2}, \mu). %\\
%&= \frac{ (1-e^{-2\pi \i w})(4\sqrt{2}\pi\mu)^we^{-\frac{1}{2}w(w+3)}G(w+2)\cos(\frac{\pi}{2}(w+1)) \Gamma(-w)\Gamma(\frac{1}{2}w+1)}{\sqrt{\pi}G(-w)}.
\end{align*}
\end{thm}

\subsection{The two-point function at $b=\frac{1}{\sqrt{2}}$}\label{twosec}
Now, we consider the two-point function at $b=\frac{1}{\sqrt{2}}$, with the two points being antipodal to each other. Under some restrictions on $\alpha_{1},\alpha_{2}$, we have the following result for the two-point function when the path integral is taken over all fields with a given zero mode and when $\Re(w)\in(-\tfrac12,0)$. The proof is in \textsection\ref{corthmpf}. The conformal-field-theoretic choice leading to $\Re(w)=0$ is treated separately below. 
\begin{thm}\label{corthm}
Let $b = \frac{1}{\sqrt{2}}$, $k=2$, $x_1 = -e_3$, and $x_2=e_3$. Assume that $\Re(\alpha_1), \Re(\alpha_2) >-\frac{1}{\sqrt{2}}$. Let $w := -1-\sqrt{2}(\alpha_1+\alpha_2)$ and assume that $\Re(w)\in (-\frac{1}{2},0)$. Then for any $\mu>0$ and $c\in \R$, we have
\[
C(\balpha, \bx, 2^{-1/2}, \mu,c) = \frac{e^{2\alpha_1\alpha_2}}{2\pi}\int_{-\infty}^\infty \Gamma(-w-\i y)f(w+\i y) (\mu e^{\sqrt{2}c})^{w+\i y} dy,
\]
where %$c\in (c_0,0)$ is arbitrary, 
\[
f(z) := \frac{(4\pi)^z e^{\frac{1}{2}z^2  -(\frac{3}{2}+w)z}\Gamma(z+1) G(z+\beta_1)G(z + \beta_2)}{(\Gamma(z-w))^zG(\beta_1)G(\beta_2)},
\]
with $\beta_j := 1+\sqrt{2}\alpha_j$ for $j=1,2$. %Moreover, if $c_0> -1$, we can take $c=c_0$ in the above integral.
\end{thm}
Next, we integrate over the zero mode following equation \eqref{cororiginal}. As before, we define the regularized two-point function 
\begin{align}\label{regcortwo}
C_\ep(\balpha, \bx, 2^{-1/2}, \mu) := \int_{-\infty}^\infty e^{-\sqrt{2}wc-\ep^2 c^2} C(\balpha, \bx, 2^{-1/2}, \mu,c) dc.
\end{align}
The following theorem shows that when $\Re(w)\in(-\tfrac12,0)$, the regularized two-point function
 diverges as $\epsilon\to 0$. This contrasts with the one-point function, where the unrenormalized quantity collapses to $0$ when $\Re(\alpha)<0$.  The proof is in \textsection\ref{infthmpf}.
\begin{thm}\label{infthm}
Under the assumptions of Theorem \ref{corthm}, we have  
\begin{align*}
&\lim_{\ep \to 0} \ep^{-w} C_\ep(\balpha, \bx, 2^{-1/2}, \mu) \\
&=   \frac{(4\sqrt{2}\pi\mu)^w e^{2\alpha_1\alpha_2-\frac{1}{2}w^2  -\frac{3}{2}w}\Gamma(-w)\Gamma(w+1) G(w+\beta_1)G(w + \beta_2)\cos(\frac{\pi w}{2})\Gamma(\frac{w+1}{2})}{\sqrt{2\pi } G(\beta_1)G(\beta_2)}.  %\frac{\i e^{2\alpha_1\alpha_2}(4\pi)^w e^{-\frac{1}{2}w^2  -\frac{3}{2}w}\cos(\frac{\pi w}{2})\Gamma(\frac{w+1}{2})\Gamma(w+1) G(-\sqrt{2}\alpha_1)G(-\sqrt{2}\alpha_2)}{\sqrt{2\pi }G(\beta_1)G(\beta_2)},
\end{align*}
and the right side is nonzero if $\alpha_1$ and $\alpha_2$ are both nonzero, and zero otherwise.
\end{thm}
Next, we deal with the case $\Re(w)=0$. This is the physically interesting case, for the following reason. Recall from equation \eqref{alphajform}, physically interesting choices of $\alpha_j$'s correspond to $\alpha_j$'s of the form $\frac{1}{2}Q + \i P_j$, where $P_j\in \R$. When $b= \frac{1}{\sqrt{2}}$, we have $Q = -\frac{1}{\sqrt{2}}$. Thus, for the two-point function, we take
\begin{align}\label{cftalpha}
\alpha_1 = -\frac{1}{2\sqrt{2}} + \i P_1, \ \ \ \alpha_2 = -\frac{1}{2\sqrt{2}} + \i P_2
\end{align}
for some $P_1,P_2\in \R$. Then 
\[
w = -1-\sqrt{2}(\alpha_1+\alpha_2) = -\sqrt{2}\i (P_1+P_2),
\]
which is purely imaginary. The next theorem gives the analogue of Theorem \ref{corthm} when $\Re(w)=0$. The proof is in \textsection\ref{zerocorthmpf}.
\begin{thm}\label{zerocorthm}
Let $b = \frac{1}{\sqrt{2}}$, $k=2$, $x_1 = -e_3$, and $x_2=e_3$. Assume that $\Re(\alpha_1), \Re(\alpha_2) >-\frac{1}{\sqrt{2}}$. Let $w := -1-\sqrt{2}(\alpha_1+\alpha_2)$, and assume that $\Re(w)=0$ and $\Im (w)\ne 0$. Then for any $\mu>0$, $c\in \R$, and $q\in (0,1)$, we have
\begin{align*}
C(\balpha, \bx, 2^{-1/2}, \mu,c) &= \frac{e^{2\alpha_1\alpha_2}}{2\pi }\int_{-\infty}^\infty \Gamma(-q-\i y)f(q+\i y) (\mu e^{\sqrt{2}c})^{q+\i y} dy + e^{2\alpha_{1}\alpha_2},
\end{align*}
where $f$ is as in Theorem \ref{corthm}.
\end{thm}
The next theorem identifies the limiting behavior of the regularized two-point correlation function~\eqref{regcortwo} as $\ep\to 0$, when $\alpha_1,\alpha_2$ are of the form~\eqref{cftalpha}. It turns out that the limit is a distribution, instead of a function. This theorem is proved in \textsection\ref{distthmpf}.
\begin{thm}\label{distthm}
Let $b = \frac{1}{\sqrt{2}}$, $k=2$, $x_1 = -e_3$, and $x_2=e_3$. Suppose that $\alpha_1,\alpha_2$ are of the form displayed in equation \eqref{cftalpha}. Let us write the regularized correlation function $C_\ep(\balpha, \bx, 2^{-1/2},\mu)$ defined in equation \eqref{regcortwo} simply as $C_\ep(P_1,P_2)$. Then 
\[
C(P_1,P_2) := \lim_{\ep\to 0}C_\ep(P_1,P_2)= \pi e^{\frac{1}{4}+2P_1^2}\delta(P_1+P_2),
\]
in the sense that for any smooth function $\varphi:\R^2\to \R$ with compact support,
\[
\lim_{\ep \to 0} \iint C_\ep(P_1,P_2) \varphi(P_1,P_2)dP_1dP_2 =  \pi \int_{-\infty}^\infty e^{\frac{1}{4}+ 2P_1^2}\varphi(P_1,-P_1)dP_1.
\]
\end{thm}
There is a recent proposal for the two-point function of timelike Liouville field theory when $\alpha_1,\alpha_2$ are of the form~\eqref{cftalpha}, by \citet*{collieretal24}. The proposal given in \cite[Equation (3.7)]{collieretal24} reads
\begin{align*}
C(P_1,P_2) = \frac{8\sqrt{2}\sin(2\pi b P_1)\sin(2\pi b^{-1}P_1)}{P_1^2}(\delta(P_1-P_2)+\delta(P_1+P_2)).
\end{align*}
This has some similarity with the result from Theorem \ref{distthm}, but is clearly not the same. In particular, their proposal is
invariant under the symmetries $P_{j}\mapsto\pm P_{j}$, whereas the formula from Theorem \ref{distthm}  is not. Understanding how (or whether) these prescriptions can be reconciled --- e.g., via a different choice of contour/regularization or an additional symmetry requirement --- remains an open question.

Let us now work out the two-point function with the Hankel contour. Proceeding as before, we arrive at the regularized two-point function
\begin{align*}
\tilde{C}_\ep(\balpha, \bx,2^{-1/2},\mu) &:= (1-e^{-2\pi \i w})\int_0^\infty e^{-\sqrt{2}wc-\ep^2c^2} C(\balpha, \bx, 2^{-1/2}, \mu, c) dc \\
&\qquad -\i \int_0^{\sqrt{2}\pi} e^{-\sqrt{2}\i wt} C(\balpha, \bx, 2^{-1/2}, \mu, \i t) dt.
\end{align*}
When $\alpha_1,\alpha_2$ are of the form~\eqref{cftalpha}, let us denote the above by $\tilde{C}_\ep(P_1,P_2)$. The following theorem, proved in \textsection\ref{distthm2pf}, shows that $\tilde{C}_\ep(P_1,P_2)$ converges to zero as $\ep\to 0$.
\begin{thm}\label{distthm2}
As $\ep\to 0$, $\tilde{C}_\ep(P_1,P_2)$ converges to zero in the sense of distributions.
\end{thm}

\subsection{The three-point function at $b=\alpha_2=\frac{1}{\sqrt{2}}$}
Let us now consider the three-point function, with the three points being $x_{1}=-e_{3}$, $x_{2}=e_{1}$,
and $x_{3}=e_{3}$. We treat here the resonant case $b=\alpha_{2}=\frac{1}{\sqrt{2}}$, in which the timelike DOZZ structure constant predicts a pole; our theorems below verify (for the path-integral definition used in this paper) that the Gaussian-regularized zero-mode integral
indeed diverges with an explicit leading coefficient.

For the three-point function, the standard prediction in the physics literature is called the `timelike DOZZ formula'.  The formula involves the special function $\Upsilon_{b}$ introduced by \citet{dornotto94}, defined as 
\begin{align*}
\Upsilon_b(z) := \exp\biggl\{\int_0^\infty \frac{1}{\tau}\biggl(\biggl(\frac{b}{2} + \frac{1}{2b} - z\biggr)^2 e^{-\tau} - \frac{\sinh^2((\frac{b}{2}+\frac{1}{2b} - z)\frac{\tau}{2})}{\sinh(\frac{b\tau}{2}) \sinh(\frac{\tau}{2b})}\biggr)d\tau \biggr\}
\end{align*}
on the strip $\{z\in  \C: 0< \Re(z)<b + \frac{1}{b}\}$ and continued analytically to the whole plane. 
Let $\Upsilon_b$ be as above, and let $\gamma(z) := \Gamma(z)/\Gamma(1-z)$. The timelike DOZZ formula says that the three-point function of timelike Liouville field theory, with $x_1,x_2,x_3$ as above, is given by
\begin{align*}
C(\balpha, \bx, b, \mu) &= e^{-\i\pi w} (-\pi \mu\gamma(-b^2))^w (4/e)^{1-1/b^2}b^{2b^2w + 2w} \frac{\Upsilon_b(bw+b)}{\Upsilon_b(b)}\\
&\qquad \cdot \frac{\Upsilon_b(\alpha_1-\alpha_2-\alpha_3+b)\Upsilon_b(\alpha_2-\alpha_1-\alpha_3+b)\Upsilon_b(\alpha_3-\alpha_1-\alpha_2+b)}{\Upsilon_b(b - 2\alpha_1)\Upsilon_b(b-2\alpha_2) \Upsilon_b(b-2\alpha_3)}.
\end{align*}
This claim was made originally by \citet{schomerus03}, \citet{zamolodchikov05}, and \citet{kostovpetkova06, kostovpetkova07, kostovpetkova07a}, using heuristic arguments based on the assumption that certain recursion relations of \citet{teschner95} for spacelike Liouville field theory continue to hold (after suitable modification) in the timelike theory. \citet*{harlowetal11} explored how the same expression can arise from the path integral and provided further evidence. A rigorous proof, when $w$ is a positive integer (charge neutrality) and $\Re(\alpha_j)>-\frac{1}{2b}$, was recently given in~\cite{chatterjee25}, building on a calculation by \citet{giribet12}. 

Of particular interest are the poles of the timelike DOZZ formula, which play a role in the notion of operator resonances introduced by  \citet{zamolodchikov91} in conformal perturbation theory. In that framework, resonances are reflected in poles of correlation functions as parameters cross special ``resonant'' values. 
It is known that the zeros of $\Upsilon_{b}$ occur precisely at $mb+nb^{-1}$ when either both $m,n$ are positive integers or both are nonpositive integers~\cite{harlowetal11}. Thus, the denominator in the timelike DOZZ formula vanishes precisely at those $\balpha$ which satisfy, for some $1\le j\le 3$,
\[
\alpha_j = \frac{(1-m)b}{2} - \frac{n}{2b},
\]
where either both $m,n$ are positive integers, or both $m,n$ are nonpositive integers.  The rigorous derivation in \cite{chatterjee25} applies under the constraints $b\in(0,1)$, $w\in\mathbb{N}$, and $\Re(\alpha_{j})>-\frac{1}{2b}$. In that regime one is forced into $-\frac{1}{2b}<\Re(\alpha_{j})<0$, and consequently the only DOZZ poles compatible with these bounds are those with $\alpha_{j}=-\frac{1}{2b}$  (corresponding to $m=n=1$).

We will now show how a different kind of pole can be accessed by the path integral when $b=\frac{1}{\sqrt{2}}$. This is the case $m=-1$, $n=0$, which produces poles at $\alpha_{j}=b=\frac{1}{\sqrt{2}}$.
Specifically, we will exhibit the existence of a pole when $b=\alpha_{2}=\frac{1}{\sqrt{2}}$.

Let $G$ be the Barnes $G$-function, as before. We will also use the Gauss hypergeometric function ${}_{2}F_{1}$ in the expression for the fixed-zero-mode correlation. The following theorem, proved in \textsection\ref{threecthmpf}, gives the value of the three-point correlation when $b=\alpha_{2}=\frac{1}{\sqrt{2}}$ and the path integral is taken over fields with a given zero mode.

\begin{thm}\label{threecthm}
Let  $k=3$, $x_1=-e_3$, $x_2 = e_1$, and $x_3=e_3$. Suppose that $\Re(\alpha_1),\Re(\alpha_3)>-\frac{1}{\sqrt{2}}$, and $b = \alpha_2 =\frac{1}{\sqrt{2}}$. Let $w := -1-\sqrt{2}(\alpha_1+\alpha_2+\alpha_3)$, and assume that $\Re(w)\in (-\frac{1}{2},0)$. Then for any $\mu>0$ and $c\in \R$,
\begin{align*}
&C(\balpha, \bx, 2^{-1/2},\mu,c) \\
&= \frac{e^{\sqrt{2}(\alpha_1+\alpha_3) + 2\alpha_1\alpha_3}2^{-\sqrt{2}(\alpha_1+\alpha_3) }}{2\pi}\int_{-\infty}^\infty \Gamma(-w-\i y)f(w+\i y) (\mu e^{\sqrt{2}c})^{w+\i y} dy,
\end{align*}
where $f$ is the analytic function
\begin{align*}%\label{fdef2}
f(z) &:=\frac{(2\pi)^z e^{\frac{1}{2}z(z-3-2w)} \Gamma(z+1)G(z+2+\sqrt{2}\alpha_1) G(z+2+\sqrt{2}\alpha_3)}{\Gamma(z-w)^zG(1+\sqrt{2}\alpha_1)G(1+\sqrt{2}\alpha_3)}\\
&\qquad \cdot \biggl\{\frac{{_2F_1}(1, z-w-1; 1+\sqrt{2}\alpha_1; \frac{1}{2})}{2\Gamma(1+\sqrt{2}\alpha_1) \Gamma(z+1+\sqrt{2}\alpha_3)} - \frac{\sqrt{2}\alpha_3\, {_2F_1}(1,z-w-1;z+2+\sqrt{2}\alpha_1;\frac{1}{2})}{2\Gamma(z+2+\sqrt{2}\alpha_1)\Gamma(1+\sqrt{2}\alpha_3)}\biggr\}
\end{align*}
defined on the domain $\Omega$ where $z-w\notin (-\infty,0]$, $z+1$ is not a nonpositive integer, and $z+2+\sqrt{2}\alpha_1$ is not a nonpositive integer. %, where ${_2F_1}$ denotes the Gauss hypergeometric function.
\end{thm}
 Next, we define the regularized three-point function 
\begin{align}\label{regcorthree}
C_\ep(\balpha, \bx, 2^{-1/2}, \mu) := \int_{-\infty}^\infty e^{-\sqrt{2}wc-\ep^2 c^2} C(\balpha, \bx, 2^{-1/2}, \mu,c) dc.
\end{align}
The next result, proved in \textsection\ref{threethmpf}, identifies the behavior of this function as~$\ep\to 0$. %, after scaling by $\ep^{-w}$. The proof  is in Subsection \ref{threethmpf}.
%shows that when $\Re(w)\in (-\frac{1}{2},0)$, the regularized two-point function blows up to infinity as $\ep\to 0$. Notice the contrast with the one-point function, which collapses to zero. The proof is in Subsection \ref{infthmpf}.

\begin{thm}\label{threethm}
Under the assumptions of Theorem \ref{threecthm}, we have
\begin{align*}
&\lim_{\ep \to 0} \ep^{-w} C_\ep(\balpha, \bx, 2^{-1/2}, \mu) \\
&=  e^{\sqrt{2}(\alpha_1+\alpha_3) + 2\alpha_1\alpha_3- \frac{1}{2}w(w+3)}2^{-\sqrt{2}(\alpha_1+\alpha_3)+\frac{1}{2}(w+1) }\\
&\qquad \cdot \frac{(2\pi \mu)^w\Gamma(-w)\Gamma(w+1)G(-\sqrt{2}\alpha_1) G(-\sqrt{2}\alpha_3)\cos(\frac{\pi w}{2})\Gamma(\frac{w+1}{2})}{2\sqrt{\pi}G(1+\sqrt{2}\alpha_1)G(1+\sqrt{2}\alpha_3)}\\
&\qquad \cdot \frac{1}{4\pi}\{(1+2\sqrt{2}\alpha_1)\sin(\sqrt{2}\pi \alpha_1) + (1+2\sqrt{2}\alpha_3)\sin(\sqrt{2}\pi\alpha_3)\}.
% \frac{(4\sqrt{2}\pi\mu)^w e^{2\alpha_1\alpha_2-\frac{1}{2}w^2  -\frac{3}{2}w}\Gamma(-w)\Gamma(w+1) G(w+\beta_1)G(w + \beta_2)\cos(\frac{\pi w}{2})\Gamma(\frac{w+1}{2})}{\sqrt{2\pi } G(\beta_1)G(\beta_2)}.  %\frac{\i e^{2\alpha_1\alpha_2}(4\pi)^w e^{-\frac{1}{2}w^2  -\frac{3}{2}w}\cos(\frac{\pi w}{2})\Gamma(\frac{w+1}{2})\Gamma(w+1) G(-\sqrt{2}\alpha_1)G(-\sqrt{2}\alpha_2)}{\sqrt{2\pi }G(\beta_1)G(\beta_2)},
\end{align*}
Moreover, the above limit is nonzero. % unless $\alpha_1=\alpha_3=-\frac{1}{2\sqrt{2}}$. 
\end{thm}
%This does not appear, at first glance, to match the timelike DOZZ formula for the three-point function. One can however manipulate the expression to arrive at a result which has some similarity with the timelike DOZZ formula, as follows. 

The above result shows that $C_\epsilon(\balpha,\bx,2^{-1/2},\mu)$ blows up as $\epsilon\to 0$, since $\Re(w)<0$ and the limit is nonzero. This matches the timelike DOZZ prediction of a pole at the resonant value $\alpha_{2}=b=\frac{1}{\sqrt{2}}$, provided the numerator in the DOZZ expression does not vanish at the same point. Under the assumptions of Theorem \ref{threecthm} one checks that the numerator terms stay away from the zero set of $\Upsilon_b$, hence the DOZZ structure constant indeed has a pole in this regime.

To wrap up the discussion, we verify that the pole is detected even if we integrate the zero mode along the Hankel contour proposed by \citet{usciatietal25}. Proceeding as before, we arrive at the regularized three-point function
\begin{align*}
\tilde{C}_\ep(\balpha, \bx,2^{-1/2},\mu) &:= (1-e^{-2\pi \i w})\int_0^\infty e^{-\sqrt{2}wc-\ep^2c^2} C(\balpha, \bx, 2^{-1/2}, \mu, c) dc \\
&\qquad -\i \int_0^{\sqrt{2}\pi} e^{-\sqrt{2}\i wt} C(\balpha, \bx, 2^{-1/2}, \mu, \i t) dt.
\end{align*}
The following theorem shows that the limiting behavior of $\tilde{C}_\ep$ as $\ep\to 0$ is the same as that of $C_\ep$, up to a multiplicative factor of $1-e^{-2\pi \i w}$. The proof is in \textsection\ref{distthm3pf}. Since $w$ is not an integer, this proves that $\tilde{C}_\ep$ also blows up as $\ep \to 0$. 
\begin{thm}\label{distthm3}
In the setting of Theorem \ref{threecthm}, we have
\[
\lim_{\ep\to 0} \ep^{-w} \tilde{C}_\ep(\balpha, \bx, 2^{-1/2}, \mu) = (1-e^{-2\pi \i w}) \lim_{\ep\to 0} \ep^{-w} C_\ep(\balpha, \bx, 2^{-1/2}, \mu).
\]
\end{thm}

%This completes the statement of results. Sketches of the proofs of Theorems \ref{zerocthm} and~\ref{zerothm} are given in \textsection\ref{sketchsec}. The main proofs are in \textsection\ref{proofsec}, and various auxiliary lemmas are collected together in the Appendix. 

\section{Proof sketch and main ideas}\label{sketchsec}
%\addtocontents{toc}{\protect\setcounter{tocdepth}{1}}
The following is a sketch of the proof of Theorem \ref{zerocthm}. This is the simplest result of this paper, but its proof contains many of the key ideas. Recall from equation \eqref{zcor} that 
\begin{align*}%\label{zcor}
C(2^{-1/2}, \mu,c) &= 1-4\pi\mu e^{\sqrt{2}c} + \sum_{n=2}^\infty \frac{(-\mu e^{\sqrt{2}c})^n}{n!}a_n,
\end{align*}
where 
\begin{align*}%\label{cndef}
a_n = \int_{(\S^2)^n}\exp\biggl( -2 \sum_{1\le i < j\le n} G_{\S^2}(y_i,y_j) \biggr) da(y_1)\cdots da(y_n).
\end{align*}
Shifting the integral to the complex plane by stereographic projection, we show that 
\begin{align*}%\label{anspecial}
a_n &= 4^n e^{\frac{1}{2} n(n-1)} \int_{\C^n}\prod_{i=1}^n (1+|z_i|^2)^{-n-1}\prod_{1\le i<j\le n}|z_i-z_j|^{2} d^2z_1\cdots d^2z_n.
\end{align*}
The choice $b = \frac{1}{\sqrt{2}}$ results in the Vandermonde term in the above integral. We evaluate the above integral using a familiar trick from random matrix theory. First, let us rewrite 
\begin{align}\label{anspecial0}
a_n &= 4^n e^{\frac{1}{2}n(n-1)} \int_{\C^n} \prod_{1\le i<j\le n}|z_i-z_j|^{2} d\nu(z_1)\cdots d\nu(z_n),
\end{align}
where $\nu$ denotes the measure on $\C$ that has density $(1+|z|^2)^{-n-1}$ with respect to Lebesgue measure. Now recall the Vandermonde determinant formula
\begin{align*}
\prod_{1\le i<j\le n} (z_j-z_i) &= \det 
\begin{pmatrix}
1 & z_1 & z_1^2 & \cdots & z_1^{n-1}\\
1 & z_2 & z_2^2 & \cdots & z_2^{n-1}\\
\vdots & \vdots & \vdots &\ddots & \vdots\\
1 & z_n & z_n^2 & \cdots & z_n^{n-1}
\end{pmatrix}\\
&= \sum_{\sigma \in S_n} \sign(\sigma)\prod_{i=1}^n z_i^{\sigma(i)-1}.
\end{align*}
This gives 
\begin{align*}
\prod_{1\le i<j\le n}|z_i-z_j|^{2} &= \prod_{1\le i<j\le n}(z_j-z_i) \prod_{1\le i<j\le n}(\overline{z}_j-\overline{z}_i)\\
&= \biggl(\sum_{\sigma \in S_n} \sign(\sigma)\prod_{i=1}^n z_i^{\sigma(i)-1}\biggr) \biggl(\sum_{\sigma \in S_n} \sign(\sigma)\prod_{i=1}^n \overline{z}_i^{\sigma(i)-1}\biggr)\\
&= \sum_{\sigma,\tau \in S_n} \sign(\sigma)\sign(\tau) \prod_{i=1}^n (z_i^{\sigma(i)-1}\overline{z}_i^{\tau(i)-1}).
\end{align*}
Plugging this into equation \eqref{anspecial0}, we get 
\begin{align*}
a_n &= 4^n e^{\frac{1}{2}n(n-1)}\sum_{\sigma,\tau \in S_n} \sign(\sigma)\sign(\tau) \int_{\C^n} \prod_{i=1}^n (z_i^{\sigma(i)-1}\overline{z}_i^{\tau(i)-1}) d\nu(z_1)\cdots d\nu(z_n)\\
&= 4^n e^{\frac{1}{2}n(n-1)}\sum_{\sigma,\tau \in S_n} \sign(\sigma)\sign(\tau) \prod_{i=1}^n \int_{\C} z^{\sigma(i)-1}\overline{z}^{\tau(i)-1} d\nu(z).
\end{align*}
Since $\nu$ is a radially symmetric measure, we have that for any $k,l\in \Z$,
\[
\int_{\C} z^k \overline{z}^l d\nu(z) =
\begin{cases}
\int_{\C} |z|^{2k} d\nu(z) &\text{ if } k=l,\\
0 &\text{ otherwise.}
\end{cases}
\]
Thus, only terms with $\sigma=\tau$ survive in the sum displayed above. We conclude that 
\begin{align*}%\label{an1}
a_n = 4^n e^{\frac{1}{2}n(n-1)} n! \prod_{j=0}^{n-1} \int_{\C} |z|^{2j} d\nu(z) = 4^n e^{\frac{1}{2}n(n-1)} n! \prod_{j=0}^{n-1} \int_{\C} \frac{|z|^{2j}}{(1+|z|^2)^{n+1}} d^2z. 
\end{align*}
Applying a simple change of variable and the formula for the Beta integral, we get the integral on the right equals
\begin{align*}%\label{an2}
\frac{\pi \Gamma(n-j)\Gamma(j+1)}{\Gamma(n+1)}.
\end{align*}
Thus, we get
\begin{align*}%\label{an3}
a_n&=  \frac{(4\pi)^n e^{\frac{1}{2}n(n-1)}}{\Gamma(n+1)^{n-1}} \biggl(\prod_{j=0}^{n-1}\Gamma(j+1)\biggr)^2.
\end{align*}
Then, using the identity 
\begin{align*}%\label{gn1id}
G(n+1) = \prod_{k=1}^{n-1} k!
\end{align*}
for the Barnes $G$-function, we get that for each $n\ge 2$, 
\begin{align*}%\label{an3}
a_n&=  \frac{(4\pi)^n e^{\frac{1}{2}n(n-1)}G(n+1)^2}{\Gamma(n+1)^{n-1}}.
\end{align*}
For $n=0,1$, this formula gives $a_0=1$ and $a_1 = 4\pi$. Thus, we arrive at the identity
\begin{align}\label{chalf}
C(2^{-1/2}, \mu, c) =\sum_{n=0}^\infty \frac{(-\mu e^{\sqrt{2}c})^n}{n!}f(n),
\end{align}
where $f:\C \setminus(-\infty,-1]\to \C$ is the analytic function 
\begin{align}\label{fzdefin}
f(z) = \frac{(4\pi)^z e^{\frac{1}{2}z(z-1)}G(z+1)^2}{\Gamma(z+1)^{z-1}}.
\end{align}
To express this as the integral displayed in Theorem \ref{zerocthm}, we proceed as follows, using a method inspired by the Barnes integral formula for the hypergeometric function. Consider the function 
\[
g(z) = \Gamma(-z) f(z)(\mu e^{\sqrt{2}c})^z. 
\]
On the domain $\Re(z)>-1$, the only poles of this function are at the nonnegative integers, arising due to the poles of the Gamma function. From this observation and some decay estimates (specifically, that $\Gamma(-z)f(z)$ decays rapidly to zero as $\Re(z)\to \infty$), we use Cauchy's theorem to get that for any $x\in (-1,0)$,
\begin{align*}
\frac{1}{2\pi \i}\int_{\Re(z)=x} g(z) dz &= -\sum_{n=0}^\infty \Res(g,n)\\
&= \sum_{n=0}^\infty \Res(\Gamma,-n)f(n)(\mu e^{\sqrt{2}c})^n\\
&= \sum_{n=0}^\infty \frac{(-1)^n}{n!}f(n)(\mu e^{\sqrt{2}c})^n.
\end{align*}
But by equation \eqref{chalf}, this is equal to $C (2^{-1/2},\mu,c)$. The final step involves taking $x\to -1$. This completes the sketch of the proof of Theorem \ref{zerocthm}.

To get Theorem \ref{zerothm} from Theorem \ref{zerocthm}, the heuristic argument goes as follows. By equation \eqref{chalf0} and Theorem \ref{zerocthm},
\begin{align*}
C (2^{-1/2},\mu) &= \int_{-\infty}^\infty e^{\sqrt{2}c} C (2^{-1/2},\mu,c) dc\\
&= \frac{1}{2\pi}\int_{-\infty}^\infty e^{\sqrt{2}c}\biggl(\int_{-\infty}^\infty \Gamma(1-\i y) f(-1+\i y) (\mu e^{\sqrt{2}c})^{-1+\i y} dy\biggr) dc\\
&= \frac{1}{2\pi}\int_{-\infty}^\infty \int_{-\infty}^\infty \Gamma(1-\i y) f(-1+\i y) \mu^{-1+\i y}e^{\sqrt{2}\i cy} dy dc\\
&= \frac{1}{2\sqrt{2}\pi}\int_{-\infty}^\infty \int_{-\infty}^\infty \Gamma(1-\i y) f(-1+\i y) \mu^{-1+\i y}e^{\i ty} dy dt.
\end{align*}
Suppose we are allowed to exchange the order of integration. Using the heuristic relation
\[
\frac{1}{2\pi}\int_{-\infty}^\infty e^{\i ty} dt = \delta(y),
\]
we get
\begin{align*}
C (2^{-1/2},\mu) &= \frac{1}{\sqrt{2}}\int_{-\infty}^\infty \Gamma(1-\i y) f(-1+\i y) \mu^{-1+\i y}\delta(y) dy\\
&= \frac{f(-1)}{\sqrt{2} \mu}.
\end{align*}
Note that $f(-1)$ is not well-defined via equation \eqref{fzdefin}. But an easy calculation shows that $\lim_{z\to -1}f(z)=\frac{1}{4\pi}$. This gives us Theorem \ref{zerothm}.

\section{Proofs}\label{proofsec}
\subsection{Zero-point function}
\subsubsection{Proof of Theorem \ref{zerocthm}}\label{zeroproofsec}
The first step is to prove the following lemma.
\begin{lmm}\label{zeroform1}
We have 
\[
C(2^{-1/2}, \mu, c) =\sum_{n=0}^\infty \frac{(-\mu e^{\sqrt{2}c})^n}{n!}f(n),
\]
%where $a_n = f(n)$ for each nonnegative integer $n$, 
where $f$ is the analytic function
\begin{align}\label{fdef}
f(z) := \frac{(4\pi)^z e^{\frac{1}{2}z(z-1)}G(z+1)^2}{\Gamma(z+1)^{z-1}},
\end{align}
defined on the domain $\C \setminus(-\infty,-1]$.
\end{lmm}
To prove Lemma \ref{zeroform1}, we need the following calculation. The proof is in \textsection\ref{intlmmpf}. 
\begin{lmm}\label{intlmm}
Let $\alpha$ and $\beta$ be two complex numbers such that $\Re(\alpha) > -2$ and $\Re(2\beta - \alpha)> 2$. Then 
\[
\int_{\C} \frac{|z|^\alpha}{(1+|z|^2)^\beta} d^2z = \frac{\pi \Gamma(\beta - \frac{1}{2}\alpha-1) \Gamma(\frac{1}{2}\alpha+1)}{\Gamma(\beta)},
\]
and the integral on the left is convergent.
\end{lmm}

We are now ready to prove Lemma \ref{zeroform1}.
\begin{proof}[Proof of Lemma \ref{zeroform1}]
Let $e_3 := (0,0,1)$. Let $\sigma:\S^2\setminus\{e_3\} \to\C$ be the stereographic projection
\begin{align}\label{stereodef}
\sigma(x,y,z) := \frac{x+\i y}{1-z}.
\end{align}
A simple calculation~(see \cite[Lemma 3.3.1]{chatterjee25})  shows that if $x,y\in \S^2\setminus\{e_3\}$ and $u,v$ are the stereographic projections of $x,y$ on $\C$, then 
\[
\|x-y\|^2 = \frac{4|u-v|^2}{(1+|u|^2)(1+|v|^2)}.
\]
By equation \eqref{gform}, this shows that for any $x,y\in \S^2\setminus\{e_3\}$, 
\begin{align}
G_{\S^2}(x,y) &= -\ln |\sigma(x)-\sigma(y)| + \frac{1}{2}\ln (1+|\sigma(x)|^2) \notag \\
&\qquad \qquad +\frac{1}{2}\ln (1+|\sigma(y)|^2) -\frac{1}{2}. \label{gpform}
\end{align}
Take any $y_1,\ldots, y_n\in \S^2\setminus\{e_3\}$. Then the above identity shows that
\begin{align}\label{gsum}
\sum_{1\le i<j\le n} G_{\S^2}(y_i, y_j) &= - \sum_{1\le i<j\le n} \ln |\sigma(y_i)-\sigma(y_j)|\notag \\
&\qquad  +\frac{1}{2}(n-1)\sum_{i=1}^n \ln (1+|\sigma(y_i)|^2) -\frac{1}{4}n(n-1). 
\end{align}
Now recall from equation \eqref{zcor} that 
\begin{align*}%\label{zcor}
C(2^{-1/2}, \mu,c) &= 1-4\pi\mu e^{\sqrt{2}c} +\sum_{n=2}^\infty \frac{(-\mu e^{\sqrt{2}c})^n}{n!}a_n,
\end{align*}
where 
\begin{align*}%\label{cndef}
a_n = \int_{(\S^2)^n}\exp\biggl( -2 \sum_{1\le i < j\le n} G_{\S^2}(y_i,y_j) \biggr) da(y_1)\cdots da(y_n).
\end{align*}
It is well known that for any $f:\S^2\to\C$ that is integrable with respect to the area measure on $\S^2$, we have
\begin{align}\label{change}
\int_{\S^2} f(y) da(y) = \int_{\C} f(\sigma^{-1}(z)) \frac{4}{(1+|z|^2)^2} d^2z,
\end{align}
where $d^2 z$ denotes Lebesgue measure on $\C$. 
Combining this with equation \eqref{gsum}, we get
\begin{align}\label{anspecial}
a_n &= 4^n e^{\frac{1}{2} n(n-1)} \int_{\C^n}\prod_{i=1}^n (1+|z_i|^2)^{-n-1}\prod_{1\le i<j\le n}|z_i-z_j|^{2} d^2z_1\cdots d^2z_n\notag \\
&= 4^n e^{\frac{1}{2}n(n-1)} \int_{\C^n} \prod_{1\le i<j\le n}|z_i-z_j|^{2} d\nu(z_1)\cdots d\nu(z_n),
\end{align}
where $\nu$ denotes the measure on $\C$ that has density $(1+|z|^2)^{-n-1}$ with respect to Lebesgue measure. 
Now take any $n\ge 2$, and recall the Vandermonde determinant formula
\begin{align*}
\prod_{1\le i<j\le n} (z_j-z_i) &= \det 
\begin{pmatrix}
1 & z_1 & z_1^2 & \cdots & z_1^{n-1}\\
1 & z_2 & z_2^2 & \cdots & z_2^{n-1}\\
\vdots & \vdots & \vdots &\ddots & \vdots\\
1 & z_n & z_n^2 & \cdots & z_n^{n-1}
\end{pmatrix}\\
&= \sum_{\sigma \in S_n} \sign(\sigma)\prod_{i=1}^n z_i^{\sigma(i)-1}.
\end{align*}
This gives 
\begin{align*}
\prod_{1\le i<j\le n}|z_i-z_j|^{2} &= \prod_{1\le i<j\le n}(z_j-z_i) \prod_{1\le i<j\le n}(\overline{z}_j-\overline{z}_i)\\
&= \biggl(\sum_{\sigma \in S_n} \sign(\sigma)\prod_{i=1}^n z_i^{\sigma(i)-1}\biggr) \biggl(\sum_{\sigma \in S_n} \sign(\sigma)\prod_{i=1}^n \overline{z}_i^{\sigma(i)-1}\biggr)\\
&= \sum_{\sigma,\tau \in S_n} \sign(\sigma)\sign(\tau) \prod_{i=1}^n (z_i^{\sigma(i)-1}\overline{z}_i^{\tau(i)-1}).
\end{align*}
Plugging this into equation \eqref{anspecial}, we get that for $n\ge 2$, 
\begin{align*}
a_n &= 4^n e^{\frac{1}{2}n(n-1)}\sum_{\sigma,\tau \in S_n} \sign(\sigma)\sign(\tau) \int_{\C^n} \prod_{i=1}^n (z_i^{\sigma(i)-1}\overline{z}_i^{\tau(i)-1}) d\nu(z_1)\cdots d\nu(z_n)\\
&= 4^n e^{\frac{1}{2}n(n-1)}\sum_{\sigma,\tau \in S_n} \sign(\sigma)\sign(\tau) \prod_{i=1}^n \int_{\C} z^{\sigma(i)-1}\overline{z}^{\tau(i)-1} d\nu(z).
\end{align*}
It is easy to see from the definition of $\nu$ that the integral displayed on the right is absolutely convergent. Now note that $\nu$ is a radially symmetric measure. This implies that for any $k,l\in \Z$,
\[
\int_{\C} z^k \overline{z}^l d\nu(z) =
\begin{cases}
\int_{\C} |z|^{2k} d\nu(z) &\text{ if } k=l,\\
0 &\text{ otherwise.}
\end{cases}
\]
Thus, only terms with $\sigma=\tau$ survive in the sum displayed above. We conclude that 
\begin{align}\label{an1}
a_n = 4^n e^{\frac{1}{2}n(n-1)} n! \prod_{j=0}^{n-1} \int_{\C} |z|^{2j} d\nu(z). 
\end{align}
Take any $n\ge 2$ and  $0\le j\le n-1$. By Lemma \ref{intlmm} (noting that $2j> -2$ and $2(n+1)-2j > 2$), we have
\begin{align*}%\label{an2}
\int_{\C} |z|^{2j} d\nu(z) &= \int_{\C} \frac{|z|^{2j}}{(1+|z|^2)^{n+1}} d^2z = \frac{\pi \Gamma(n-j)\Gamma(j+1)}{\Gamma(n+1)}.
\end{align*}
Plugging this into equation \eqref{an1}, we get that for each $n\ge 2$,
\begin{align}\label{an3}
a_n&=  \frac{(4\pi)^n e^{\frac{1}{2}n(n-1)}}{\Gamma(n+1)^{n-1}} \biggl(\prod_{j=0}^{n-1}\Gamma(j+1)\biggr)^2.
\end{align}
It is known~\cite[Example 12.48]{whittakerwatson21} is that the Barnes $G$-function satisfies,  for any positive integer $n$,
\begin{align}\label{gn1id}
G(n+1) = \prod_{k=1}^{n-1} k!.
\end{align}
Substituting this in equation \eqref{an3}, we get that for each $n\ge 2$, 
\begin{align*}%\label{an3}
a_n&=  \frac{(4\pi)^n e^{\frac{1}{2}n(n-1)}G(n+1)^2}{\Gamma(n+1)^{n-1}}.
\end{align*}
If we take $a_0=1$ and $a_1 = 4\pi$, then the above equation holds also for $n=0,1$. This completes the proof.
\end{proof}

We will now use Lemma \ref{zeroform1} to prove the formula stated in Theorem \ref{zerocthm}. For that, we first need to establish certain estimates for $f$. These are provided by the following lemma, which is proved in \textsection\ref{fzfinalpf}. Below and henceforth, $\log$ denotes the analytic branch of logarithm on $\C\setminus(-\infty,0]$ which is real-valued on $(0,\infty)$. Also, $\arg z$ denotes the argument of $z$ in $[-\pi,\pi)$.
\begin{lmm}\label{fzfinal}
For $z \in \C$ with $\Re(z)\ge -1$ and $z\ne -1$, we have
\begin{align*}
f(z) &= \exp\biggl\{\frac{1}{2}(z+1)\log (z+1)+ Az + R(z)\biggr\},
\end{align*}
where $A$ is a real universal constant, and $|R(z)|\le C\ln(2+|z|)$ for some positive universal constant $C$. Consequently, %for $z = x+\i y$, 
\begin{align*}
|f(z)| &\le \exp\biggl\{\frac{1}{2}(\Re(z)+1)\ln |z+1| - \frac{1}{2} \Im(z) \arg(z+1) + A\Re(z) + C\ln(2+|z|)\biggr\}.
\end{align*}
%for some other universal constant $C'$.
\end{lmm}
Another estimate that we need is the following upper bound on the modulus of the Gamma function. The proof is in \textsection\ref{gammalmmpf}.
\begin{lmm}\label{gammalmm}
Take any $z = x+\i y$ such that $x>-1$ and $z$ is not an integer. Let $n:= \lceil x+1\rceil$ and $a := n-x$, so that $1\le a<2$. Then
\[
|\Gamma(-z)|\le \frac{C_1\exp(C_2\ln(2+|y|)+ y \arg(a-\i y))}{\prod_{j=0}^{n-1}|j-z|},
\]
where $C_1, C_2$ are positive universal constants.
\end{lmm}
Now, for nonzero $x>-1$, consider the contour integral
\begin{align}\label{fcdef}
F(x) := \frac{1}{2\pi \i} \int_{\Re(z) = x} \Gamma(-z) f(z)(\mu e^{\sqrt{2}c})^z dz,
\end{align}
where the contour goes from $x-\i \infty$ to $x+\i \infty$.
\begin{lmm}\label{fconvlmm}
For any non-integer $x> -1$, the integral defining $F(x)$ in equation \eqref{fcdef} is absolutely convergent.
\end{lmm}
\begin{proof}
Since $\arg(x+\i y) \to \pm \frac{\pi}{2}$ as $y\to \pm \infty$, Lemma \ref{gammalmm} shows that $\Gamma(-x-\i y)$ decays exponentially in $|y|$ as $|y|\to \infty$. Similarly by Lemma \ref{fzfinal}, $|f(x+\i y)|$ also decays exponentially in $|y|$. Since $x$ is not an integer, these functions remain bounded near $y=0$. Lastly, $|(\mu e^{\sqrt{2}c})^z| = (\mu e^{\sqrt{2}c})^x$ remains bounded. This shows that the integral in equation \eqref{fcdef} is absolutely convergent.
\end{proof}

% Preamble:
% \usepackage{tikz}
% \usetikzlibrary{arrows.meta}
% Preamble:
% \usepackage{tikz}
% \usetikzlibrary{arrows.meta}

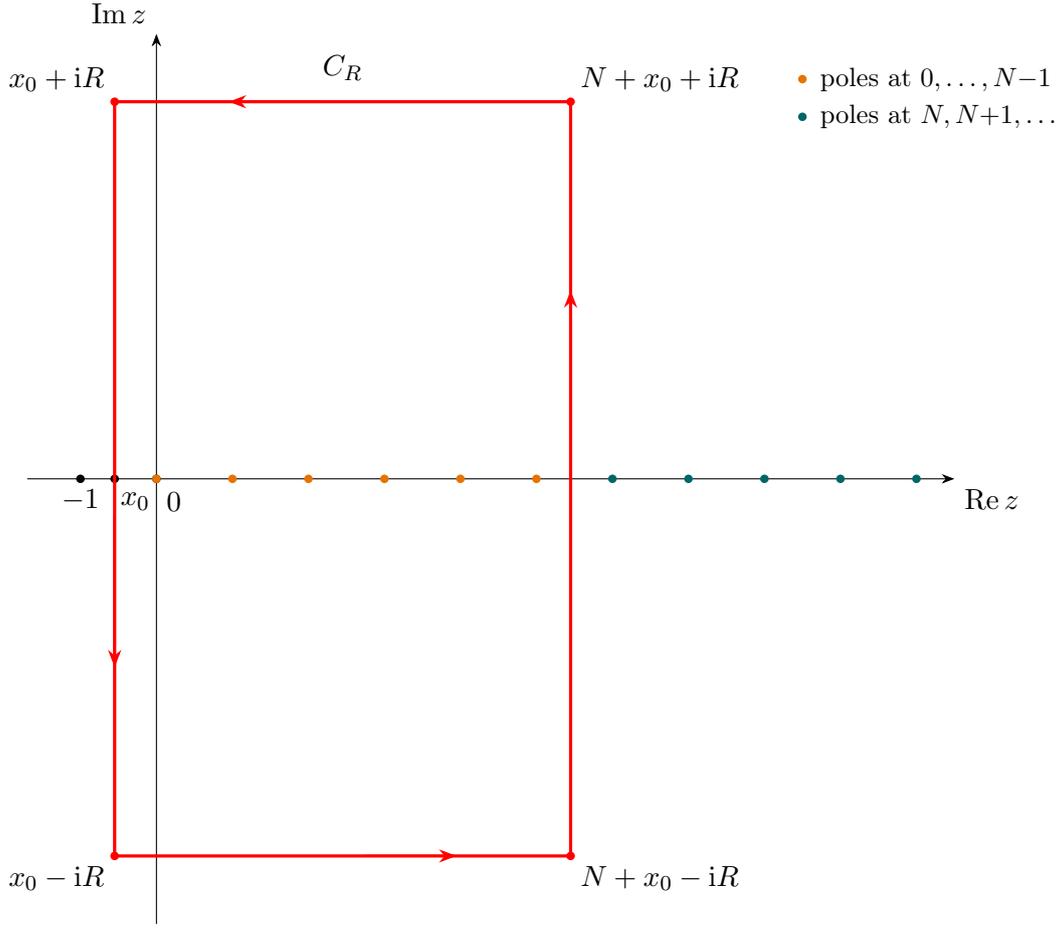
\begin{figure}[t]
\centering
\begin{tikzpicture}[>=Stealth, scale=1.0]

% ===== parameters =====
\def\Nint{6}      
\def\xo{-0.55}    
\def\R{5.0}       
\def\extra{5}     

% ===== extents =====
\pgfmathsetmacro{\Xmin}{-1.7}
\pgfmathsetmacro{\Xmax}{\Nint+\extra+1.5}
\pgfmathsetmacro{\Ymin}{-\R-0.9}
\pgfmathsetmacro{\Ymax}{ \R+0.9}

% ===== axes =====
\draw[black,->] (\Xmin,0) -- (\Xmax-2,0) node[below right] {$\Re z$};
\draw[black,->] (0,\Ymin) -- (0,\Ymax) node[above left] {$\Im z$};

% ===== dots for -1, 0, x0 =====
\fill[black] (-1,0) circle (1.6pt) node[below] {$-1$};
\fill[black] (0,0) circle (1.6pt);
\node[black, below right] at (0,-0.05) {$0$};

\fill[black] (\xo,0) circle (1.6pt);
\node[black, below left] at (\xo+.6,-0.05) {$x_0$};

% ===== poles of g =====
% enclosed: 0,...,N-1
\foreach \m in {0,...,\numexpr\Nint-1\relax} {
  \fill[orange!90!black] (\m,0) circle (1.6pt);
}
% not enclosed: N, N+1,...
\foreach \j in {0,...,\numexpr\extra-1\relax} {
  \pgfmathsetmacro{\p}{\Nint+\j}
  \fill[teal!80!black] (\p,0) circle (1.6pt);
}

% ===== legend (top right) =====
\fill[orange!90!black] (\Xmax-4,\Ymax-0.6) circle (1.6pt);
\node[black, anchor=west] at (\Xmax-3.9,\Ymax-0.6) {\small poles at $0,\dots,N{-}1$};
\fill[teal!80!black] (\Xmax-4,\Ymax-1.1) circle (1.6pt);
\node[black, anchor=west] at (\Xmax-3.9,\Ymax-1.1) {\small poles at $N,N{+}1,\dots$};

% ===== contour C_R =====
\coordinate (A) at (\xo, \R);           
\coordinate (B) at (\xo+\Nint, \R);     
\coordinate (C) at (\xo+\Nint,-\R);     
\coordinate (D) at (\xo,-\R);           

% draw CCW: A -> B -> C -> D -> A is clockwise,
% so reverse: A -> D -> C -> B -> A
\draw[very thick, red] (A) -- (D) -- (C) -- (B) -- cycle;

% corner dots
\fill[red] (A) circle (1.6pt);
\fill[red] (B) circle (1.6pt);
\fill[red] (C) circle (1.6pt);
\fill[red] (D) circle (1.6pt);

% direction arrows (CCW)
\draw[red,->, thick] ($(A)!0.55!(D)$) -- ($(A)!0.75!(D)$); % down
\draw[red,->, thick] ($(D)!0.55!(C)$) -- ($(D)!0.75!(C)$); % right
\draw[red,->, thick] ($(C)!0.55!(B)$) -- ($(C)!0.75!(B)$); % up
\draw[red,->, thick] ($(B)!0.55!(A)$) -- ($(B)!0.75!(A)$); % left

% vertex labels
\node[black, above left]  at (A) {$x_0+\i R$};
\node[black, above right] at (B) {$N+x_0+\i R$};
\node[black, below right] at (C) {$N+x_0-\i R$};
\node[black, below left]  at (D) {$x_0-\i R$};

% contour label
\node[black] at (\xo+\Nint/2,\R+0.45) {$C_R$};

\end{tikzpicture}
\caption{Rectangular contour $C_R$ with vertices $x_0\pm \i R$ and $N+x_0\pm \i R$, traversed counter-clockwise. Poles of $g$ at $0,\dots,N{-}1$ lie inside $C_R$.}
\label{fig:CR-with-poles}
\end{figure}

The next lemma is the key step of the proof.

\begin{lmm}\label{zerofinal0}
For any $x_0\in (-1,0)$ and $N\ge 1$, 
\begin{align*}%\label{contoureq}
F(x_0) = F(N+x_0) + \sum_{n=0}^{N-1}\frac{(-1)^n}{n!}f(n) (\mu e^{\sqrt{2}c})^n.
\end{align*}
\end{lmm}
\begin{proof}
Define the function 
\[
g(z) := \Gamma(-z) f(z) (\mu e^{\sqrt{2}c})^{z}
\]
on the domain $(\C \setminus(-\infty,-1])\setminus\{0,1,2,\ldots\}$. Fix some $x_0\in (-1,0)$. 
Take any $N\ge 1$ and $R\ge 1$. Let $C_R$ be the rectangular contour with vertices $x_0\pm \i R$ and $N + x_0\pm \i R$, traversed counter-clockwise (see Figure \ref{fig:CR-with-poles}). Since $f$ is analytic in $\C\setminus(-\infty,-1]$, we deduce that the only poles of $g$ are at the nonnegative integers, arising due to the poles of the Gamma function. Since
\[
\Res(\Gamma, -n) = \frac{(-1)^n}{n!}
\]
for each nonnegative integer $n$, we get that 
\[
\Res(g, n) = -\frac{(-1)^n}{n!}f(n) (\mu e^{\sqrt{2}c})^n. % = \frac{(-1)^n}{n!}a_n t^n.
\]
Since $C_R$ encloses the poles at $0,1,\ldots,N-1$, Cauchy's theorem gives
\begin{align}\label{cr1}
\frac{1}{2\pi \i} \oint_{C_R} g(z) dz = -\sum_{n=0}^{N-1}\frac{(-1)^n}{n!}f(n) (\mu e^{\sqrt{2}c})^n,
\end{align}
Now, by Lemma \ref{fzfinal} and Lemma \ref{gammalmm}, for any $x\ge x_0$, 
\begin{align}\label{gxir}
|g(x + \i R)| &\le \frac{C_1\exp(C_2\log(2+R)+ R \arg(\lceil x+1 \rceil -x- \i R))}{\prod_{j=0}^{\lceil x+1\rceil-1 }|j-x+ \i R|}\notag \\
&\qquad \cdot \exp\biggl(\frac{1}{2}(x+1)\log |x+ \i R+1| - \frac{1}{2} R \arg(x+ \i R+1) \notag \\
&\qquad \qquad \qquad + Ax + C\log(2+|x+\i R|)\biggr) e^{\sqrt{2}cx},
\end{align}
where  $A, C, C_1, C_2$ are universal constants. Clearly, the bound is decreasing exponentially in $R$ as $R \to \infty$, uniformly over $x$ in any given bounded range. Thus,
\begin{align}\label{cr2}
\lim_{R\to \infty} \max_{x_0\le x\le N+x_0} |g(x+ \i R)| = 0.
\end{align}
Similarly, 
\begin{align}\label{cr3}
\lim_{R\to \infty} \max_{x_0\le x\le N+x_0} |g(x-\i R)| = 0.
\end{align}
By equations \eqref{cr2} and \eqref{cr3}, we can take $R\to \infty$ in equation \eqref{cr1}. This completes the proof.
\end{proof}

The next lemma shows that $N$ can be taken to infinity in Lemma \ref{zerofinal0}, resulting in the removal of the $F(N+x_0)$ terms.
\begin{lmm}\label{zerofinal}
For any $x_0\in (-1,0)$, $\lim_{N\to \infty} F(N+x_0) = 0$. Thus, by Lemma \ref{zerofinal0}, 
\[
F(x_0) = \sum_{n=0}^\infty \frac{(-\mu e^{\sqrt{2}c})^n}{n!}f(n) = C(2^{-1/2}, \mu,c).
\]
\end{lmm}
The proof of this lemma is in \textsection\ref{zerofinalpf}. We are now finally ready to complete the proof of Theorem \ref{zerocthm}. We will use $C, C_1, C_2,\ldots$ to denote arbitrary positive universal constants, whose values may change from line to line. 

Explicitly, Lemma \ref{zerofinal} says that for any $x_0\in (-1,0)$, we have
\begin{align*}
C(2^{-1/2}, \mu,c) &= \frac{1}{2\pi}\int_{-\infty}^\infty \Gamma(-x_0-\i y)f(x_0+\i y)(\mu e^{\sqrt{2}c})^{x_0+\i y}dy.
\end{align*}
Consequently,
\begin{align}\label{cmuc}
C(2^{-1/2}, \mu,c) = \frac{1}{2\pi}\lim_{x_0\downarrow -1} \int_{-\infty}^\infty \Gamma(-x_0-\i y)f(x_0+\i y)(\mu e^{\sqrt{2}c})^{x_0+\i y}dy.
\end{align}
Our goal is to show that
\begin{align*}%\label{zeroalt}
C(2^{-1/2}, \mu,c) = \frac{1}{2\pi} \int_{-\infty}^\infty \Gamma(1-\i y)f(-1+\i y)(\mu e^{\sqrt{2}c})^{-1+\i y}dy.
\end{align*}
Thus, all we have to show is that the limit in equation \eqref{cmuc} can be moved inside the integral. To show this, take any $x_0\in (-1,-\frac{1}{2})$ and any $y\in \R$. Let $z := x_0+\i y$. Recall that by Lemma \ref{fzfinal},
\begin{align}\label{fx0iy}
|f(x_0+\i y)| &\le \exp\biggl(\frac{1}{2}(x_0+1)\ln |z+1| - \frac{1}{2} y \arg(z+1) + Ax_0 + C\ln(2+|z|)\biggr),
\end{align}
where $C$ and $A$ are universal constants. Since $-1<x_0<-\frac{1}{2} $, we have
\[
y \arg(z+1) = |y|\arg(x_0+\i |y|+1) \ge |y|\arg(1+\i |y|),
\]
and 
\[
(x_0+1)\ln|z+1| = \frac{1}{2}(x_0+1)\ln((x_0+1)^2+ y^2) \le C\ln(2+|y|).
\]
%since $x_0 < -\frac{1}{2}$, we have $|z+1|\ge \frac{1}{2}$, and hence,
%\[
%(x_0+1)\ln |z+1| = (x_0+1)\ln(2|z+1|) -(x_0+1)\ln 2\le \frac{1}{2}\ln(2|z+1|).
%\]
Using these inequalities in equation \eqref{fx0iy}, we get the $x_0$-free bound
\begin{align}\label{dom1}
|f(x_0+\i y)| \le C_1 (2+|y|)^{C_2}e^{-\frac{1}{2} |y|\arg(1+\i |y|)}.
\end{align}
Next, by Lemma \ref{gammalmm}, 
\begin{align*}
|\Gamma(-x_0-\i y)| &\le  \frac{C_1\exp(C_2\ln(2+|y|)+ y \arg(1-x_0-\i y))}{|x_0+\i y|}.
\end{align*}
Since $x_0\in (-1,-\frac{1}{2})$, we have $|x_0+\i y |\ge \frac{1}{2}$, and 
\begin{align*}
y \arg(1-x_0-\i y) &= -|y|\arg(1-x_0 + \i |y|)\le -|y|\arg(1+\i |y|).
\end{align*}
Thus, again we have an $x_0$-free bound
\begin{align}\label{dom2}
|\Gamma(-x_0-\i y)| \le C_1 (2+|y|)^{C_2} e^{-|y|\arg(1+\i |y|)},
\end{align}
where $C_1,C_2$ are universal constants. It is now easy to see that the bounds \eqref{dom1} and~\eqref{dom2} allow us to apply the dominated convergence theorem to move the limit inside the integral in equation \eqref{cmuc}. This completes the proof of Theorem \ref{zerocthm}.

\subsubsection{Proof of Theorem \ref{zerothm}}\label{zerothmpf}
By Theorem \ref{zerocthm}, we have
\begin{align*}
C_\ep(2^{-1/2}, \mu) &= \frac{1}{2\pi}\int_{-\infty}^\infty \int_{-\infty}^\infty \Gamma(1-\i y)f(-1+\i y)\mu^{-1+\i y}e^{\sqrt{2}c \i y- \ep^2 c^2}dy dc.
\end{align*}
Since the bounds \eqref{dom1} and \eqref{dom2} have no dependence on $x_0$, they hold even if we take $x_0=-1$. Using these bounds, we can  conclude that the above double integral is absolutely convergent. Consequently, by Fubini's theorem, we can interchange the order of integration. Since
\[
 \int_{-\infty}^\infty e^{\sqrt{2}c \i y- \ep^2 c^2}dc = \frac{\sqrt{\pi}}{\ep}\exp\biggl(-\frac{y^2}{2\ep^2}\biggr),
\]
this gives
\begin{align*}
C_\ep(2^{-1/2}, \mu) &= \frac{1}{2\sqrt{\pi}\ep}\int_{-\infty}^\infty \Gamma(1-\i y)f(-1+\i y)\mu^{-1+\i y}\exp\biggl(-\frac{y^2}{2\ep^2}\biggr)dy\\
&= \frac{1}{2\sqrt{\pi}}\int_{-\infty}^\infty \Gamma(1-\i \ep u)f(-1+\i \ep u)\mu^{-1+\i \ep u}e^{-\frac{1}{2}u^2}du.
\end{align*}
The bounds \eqref{dom1} and \eqref{dom2} show that $|\Gamma(1-\i \ep u)f(-1+\i \ep u)\mu^{-1+\i \ep u}|$ is uniformly bounded over $u\in \R$. This allows us to apply the dominated convergence theorem and conclude that 
\begin{align*}
\lim_{\ep\to 0} C_\ep(2^{-1/2}, \mu) &= \frac{1}{2\sqrt{\pi}}\Gamma(1)f(-1)\mu^{-1}\int_{-\infty}^\infty e^{-\frac{1}{2}u^2}du\\
&=  \frac{1}{\sqrt{2}\mu}f(-1),
\end{align*}
where, with the aid of equation \eqref{gzid}, 
\begin{align*}
f(-1) &:= \lim_{\substack{z\to -1, \\ z\in \C\setminus(-\infty,-1]}}f(z) \\
&= \lim_{\substack{z\to -1, \\ z\in \C\setminus(-\infty,-1]}}\frac{(4\pi)^z e^{\frac{1}{2}z(z-1)}G(z+2)^2}{\Gamma(z+1)^{z+1}} = \frac{e}{4\pi}.
\end{align*}
This completes the proof.

\subsubsection{Proof of Theorem \ref{zerohankel}}\label{zerohankelpf}
For $b = \frac{1}{\sqrt{2}}$, we formally have 
\begin{align}\label{ctildenew}
\tilde{C}(2^{-1/2},\mu) &= \int_{\mc} e^{\sqrt{2}c} C(2^{-1/2}, \mu, c) dc \notag \\
&= \int_0^\infty e^{\sqrt{2}c} C(2^{-1/2}, \mu, c) dc  - \int_0^\infty e^{\sqrt{2}(c+\sqrt{2}\pi \i)} C(2^{-1/2}, \mu, c+\sqrt{2}\pi \i) dc \notag \\
&\qquad -\i \int_0^{\sqrt{2}\pi} e^{\sqrt{2}\i t} C(2^{-1/2}, \mu, \i t) dt.
\end{align}
Note that by the formula \eqref{zcor}, we have that for any $c\in \R$,
\begin{align*}
C(2^{-1/2}, \mu, c+\sqrt{2}\pi\i) &= \sum_{n=0}^\infty \frac{(-\mu e^{\sqrt{2}(c+\sqrt{2}\pi \i)})^n}{n!} a_n\\
&= \sum_{n=0}^\infty \frac{(-\mu e^{\sqrt{2}c})^n}{n!} a_n = C(2^{-1/2}, \mu, c),
\end{align*}
Also, $e^{\sqrt{2}(c+\sqrt{2}\pi \i)} = e^{\sqrt{2}c}$. Thus, the first two integrals on the right in equation \eqref{ctildenew} cancel out each other, leaving us with the simplified formula 
\begin{align*}
\tilde{C}(2^{-1/2},\mu) &= -\i \int_0^{\sqrt{2}\pi} e^{\sqrt{2}\i t} C(2^{-1/2}, \mu, \i t) dt.
\end{align*}
But again, recalling the formula \eqref{zcor} and using the bound on $a_n$ provided by the inequality~\eqref{wagner2} to justify exchanging the order of integration and summation, we get
\begin{align*}
\int_0^{\sqrt{2}\pi} e^{\sqrt{2}\i t} C(2^{-1/2}, \mu, \i t) dt &= \sum_{n=0}^\infty \frac{(-\mu)^n}{n!} a_n\int_0^{\sqrt{2}\pi} e^{\sqrt{2}\i t(n+1)} dt=0.
\end{align*}
This completes the proof.

\subsection{One-point function}
\subsubsection{Proof of Theorem \ref{onecthm}}\label{onecthmpf}
Let us first express $C(\balpha,\bx,b,\mu,c)$ for general $k$, via integrals over the complex plane instead of the sphere, as we did for $C(b,\mu,c)$ in \textsection\ref{zeroproofsec}. Let $u_j := \sigma(x_j)$ for $j=1,\ldots,k$, where $\sigma$ is the stereographic projection defined in equation~\eqref{stereodef}. Let us assume, for now, that none of the $x_j$'s are the north pole of $\S^2$, so that the $u_j$'s are all finite. By the formula~\eqref{gpform}, we have
\begin{align*}
\sum_{j=1}^k \sum_{l=1}^n \alpha_jG_{\S^2}(x_j,y_l) &=  -\sum_{j=1}^k \sum_{l=1}^n \alpha_j \ln |u_j-\sigma(y_l)| + \frac{n}{2}\sum_{j=1}^k\alpha_j\ln(1+|u_j|^2)\\
&\qquad  + \frac{1}{2}\biggl(\sum_{j=1}^k\alpha_j\biggr)\sum_{l=1}^n \ln (1+|\sigma(y_l)|^2) -\frac{n}{2}\sum_{j=1}^k \alpha_j.
\end{align*}
Let $a_n$ be as in equation \eqref{andef}. Then by the above equation and equations \eqref{change} and  \eqref{gsum}, we get
\begin{align*}
a_n &= \int_{\C^n}\exp\biggl(4b\sum_{j=1}^k \sum_{l=1}^n \alpha_j \ln |u_j-z_l| -2bn\sum_{j=1}^k\alpha_j\ln(1+|u_j|^2)\\
&\qquad  - 2b\biggl(\sum_{j=1}^k\alpha_j\biggr)\sum_{l=1}^n \ln (1+|z_l|^2) +2bn\sum_{j=1}^k \alpha_j + 4b^2\sum_{1\le i<j\le n} \ln |z_i-z_j|\notag \\
&\qquad  -2b^2(n-1)\sum_{l=1}^n \ln (1+|z_l|^2) +b^2n(n-1)\biggr)\prod_{l=1}^n \frac{4}{(1+|z_l|^2)^2}d^2z_1\cdots d^2z_n.
\end{align*}
Letting $w := (Q- \sum_{j=1}^k \alpha_j)/b$, note that the cumulative power of $1+|z_l|^2$ in the above integrand is
\begin{align*}
- 2b\biggl(\sum_{j=1}^k\alpha_j\biggr) - 2b^2(n-1)-2 &= - 2b(Q-bw) - 2b^2(n-1)-2\\
&= -2b^2(n-w). %-2b^2 +  2+2b^2 w-2b^2n + 2b^2 -2
\end{align*}
Similarly, the constant coefficient is 
\begin{align*}
 4^n\exp\biggl(2bn\sum_{j=1}^k \alpha_j + b^2n(n-1)\biggr) &= 4^n\exp\biggl(2bn(Q-bw)+ b^2n(n-1)\biggr)\\
 &= 4^ne^{b^2 n(n+1-2w)-2n}. %\exp\biggl(2b^2n- 2n -2b^2nw+ b^2n^2- b^2n\biggr)
\end{align*}
%and proceeding as in Subsection \ref{zeroproofsec}, this gives
Combining, we get 
\begin{align}\label{ancor}
a_n &= 4^n e^{b^2 n(n+1-2w)-2n} \biggl(\prod_{j=1}^k (1+|u_j|^2)^{-2bn\alpha_j}\biggr)\int_{\C^n}\prod_{i=1}^n (1+|z_i|^2)^{-2b^2(n-w)} \notag\\
&\hskip 1in \cdot \prod_{i=1}^n \prod_{j=1}^k |z_i -u_j|^{4b\alpha_j} \prod_{1\le i<j\le n}|z_i-z_j|^{4b^2} d^2z_1\cdots d^2z_n.
\end{align}
%\[
%c_n := \int_{(\S^2)^n}\exp\biggl(-4b\sum_{j=1}^k \sum_{l=1}^n \alpha_j G(x_j, y_l)-4b^2\sum_{1\le l < l'\le n}G(y_l, y_{l'})\biggr)  da(y_1)\cdots da(y_n).
%\]
Again by equation \eqref{gpform},
\begin{align*}
\prod_{1\le j<j'\le k} e^{-4\alpha_j\alpha_{j'}  G_{\S^2}(x_j,x_{j'})} &= e^{2\sum_{1\le j<j'\le k} \alpha_j\alpha_{j'}}\prod_{j=1}^k (1+|u_j|^2)^{-2\alpha_j \sum_{j'\ne j}\alpha_{j'} }\\
&\qquad \cdot\prod_{1\le j<j'\le k} |u_j -u_{j'}|^{4\alpha_j\alpha_{j'}}. %\\
%&=  e^{2(Q-bw)^2 - 2\sum_{j=1}^k \alpha_j^2}
%&= e^{2\sum_{1\le j<j'\le k} \alpha_j\alpha_{j'}}\prod_{j=1}^k (1+|w_j|^2)^{-2\alpha_j (Q-bw - \alpha_j)}\\
%&\qquad \cdot\prod_{1\le j<j'\le k} |w_j -w_{j'}|^{4\alpha_j\alpha_{j'}}
\end{align*}
Using these computations in equation \eqref{zcor}, we get 
\begin{align}\label{calphabasic}
C(\balpha, \bx, b,\mu,c) &= e^{2\sum_{1\le j<j'\le k} \alpha_j\alpha_{j'}}\prod_{j=1}^k (1+|u_j|^2)^{-2\alpha_j \sum_{j'\ne j}\alpha_{j'} }\notag \\
&\qquad \cdot \prod_{1\le j<j'\le k} |u_j -u_{j'}|^{4\alpha_j\alpha_{j'}}  \biggl\{1+ \sum_{n=1}^\infty \frac{(-\mu e^{2bc})^n}{n!} a_n \biggr\},
\end{align}
where $a_n$ is as in equation \eqref{ancor}.

We now specialize to the case $k=1$. Then $\balpha$ is just a single coefficient $\alpha$, and $\bx$ is a single point $x$, $u=\sigma(x)$, and $w=(Q-\alpha)/b$. To have $u$ finite, we take $x\ne e_3=(0,0,1)$. We get
\begin{align*}
C(\alpha, x, b,\mu,c) &= 1+ \sum_{n=1}^\infty \frac{(-\mu e^{2bc})^n}{n!} a_n,
\end{align*}
where
\begin{align}\label{ank1}
a_n &= 4^n e^{b^2 n(n+1-2w)-2n} (1+|u|^2)^{-2bn\alpha}\int_{\C^n}\prod_{i=1}^n (1+|z_i|^2)^{-2b^2(n-w)} \notag\\
&\hskip 1in \cdot \prod_{i=1}^n |z_i -u|^{4b\alpha} \prod_{1\le i<j\le n}|z_i-z_j|^{4b^2} d^2z_1\cdots d^2z_n.
\end{align}
We will now take $x$ approaching $e_3$, which is equivalent to taking $|u|\to\infty$ in the above formula. 
\begin{lmm}\label{zeroform1one}
Let $x=e_3$. Suppose that $b\in (0,1)$ and $\Re(\alpha) >-\frac{1}{2b}$. Then
\begin{align*}
C(\alpha, x, b, \mu,c) &= 1+ \sum_{n=1}^\infty \frac{(-\mu e^{2bc})^n}{n!} a_n,
\end{align*}
where
\begin{align}\label{ank10}
a_n &= 4^n e^{b^2 n(n+1-2w)-2n} \int_{\C^n}\prod_{i=1}^n (1+|z_i|^2)^{-2b^2(n-w)} \prod_{1\le i<j\le n}|z_i-z_j|^{4b^2} d^2z_1\cdots d^2z_n.
\end{align}
\end{lmm}
\begin{proof}
First, take $x\ne e_3$. Let $a_n(u)$ denote the formula displayed in equation \eqref{ank1}, where $u = \sigma(x)$. By Lemma \ref{convergence}, $a_n$ is a continuous function of $x$. Thus, we have to show that as $|u|\to \infty$, $a_n(u)$ approaches the formula displayed in equation~\eqref{ank10}. To do that, take any $M>0$, and define $a_n(u,M)$ using the same expression as in equation~\eqref{ank1}, but restricting the integration to the region $(\Omega_M)^n$ where $\Omega_M:=\{z\in \C: |z|\le M\}$. Then it follows easily by the dominated convergence theorem that for any $M$, $\lim_{|u|\to \infty} a_n(u, M)$ is given by the formula in equation \eqref{ank10}, but with the domain of integration restricted to $(\Omega_M)^n$. Then, by the monotone convergence theorem, we conclude that the double limit
\[
\lim_{M\to \infty} \lim_{|u|\to \infty} a_n(u,M)
\]
equals the right side of equation \eqref{ank10}. Thus, to complete the proof, we need to show that
\begin{align}\label{an2show}
\lim_{M\to \infty} \lim_{|u|\to \infty} a_n(u,M) =\lim_{|u|\to \infty}    a_n(u).
\end{align}
Now, converting the integration back to the sphere, it is easy to see that 
\begin{align*}
a_n(u,M) &= \int_{(\S^2_M)^n}\exp\biggl(-4b\sum_{l=1}^n \alpha G(x, y_l) \\
&\qquad \qquad -4b^2 \sum_{1\le l < l'\le n} G(y_l,y_{l'}) \biggr) da(y_1)\cdots da(y_n),
\end{align*}
where $\S^2_M$ is the set of all $y\in \S^2$ such that $|\sigma(y)|\le M$. Thus, by equation \eqref{wagner}, 
\begin{align*}
|a_n(u)-a_n(u,M)| &= \int_{(\S^2)^n \setminus (\S^2_M)^n}\exp\biggl(-4b\sum_{l=1}^n \alpha G(x, y_l) \\
&\qquad \qquad -4b^2 \sum_{1\le l < l'\le n} G(y_l,y_{l'}) \biggr) da(y_1)\cdots da(y_n)\\
&\le C_n\sum_{i=1}^n\int_{y_i\notin \S^2_M}\prod_{l=1}^n\|y_l - x\|^{4b\alpha} da(y_1)\cdots da(y_n)\\
&= C_n n \biggl(\int_{\S^2 \setminus \S^2_M}\|y - x\|^{4b\alpha} da(y)\biggr)\biggl(\int_{\S^2}\|y - x\|^{4b\alpha} da(y)\biggr)^{n-1},
\end{align*}
where $C_n$ does not depend on $M$ or $x$. Since $\Re(\alpha) > -\frac{1}{2b}$, it follows that the supremum of the right side over all $x$ in the upper hemisphere is bounded by a number $\epsilon(M)$ that tends to zero as $M\to\infty$. Consequently,
\[
\limsup_{|u|\to \infty}|a_n(u,M)-a_n(u)|\le \epsilon(M).
\]
From this, it is easy to deduce the claim \eqref{an2show}.
\end{proof}

Specializing further to $b = \frac{1}{\sqrt{2}}$, we get $w = -1-\sqrt{2}\alpha$, and for $x=e_3$, 
\begin{align*}
C(\alpha, x, 2^{-1/2},\mu,c) &= 1+ \sum_{n=1}^\infty \frac{(-\mu e^{\sqrt{2}c})^n}{n!} a_n,
\end{align*}
where
\begin{align}\label{an1form}
a_n &= 4^n e^{\frac{1}{2} n(n-3-2w)} \int_{\C^n}\prod_{i=1}^n (1+|z_i|^2)^{-(n-w)} \prod_{1\le i<j\le n}|z_i-z_j|^{2} d^2z_1\cdots d^2z_n.
\end{align}
We will now prove the following analogue of Lemma \ref{zeroform1}. Since the proof is similar to that of Lemma \ref{zeroform1}, it is relegated to \textsection\ref{oneform1pf}.
\begin{lmm}\label{oneform1}
Let $x=e_3$. Then for $\Re(\alpha)> -\frac{1}{\sqrt{2}}$, We have 
\[
C(\alpha, x, 2^{-1/2}, \mu, c) =\sum_{n=0}^\infty \frac{(-\mu e^{\sqrt{2}c})^n}{n!}f(n),
\]
%where $a_n = f(n)$ for each nonnegative integer $n$, 
where $f$ is the analytic function
\begin{align*}%\label{fdef}
f(z) := \frac{(4\pi)^z e^{\frac{1}{2}z(z-3-2w)}\Gamma(z+1)G(z+1)G(z-w)}{\Gamma(z-w)^{z}G(-w)}%\frac{(4\pi)^z e^{\frac{1}{2}z(z-1)}G(z+1)^2}{\Gamma(z+1)^{z-1}},
\end{align*}
defined on the domain $(w + (\C \setminus(-\infty,0]))\setminus\{-1,-2,\ldots\}$.
\end{lmm}
Henceforth we will work under the assumption that $\Re(\alpha)\in (-\frac{1}{\sqrt{2}},0]$, which is equivalent to saying that $\Re(w)\in [-1,0)$. First, we have the following analogue of Lemma \ref{fzfinal}. The proof is in \textsection\ref{fzfinalonepf}.
\begin{lmm}\label{fzfinalone}
Let $f$ be defined as in Lemma \ref{oneform1} and assume that $\Re(\alpha)\in (-\frac{1}{\sqrt{2}}, 0]$. Then for $z \in \C$ with $\Re(z)\ge \Re(w)$, we have
\begin{align*}
f(z) &= \exp\biggl\{(z+1)\log (z-w)-\frac{1}{2}z\log(z-w+1)+ Az + R(z)\biggr\},
\end{align*}
where $A$ is a real universal constant, and $|R(z)|\le C$ for some constant $C$ that may depend only on $w$. Consequently, %for $z = x+\i y$, 
\begin{align*}
|f(z)| &\le \exp\biggl\{\Re(z+1)\ln |z-w| -\Im(z+1)\arg(z-w) \\
&\qquad \qquad -\frac{1}{2}\Re(z)\ln|z-w+1|+ \frac{1}{2} \Im(z) \arg(z-w+1) + A\Re(z) + C\biggr\}.
\end{align*}
%for some other universal constant $C'$.
\end{lmm}
Next, for nonzero $x>\Re(w)$, define the contour integral
\begin{align}\label{fcdefone}
F(x) := \frac{1}{2\pi \i} \int_{\Re(z) = x} \Gamma(-z) f(z)(\mu e^{\sqrt{2}c})^z dz,
\end{align}
where the contour goes from $x-\i \infty$ to $x+\i \infty$.
\begin{lmm}\label{fconvlmmone}
Assume that $\Re(\alpha)\in (-\frac{1}{\sqrt{2}}, 0]$. Then for any non-integer $x> \Re(w)$, the integral defining $F(x)$ in equation~\eqref{fcdefone} is absolutely convergent.
\end{lmm}
\begin{proof}
Since $\arg(x+\i y) \to \pm \frac{\pi}{2}$ as $y\to \pm \infty$, Lemma \ref{gammalmm} shows that $\Gamma(-x-\i y)$ decays exponentially in $|y|$ as $|y|\to \infty$. Similarly by Lemma \ref{fzfinalone}, $|f(x+\i y)|$ also decays exponentially in $|y|$ as $|y|\to \infty$. Since $x$ is not an integer, these functions remain bounded near $y=0$. Lastly, $|(\mu e^{\sqrt{2}c})^z| = (\mu e^{\sqrt{2}c})^x$ remains bounded. This shows that the integral in equation \eqref{fcdefone} is absolutely convergent.
\end{proof}
The next lemma is the analogue of Lemma \ref{zerofinal0}.
\begin{lmm}\label{zerofinal0one}
Assume that $\Re(\alpha)\in (-\frac{1}{\sqrt{2}}, 0]$. Then for any $x_0\in (\Re(w),0)$ and $N\ge 1$, 
\begin{align*}%\label{contoureq}
F(x_0) = F(N+x_0) + \sum_{n=0}^{N-1}\frac{(-1)^n}{n!}f(n) (\mu e^{\sqrt{2}c})^n.
\end{align*}
\end{lmm}
Since the proof is very similar to the proof of Lemma \ref{zerofinal0}, it is relegated to \textsection\ref{zerofinal0onepf}. Next, we have the analogue of Lemma~\ref{zerofinal}. The proof is in \textsection\ref{zerofinalonepf}.
\begin{lmm}\label{zerofinalone}
%Suppose that $\Re(w)\in (-1,0)$. 
Assume that $\Re(\alpha)\in (-\frac{1}{\sqrt{2}}, 0]$. Then for any $x_0\in (\Re(w),0)$, 
\[
\lim_{N\to \infty} F(N+x_0) = 0.
\]
Consequently, by Lemma \ref{zerofinal0one},
\[
F(x_0) =  \sum_{n=0}^\infty \frac{(-\mu e^{\sqrt{2}c})^n}{n!}f(n) = C(\alpha, x, 2^{-1/2}, \mu,c).
\]
\end{lmm}
We are now ready to complete the proof of Theorem \ref{onecthm}. We will use $C, C_1, C_2,\ldots$ to denote arbitrary positive constants that may depend only on $w$, whose values may change from line to line. Let $a$ and $r$ denote the real and imaginary parts of $w$. Explicitly, Lemma~\ref{zerofinalone} says that for any $x_0\in (a,0)$, we have
\begin{align*}
C(\alpha, x, 2^{-1/2}, \mu,c) &= \frac{1}{2\pi}\int_{-\infty}^\infty \Gamma(-x_0-\i y)f(x_0+\i y)(\mu e^{\sqrt{2}c})^{x_0+\i y}dy.
\end{align*}
Consequently,
\begin{align}\label{cmucone}
C(\alpha, x, 2^{-1/2}, \mu,c) = \frac{1}{2\pi}\lim_{x_0\downarrow a} \int_{-\infty}^\infty \Gamma(-x_0-\i y)f(x_0+\i y)(\mu e^{\sqrt{2}c})^{x_0+\i y}dy.
\end{align}
On the other, Theorem \ref{onecthm} claims that 
\begin{align}\label{zeroaltone}
C(\alpha, x, 2^{-1/2}, \mu,c) &= \frac{1}{2\pi} \int_{-\infty}^\infty \Gamma(-w-\i y)f(w+\i y)(\mu e^{\sqrt{2}c})^{w+\i y}dy\notag \\
&= \frac{1}{2\pi} \int_{-\infty}^\infty \Gamma(-a-\i y)f(a+\i y)(\mu e^{\sqrt{2}c})^{a+\i y}dy.
\end{align}
Thus, all we have to show is that the limit in equation \eqref{cmucone} can be moved inside the integral. To show this, take any $x_0\in (a,\frac{1}{2}a)$ and any $y\in \R$. Let $z := x_0+\i y$. Recall that by Lemma \ref{fzfinalone},
\begin{align}\label{fx0iyone}
|f(x_0+\i y)| &\le \exp\biggl((x_0+1)\ln |z-w| -y\arg(z-w) \notag \\
&\qquad \qquad -\frac{1}{2}x_0\ln|z-w+1|+ \frac{1}{2} y \arg(z-w+1) + Ax_0 + C\biggr), 
%\exp\biggl(\frac{1}{2}(x_0+a+2)\ln |z-w| - \frac{1}{2} (y+r) \arg(z-w)\notag\\
%&\qquad \qquad  + Ax_0 + C\ln(2+|z|)\biggr),
\end{align}
where $A$ is a universal constant and $C$ may depend only on $w$. Since $x_0> a$, we have
\begin{align*}
y\arg(z-w) &\ge (y-r)\arg(z-w) - C\\
&= |y-r|\arg(x_0-a+\i |y-r|)-C,%\\
%&\ge |y-r|\arg(x_0-a+\i |y-r|+1)-C
%&\ge |y-r|\arg(1+\i |y-r|)-C,
\end{align*}
and similarly,
\begin{align*}
\frac{1}{2}y\arg(z-w+1) &\le \frac{1}{2}(y-r)\arg(z-w+1) + C\\
&= \frac{1}{2}|y-r|\arg(x_0 - a+ \i |y-r|+1)+C\\
&\le \frac{1}{2}|y-r|\arg(x_0 - a+ \i |y-r|)+C. 
\end{align*}
Since $x_0-a\le 1$, we can combine the above inequalities to get
\begin{align*}
-y\arg(z-w) + \frac{1}{2}y\arg(z-w+1) &\le -\frac{1}{2}|y-r|\arg(x_0-a+\i |y-r|) + C\\
&\le  -\frac{1}{2}|y-r|\arg(1+\i |y-r|) + C.
\end{align*}
Next, note that since $x_0+1$ and $-\frac{1}{2}x_0$ are nonnegative and $|x_0-a+1|>|x_0-a|$, 
\begin{align*}
&(x_0+1)\ln |z-w|  -\frac{1}{2}x_0\ln|z-w+1| \\
&\le (x_0+1)\ln|x_0 -a +1 + \i (y-r)|- \frac{1}{2}x_0 \ln|x_0-a +1 + \i (y-r)|\\
& = \frac{1}{2}(x_0+2)\ln|x_0-a+1+\i(y-r)|\le \ln(2+|y-r|).
%(x_0+a+2)\ln |z-w| &\le (x_0+a+2)\ln (C_1+C_2|y|)\\
%&\le C_3\ln(C_1+C_2|y|).  %\frac{1}{2}(x_0+\Re(w)+2) \ln ((x_0-\Re(w))^2 + (y-\Im(w))^2)
\end{align*}
Using these inequalities in equation \eqref{fx0iyone}, we get
\begin{align}\label{dom1one}
|f(x_0+\i y)| \le C_1 (2+|y-r|)e^{-\frac{1}{2} |y-r|\arg(1+\i |y-r|)}.
\end{align}
It is now easy to see that the $x_0$-free bounds \eqref{dom1one} and~\eqref{dom2} allow us to apply the dominated convergence theorem to move the limit inside the integral in equation \eqref{cmucone}. This completes the proof of Theorem \ref{onecthm}.

\subsubsection{Proof of Theorem \ref{onethm}}\label{onethmpf}
Let $a$ and $r$ denote the real and imaginary parts of $w$. By Theorem \ref{onecthm}, we have
\begin{align*}
&C_\ep(\alpha, x, 2^{-1/2}, \mu) \\
&= \frac{1}{2\pi}\int_{-\infty}^\infty \int_{-\infty}^\infty \Gamma(-a-\i (y+r))f(a+\i (y+r))\mu^{a+\i (y+r)}e^{\sqrt{2}c \i y- \ep^2 c^2}dy dc\\
&=  \frac{1}{2\pi}\int_{-\infty}^\infty \int_{-\infty}^\infty \Gamma(-a-\i y)f(a+\i y)\mu^{a+\i y}e^{\sqrt{2}c \i (y-r)- \ep^2 c^2}dy dc.
\end{align*}
Since the bounds \eqref{dom1one} and \eqref{dom2} have no dependence on $x_0$, we can use them to conclude that they hold also for $x_0=a$, which implies that the above double integral is absolutely convergent. Consequently, by Fubini's theorem, we can interchange the order of integration. Since
\[
 \int_{-\infty}^\infty e^{\sqrt{2}c \i (y-r)- \ep^2 c^2}dc = \frac{\sqrt{\pi}}{\ep}\exp\biggl(-\frac{(y-r)^2}{2\ep^2}\biggr),
\]
this gives
\begin{align*}
&C_\ep(\alpha, x, 2^{-1/2}, \mu) = \frac{1}{2\sqrt{\pi}\ep}\int_{-\infty}^\infty \Gamma(-a-\i y)f(a+\i y)\mu^{a+\i y}\exp\biggl(-\frac{(y-r)^2}{2\ep^2}\biggr)dy\\
&= \frac{1}{2\sqrt{\pi}}\int_{-\infty}^\infty \Gamma(-a-\i (r+\ep u))f(a+\i(r+ \ep u))\mu^{a+\i (r+\ep u)}e^{-\frac{1}{2}u^2}du\\
&= \frac{1}{2\sqrt{\pi}}\int_{-\infty}^\infty \Gamma(-w-\i\ep u)f(w+\i\ep u)\mu^{w+\i \ep u}e^{-\frac{1}{2}u^2}du.
\end{align*}
By the alternative expression for $f$ given in equation \eqref{fznewalt}, we get
\begin{align*}
f(w+\i\ep u) &= \frac{(4\pi)^{w+\i \ep u} e^{\frac{1}{2}(w+\i \ep u) (-3-w+\i \ep u)+(1+w+\i \ep u)\log(\i \ep u)}G(w+\i \ep u+3)G(\i \ep u +2)}{\Gamma(\i \ep u+1)^{w+\i \ep u+2}\Gamma(w+\i\ep u+2)G(-w)}.
\end{align*}
Now note that 
\[
\log(\i \ep u) = \ln \ep + \ln |u| + \frac{\i \pi}{2}\sign(u) = \ln \ep+\log(\i u). % = \ln \ep + \log(\i u). 
\]
This shows that for any $u\in \R$,
\begin{align*}
&\lim_{\ep\to 0} \ep^{-1-w} f(w+\i \ep u) \\
&= \lim_{\ep \to 0} \frac{(4\pi)^{w+\i \ep u} e^{\frac{1}{2}(w+\i \ep u) (-3-w+\i \ep u)+(1+w+\i \ep u)\log(\i u)-\i \ep u\ln \ep}G(w+\i \ep u+3)G(\i \ep u +2)}{\Gamma(\i \ep u+1)^{w+\i \ep u+2}\Gamma(w+\i\ep u+2)G(-w)}\\
&= \frac{(4\pi)^{w} e^{\frac{1}{2}w (-3-w)+(1+w)\log(\i u)}G(w+3)}{\Gamma(w+2)G(-w)}.
\end{align*}
Moreover, by equations  \eqref{dom1one} and \eqref{dom2}, $|\Gamma(-w-\i \ep u) f(w+\i \ep u)\mu^{w+\i \ep u}|$ is bounded by a constant that does not depend on $u$. Thus, by the dominated convergence theorem, we have
\begin{align*}
&\lim_{\ep\to 0} \ep^{-1-w} C_\ep(\alpha, x, 2^{-1/2}, \mu) \\
&= \frac{ (4\pi\mu )^w\Gamma(-w)e^{-\frac{1}{2}w(w+3)}G(w+3)}{2\sqrt{\pi} \Gamma(w+2)G(-w)} \int_{-\infty}^\infty|u|^{1+w}e^{\frac{\i \pi}{2}(1+w)\sign(u)}e^{-\frac{1}{2}u^2}du.
\end{align*}
Finally, note that 
\begin{align*}
&\frac{1}{2}\int_{-\infty}^\infty|u|^{1+w}e^{\frac{\i \pi}{2}(1+w)\sign(u)}e^{-\frac{1}{2}u^2}du\\
&= \cos \biggl(\frac{\pi}{2}(1+w)\biggr) \int_0^\infty u^{1+w}e^{-\frac{1}{2}u^2}du\\
&= \cos \biggl(\frac{\pi}{2}(1+w)\biggr) \int_0^\infty (\sqrt{2t})^w e^{-t} dt = \cos \biggl(\frac{\pi}{2}(1+w)\biggr)2^{\frac{1}{2}w} \Gamma\biggl(\frac{1}{2}w+1\biggr).
\end{align*}
This completes the proof.

\subsubsection{Proof of Theorem \ref{onethm2}}\label{onethm2pf}
From Lemma \ref{zeroform1one}, we have
\begin{align}\label{repeatone}
C(\alpha, x, 2^{-1/2}, \mu, c) =\sum_{n=0}^\infty \frac{(-\mu e^{\sqrt{2}c})^n}{n!}a_n,
\end{align}
where $a_n$ are the coefficients defined in equation \eqref{ank10}. Although Lemma \ref{zeroform1one} is stated for $c\in \R$, it is easy to see from the proof that the formula continues to be valid for any complex $c$. By the inequality \eqref{wagner2}, we can integrate $c$ from $0$ to $\sqrt{2}\pi \i$ by moving the integral inside the infinite sum, to get
\begin{align}\label{horizontalone}
 -\i \int_0^{\sqrt{2}\pi} e^{-\sqrt{2}\i wt} C(\alpha, x, 2^{-1/2}, \mu, \i t) dt &= -\i \sum_{n=0}^\infty \frac{(-\mu )^n}{n!}a_n\int_0^{\sqrt{2}\pi} e^{\sqrt{2}\i (n-w)t} dt\notag\\
&= \sum_{n=0}^\infty \frac{(-\mu )^n(1-e^{2\pi \i (n-w)})}{\sqrt{2} n! (n-w)}a_n\notag \\
&= (1-e^{-2\pi \i w}) \sum_{n=0}^\infty \frac{(-\mu )^n}{\sqrt{2} n! (n-w)}a_n.
\end{align}
Next, note that
\begin{align*}
&(1-e^{-2\pi \i w})\int_0^\infty e^{-\sqrt{2}wc - \ep^2 c^2} C(\alpha, x, 2^{-1/2},\mu ,c) dc\\
&= (1-e^{-2\pi\i w}) \int_{-\infty}^\infty e^{-\sqrt{2}w c - \ep^2 c^2} C(\alpha, x, 2^{-1/2},\mu ,c) dc \\
&\qquad \qquad - (1-e^{-2\pi \i w}) \int_{-\infty}^0 e^{-\sqrt{2}w c - \ep^2 c^2} C(\alpha, x, 2^{-1/2},\mu ,c) dc.
\end{align*}
Theorem \ref{onethm} tells us the behavior of the first term as $\ep \to 0$. We will now show that the difference between the second term and the quantity displayed in equation \eqref{horizontalone} is $o(\ep^{-\sqrt{2}\Re(\alpha)})$ as $\ep\to 0$. Clearly, this will complete the proof.  Since $\Re(\alpha)>-\frac{1}{\sqrt{2}}$, it suffices to show that the difference on $O(\ep)$. We will actually show that the difference is $O(\ep^2)$. 

To this end, first note that by equation \eqref{repeatone} and an application of the bound \eqref{wagner2} (to move the integration inside the infinite sum), we get
\begin{align}\label{horizontalone2}
&(1-e^{-2\pi \i w}) \int_{-\infty}^0 e^{-\sqrt{2}w c - \ep^2 c^2} C(\alpha, x, 2^{-1/2},\mu ,c) dc\notag\\
&= (1-e^{-2\pi \i w}) \sum_{n=0}^\infty \frac{(-\mu )^n}{n!}a_n\int_{-\infty}^0 e^{\sqrt{2}c(n-w) - \ep^2 c^2 } dc.
\end{align}
Now, for $n\ge 1$, 
\begin{align*}
&\biggl|\int_{-\infty}^0 e^{\sqrt{2}c(n-w) - \ep^2 c^2 }  dc - \int_{-\infty}^0 e^{\sqrt{2}c(n-w)} dc\biggr| \\
&\le \int_{-\infty}^0 |e^{\sqrt{2}c(n-w)} (1- e^{-\ep^2 c^2}) |  dc\\
&= \int_{-\infty}^0 e^{\sqrt{2}c(n-\Re(w))} |1- e^{-\ep^2 c^2} |  dc\le \ep^2  \int_{-\infty}^0 c^2e^{\sqrt{2}cn} dc= \frac{2\ep^2}{(\sqrt{2} n)^3}. 
\end{align*}
This shows that the difference between the quantities displayed in equations \eqref{horizontalone} and~\eqref{horizontalone2} is equal to $J_\ep(w)+O(\ep^2)$, where 
\begin{align*}
J_\ep(w) := (1-e^{-2\pi \i w}) \biggl\{ -\frac{1}{\sqrt{2}w} - \int_{-\infty}^0 e^{-\sqrt{2}w c - \ep^2 c^2 } dc\biggr\}.
\end{align*}
Since $\Re(w) < 0$, we get 
\begin{align*}
J_\ep(w) &= (1-e^{-2\pi \i w}) \biggl\{\int_{-\infty}^0e^{-\sqrt{2}wc} dc - \int_{-\infty}^0 e^{-\sqrt{2}w c - \ep^2 c^2 } dc\biggr\}\\
&= (1-e^{-2\pi \i w}) \int_{-\infty}^0e^{-\sqrt{2}wc} (1-e^{-\ep^2 c^2}) dc.
\end{align*}
By the inequality $1-e^{-x}\le x$, this shows that 
\begin{align*}
|J_\ep(w)|&\le 2\ep^2 \int_{-\infty}^0 c^2 e^{-\sqrt{2}\Re(w)c} dc = O(\ep^2).
\end{align*}
%Since $\Re(\alpha) \in (-\frac{1}{\sqrt{2}},0]$, this shows that $\ep^{\sqrt{2}\alpha} J_\ep(w) \to 0$ as $\ep \to 0$. 
This completes the proof.

\subsection{Two-point function}
\subsubsection{Proof of Theorem \ref{corthm}}\label{corthmpf}
First, let us prove the analogue of Lemma \ref{zeroform1one} for the two-point function.
\begin{lmm}\label{zeroform1two}
Let  $k=2$, $x_1=-e_3$, and $x_2=e_3$. Suppose that $b\in (0,1)$ and $\Re(\alpha_j) >-\frac{1}{2b}$ for $j=1,2$. Then for any $\mu>0$ and $c\in \R$,
\begin{align*}
C(\balpha, \bx, b, \mu, c)  &= e^{2\alpha_1\alpha_2}\biggl\{1+ \sum_{n=1}^\infty \frac{(-\mu e^{2bc})^n }{n!} a_n \biggr\},
\end{align*}
where
\begin{align}\label{cn2}
a_n &= 4^n e^{b^2 n(n+1-2w)-2n} \int_{\C^n}\prod_{i=1}^n \frac{|z_i|^{4b\alpha_1}}{(1+|z_i|^2)^{2b^2(n-w)}}  \prod_{1\le i<j\le n}|z_i-z_j|^{4b^2} d^2z_1\cdots d^2z_n.
\end{align}
\end{lmm}
Since the proof is similar to the proof of Lemma \ref{zeroform1one}, it is shifted to \textsection\ref{zeroform1twopf}. Specializing the above lemma to the case $b=\frac{1}{\sqrt{2}}$, and proceeding along the same steps as before, we get  the following analogue of Lemma \ref{zeroform1} and Lemma \ref{oneform1}. % The proof is in Subsection \ref{twoform1pf} of the Appendix.

\begin{lmm}\label{twoform1}
Let  $k=2$, $x_1=-e_3$, and $x_2=e_3$. Suppose that $\Re(\alpha_1),\Re(\alpha_2)>-\frac{1}{\sqrt{2}}$. Let $w := -1-\sqrt{2}(\alpha_1+\alpha_2)$ and suppose that $w\ne 0$. Let $\beta_j := 1+\sqrt{2}\alpha_j$ for $j=1,2$. Then for any $\mu>0$ and $c\in \R$,  
\[
C(\balpha, \bx, 2^{-1/2}, \mu, c) =e^{2\alpha_1\alpha_2}\sum_{n=0}^\infty \frac{(-\mu e^{\sqrt{2}c})^n}{n!}f(n),
\]
where $f$ is the analytic function
\begin{align*}%\label{fdef2}
f(z) := \frac{ (4\pi)^z e^{\frac{1}{2} z(z-3-2w)}\Gamma(z+1) G(z+\beta_1)G(z + \beta_2)}{(\Gamma(z-w))^zG(\beta_1)G(\beta_2)},
\end{align*}
defined on the domain $\Omega:= (w+(\C\setminus(-\infty, 0]))\setminus\{-1,-2,\ldots\}$.
\end{lmm}
Since the proof is similar to the proofs of Lemma \ref{zeroform1} and Lemma \ref{oneform1}, it is relegated to \textsection\ref{twoform1pf}. Next, we have the following analogue of Lemma \ref{fzfinal} and Lemma \ref{fzfinalone}. The proof is in \textsection\ref{fzfinal2pf}.
\begin{lmm}\label{fzfinal2}
%Assume that $\Re(\alpha_1),\Re(\alpha_2)>-\frac{1}{\sqrt{2}}$. 
Assume the setting of Lemma \ref{twoform1}. 
Then for any $z\in \Omega$ with $\Re(z) > -\frac{1}{2}$, we have
\begin{align*}
f(z) &= (z+\beta_1)(z+\beta_2)\exp\biggl(z\log(z-w)+z\log(z-w+1) \\
&\qquad \qquad -\frac{3}{2}z\log (z-w+2)+ Bz + Q(z)\biggr),
\end{align*}
where $B$ is a real universal constant, and $|Q(z)|\le C\log(2+|z|)$ for some constant $C$ that depends only on $\alpha_1,\alpha_2$. Moreover, if $\Re(z)>-\frac{1}{2}$ and $\Re(z) \ge  \Re(w)+\delta$ for some $\delta >0$, then 
\begin{align*}
|f(z)| &\le |(z+\beta_1)(z+\beta_2)|\exp\biggl(\frac{1}{2}\Re(z)\ln |z-w| - \frac{1}{2} \Im(z) \arg(z-w) \\
&\qquad \qquad + B\Re(z)  + C'\log(2+|z|)\biggr)
\end{align*}
for some constant $C'$ that depends only on $\alpha_1$, $\alpha_2$, and $\delta$.
\end{lmm}

%With the bound provided by Lemma \ref{fzfinal2}, we now get the following analogue of Lemma \ref{fconvlmm}.
For non-integer $x>\Re(w)$, define the contour integral
\begin{align}\label{fcdeftwo}
F(x) := \frac{1}{2\pi \i} \int_{\Re(z) = x} \Gamma(-z) f(z)(\mu e^{\sqrt{2}c})^z dz,
\end{align}
where the contour goes from $x-\i \infty$ to $x+\i \infty$. (Note that if $\Re(z)>\Re(w)$ and $z$ is not an integer, then $z$ is automatically in $\Omega$.) 
\begin{lmm}\label{fconvlmm2}
For any non-integer $x> \Re(w)$, the integral defining $F(x)$ in equation~\eqref{fcdeftwo} is absolutely convergent.
\end{lmm}
\begin{proof}
Since $\arg(x+\i y) \to \pm \frac{\pi}{2}$ as $y\to \pm \infty$, Lemma \ref{gammalmm} shows that $\Gamma(-x-\i y)$ decays like $\exp(-\frac{1}{2}\pi|y|)$ as $|y|\to \infty$. By Lemma \ref{fzfinal2}, $|f(x+\i y)|$ behaves like $\exp(-\frac{1}{4}\pi|y|)$ as $|y|\to\infty$. Thus, $\Gamma(-x-\i y)f(x+\i y)$ decays exponentially in $y$ as $|y|\to \infty$. Since $x$ is not an integer, these functions remain bounded near $y=0$. Lastly, $|(\mu e^{\sqrt{2}c})^{x+\i y}| = (\mu e^{\sqrt{2}c})^x$ has no dependence on $y$. This shows that the integral in equation \eqref{fcdeftwo} is absolutely convergent.
\end{proof}

We now proceed to complete the proof of Theorem \ref{corthm} following the same path as we did for the proofs of Theorem \ref{zerocthm} and Theorem \ref{onecthm}.  The following lemma is the analogue of Lemma \ref{zerofinal0} and Lemma \ref{zerofinal0one}. The proof is in \textsection\ref{twofinal0pf}.
\begin{lmm}\label{twofinal0}
Assume the setting of Lemma \ref{twoform1} and that $\Re(w)\in (-\frac{1}{2},0)$. Then for any $x_0\in (\Re(w),0)$ and $N\ge 1$, 
\begin{align*}%\label{contoureq}
F(x_0) = F(N+x_0) + \sum_{n=0}^{N-1}\frac{(-1)^n}{n!}f(n) (\mu e^{\sqrt{2}c})^n.
\end{align*}
\end{lmm}
The next lemma is the analogue of Lemma \ref{zerofinal} and Lemma \ref{zerofinalone}, showing that $F(N+x_0)$ tends to zero as $N\to \infty$ in the above lemma. The proof is in \textsection\ref{zerofinaltwopf}.
\begin{lmm}\label{zerofinaltwo}
Assume that $\Re(w)\in (-\frac{1}{2},0)$. Then for any $x_0\in (\Re(w),0)$, 
\[
\lim_{N\to \infty} F(N+x_0) = 0.
\]
Consequently,  by Lemma \ref{twofinal0}, 
\[
F(x_0) = \sum_{n=0}^\infty \frac{(-\mu e^{\sqrt{2}c})^n}{n!}f(n) = e^{-2\alpha_1\alpha_2}C(\balpha, \bx, 2^{-1/2}, \mu,c).
\]
\end{lmm}

%\section{Proof of Theorem \ref{newcorthm}}
We are now ready to complete the proof of Theorem \ref{corthm}. We only have to show that Lemma \ref{zerofinaltwo} remains valid even if we take $x_0 = \Re(w)$. Let $a$ and $r$ denote the real and imaginary parts of $w$. Lemma \ref{zerofinaltwo}  shows  that for any $x_0\in (a,0)$, 
\begin{align}\label{zalpha}
C(\balpha, \bx, 2^{-1/2}, \mu,c) &= \frac{e^{2\alpha_1\alpha_2}}{2\pi}\int_{-\infty}^\infty\Gamma(-x_0 - \i y)f(x_0+\i y) (\mu e^{\sqrt{2}c})^{x_0+\i y} dy.
\end{align}
Now, for almost all $y\in \R$,
\[
\lim_{x_0\downarrow a} \Gamma(-x_0 - \i y)f(x_0+\i y)  (\mu e^{\sqrt{2}c})^{x_0+\i y}  = \Gamma(-a - \i y)f(a+\i y)  (\mu e^{\sqrt{2}c})^{a+\i y}. 
\]
Thus, we can take $x_0=a$ in equation \eqref{zalpha} if we can show that the condition for applying the dominated convergence theorem holds. For this, take any $x_0\in (a,0)$ and recall that by Lemma \ref{fzfinal2}, we have that for any $z$ with $\Re(z)=x_0$,
\begin{align*}
f(z) &= (z+\beta_1)(z+\beta_2)\exp\biggl(z\log(z-w) +z\log(z-w+1)\\
&\qquad \qquad -\frac{3}{2}z\log (z-w+2)+ Bz + Q(z)\biggr),
\end{align*}
where $B$ is a real universal constant that, and $|Q(z)|\le C\log(2+|z|)$ for some positive constant $C$ that depends only on $\alpha_1,\alpha_2$. In the following, $C,C_0,C_1,\ldots$ will denote arbitrary positive constants that depend only on $\alpha_1,\alpha_2$, whose values may change from line to line. The above expression shows that for any $y\in \R$ and $z:=x_0+\i y$,
\begin{align*}
|f(z)| &\le C_1 (1+|y|)^{C_2}\biggl|\exp\biggl(z\log(z-w) + z\log(z-w+1) -\frac{3}{2}z\log (z-w+2)\biggr)\biggr|\\
&= C_1 (1+|y|)^{C_2}\exp\biggl(\Re\biggl\{z\log(z-w) + z\log(z-w+1) -\frac{3}{2}z\log (z-w+2)\biggr\}\biggr).
\end{align*}
Next, note that 
\begin{align*}
\Re\{z\log(z-w)\} &= x_0\ln|z-w| - y\arg(z-w),
\end{align*}
and similar expressions hold for the other two terms. Now, 
\begin{align*}
x_0\ln|z-w| + x_0\ln|z-w+1|-\frac{3}{2}x_0\ln|z-w+2| &\le x_0\ln|z-w|+ C\log(1+|y|),
\end{align*}
and 
\[
y \biggl(\frac{3}{2}\arg(z-w+2) - \arg(z-w+1)-\arg(z-w)\biggr) \le C_1 -C_2|y|. 
\]
Combining the above observations, we get
\begin{align}
|f(x_0+\i y)| &\le C_1 (1+|y|)^{C_2} e^{-C_3|y|} |z-w|^{x_0}\notag \\
&\le C_1 (1+|y|)^{C_2} e^{-C_3|y|}|y-r|^{x_0}\notag \\
&\le C_1 (1+|y|)^{C_2} e^{-C_3|y|}(1+|y-r|^{a}).\label{fcibd}
\end{align}
Also, by Lemma \ref{gammalmm}, 
\begin{align}\label{gammacibd}
|\Gamma(-x_0-\i y)|\le  \frac{C_1 (1+|y|)^{C_2}e^{-C_3|y|}}{|a|}.
\end{align}
Since $a\in (-\frac{1}{2},0)$  and the constants do not depend on $x_0$, the above bounds suffice to apply the dominated convergence theorem and complete the proof.

\subsubsection{Proof of Theorem \ref{infthm}}\label{infthmpf}
Note that $a$ and $r$ denote the real and imaginary parts of $w$. The $x_0$-free  bounds \eqref{fcibd} and \eqref{gammacibd} allow us take $x_0\to a$, and apply Fubini's theorem and conclude that 
\begin{align*}
C_\ep(\balpha,\bx, 2^{-1/2}, \mu) &= \frac{e^{2\alpha_1\alpha_2}}{2\pi}\int_{-\infty}^\infty\int_{-\infty}^\infty \Gamma(-a - \i y)f(a+\i y) (\mu e^{\sqrt{2}c})^{a+\i y} e^{-\sqrt{2}wc-\ep^2 c^2}dc dy\\
&= \frac{e^{2\alpha_1\alpha_2}}{2\pi}\int_{-\infty}^\infty\int_{-\infty}^\infty \Gamma(-a - \i y)f(a+\i y) \mu^{a+\i y} e^{\sqrt{2}\i (y-r)c-\ep^2 c^2}dc dy\\
&= \frac{e^{2\alpha_1\alpha_2}}{2\sqrt{\pi}\ep}\int_{-\infty}^\infty \Gamma(-a - \i y)f(a+\i y) \mu^{a+\i y} \exp\biggl(-\frac{(y-r)^2}{2\ep^2}\biggr)dy.
\end{align*}
Making the change of variable $u = (y-r)/\ep$, we get
\begin{align*}
C_\ep(\balpha,\bx, 2^{-1/2}, \mu) 
&= \frac{ e^{2\alpha_1\alpha_2}}{2\sqrt{\pi}}\int_{-\infty}^\infty \Gamma(-w- \i\ep u)f(w+\i \ep u) \mu^{w+\i \ep u} e^{-\frac{1}{2}u^2}du.
\end{align*}
Now, note that for any $u\in \R$,
\[
\lim_{\ep\to 0} \Gamma(-w- \i\ep u)\mu^{w+\i \ep u}= \Gamma(-w) \mu^{w}.
\]
Next, note that by Lemma \ref{gammalmm2}, 
\begin{align*}
f(z) &=  (z-w)^z\frac{(4\pi)^z e^{\frac{1}{2}z^2  -(\frac{3}{2}+w)z}\Gamma(z+1) G(z+\beta_1)G(z + \beta_2)}{(\Gamma(1+z-w))^{z}G(\beta_1)G(\beta_2)}
\end{align*}
for any $z$ such that $z-w\notin (-\infty,0]$. Now, 
\begin{align*}
&\lim_{z\to w} \frac{(4\pi)^z e^{\frac{1}{2}z^2  -(\frac{3}{2}+w)z}\Gamma(z+1) G(z+\beta_1)G(z + \beta_2)}{(\Gamma(1+z-w))^{z}G(\beta_1)G(\beta_2)} \\
&= \frac{(4\pi)^w e^{-\frac{1}{2}w^2  -\frac{3}{2}w}\Gamma(w+1) G(w+\beta_1)G(w + \beta_2)}{G(\beta_1)G(\beta_2)} =: S.
\end{align*}
On the other hand, if $z = w + \i \ep u$, then as $\ep \to 0$,
\begin{align*}
(z-w)^z &= (\i \ep u)^{w+\i \ep u} \\
&= \exp((w+\i \ep u)\log(\i \ep u))\\
&= \exp\biggl\{(w+\i \ep u)\biggl(\ln|\ep u| + \frac{\i \pi}{2}\sign(u) \biggr)\biggr\}\\
&\sim \ep^w (\i u)^w \ \ \text{ as $\ep \to 0$.}
\end{align*}
Next, note that by the bounds \eqref{fcibd} and \eqref{gammacibd}, 
\begin{align*}
|\ep^{-w} \Gamma(-w- \i\ep u)f(w+\i \ep u) \mu^{w+\i \ep u} | &\le C_1|\ep^{-w}|(1+ |\ep u|^{a})\le C_2,
\end{align*}
where $C_1,C_2$ do not depend on $u$. This allows us to apply the dominated convergence theorem and conclude that 
\begin{align*}
\lim_{\ep \to 0} \ep^{-w} C_\ep(\balpha,\bx,2^{-1/2},\mu) &= \frac{ e^{2\alpha_1\alpha_2}\Gamma(-w)\mu^wS}{2\sqrt{\pi }}\int_{-\infty}^\infty(\i u)^w e^{-\frac{1}{2}u^2} du.
\end{align*}
To evaluate the integral, note that 
\begin{align*}
&\int_0^\infty(\i u)^w e^{-\frac{1}{2}u^2} du = \int_0^\infty e^{w\ln u + \frac{1}{2}\i \pi w}e^{-\frac{1}{2}u^2} du,\\
&\int_{-\infty}^0(\i u)^w e^{-\frac{1}{2}u^2} du = \int_0^\infty e^{w\ln u - \frac{1}{2}\i \pi w}e^{-\frac{1}{2}u^2} du.
\end{align*}
Thus, 
\begin{align}\label{cosid}
\frac{1}{2}\int_{-\infty}^\infty(\i u)^w e^{-\frac{1}{2}u^2} du &= \cos\biggl(\frac{\pi w}{2}\biggr)\int_0^\infty u^we^{-\frac{1}{2}u^2} du \notag \\
&=2^{\frac{1}{2}(w-1)}\cos\biggl(\frac{\pi w}{2}\biggr) \int_0^\infty s^{\frac{1}{2}(w-1)} e^{-s} ds\notag \\
&=2^{\frac{1}{2}(w-1)} \cos\biggl(\frac{\pi w}{2}\biggr)\Gamma \biggl(\frac{w+1}{2}\biggr),
\end{align}
where the last identity holds because $\Re(\frac{1}{2}(w+1)) > 0$. This completes the proof of the main assertion of Theorem \ref{infthm}. To show that the limit is nonzero, we only need to show that $G(w+\beta_1)$ and $G(w+\beta_2)$ are nonzero. Since the only zeros of $G$ are at the nonpositive integers, we have to show that $w+\beta_1$ and $w+\beta_2$ are not nonpositive integers. But $w+\beta_1= -\sqrt{2}\alpha_2$ and $w+\beta_2 = -\sqrt{2}\alpha_1$. Thus, we have to show that 
$\sqrt{2}\alpha_j$ is not a nonnegative integer for $j=1,2$. By assumption, we have $\Re(\sqrt{2}\alpha_j) > -1$. On the other hand, since $\Re(w)\in (-1,0)$, we have 
\[
\Re(\sqrt{2}\alpha_1) + \Re(\sqrt{2}\alpha_2) \in (-1,0).
\]
So, if $\sqrt{2}\alpha_1 = n$ for a nonnegative integer $n$, then 
\[
\Re(\sqrt{2}\alpha_2) \in (-n-1,-n),
\]
But this is impossible unless $n=0$, since $\Re(\sqrt{2}\alpha_2)>-1$.

\subsubsection{Proof of Theorem \ref{zerocorthm}}\label{zerocorthmpf}
For each $\epsilon\in (0,\frac{1}{4})$, let $\alpha_{1,\epsilon} := \alpha_1 + \sqrt{2}\epsilon$. Let $\balpha_\epsilon := (\alpha_{1,\epsilon}, \alpha_2)$, $w_\epsilon := -1-\sqrt{2}(\alpha_{1,\epsilon}+\alpha_2)= w-2\epsilon$, and $\beta_{1,\epsilon} := 1+\sqrt{2}\alpha_{1,\epsilon} = \beta_1 + 2\epsilon$. Then note that $\Re(\alpha_{1,\ep})=\Re(\alpha_1)+\sqrt{2}\epsilon > -\frac{1}{\sqrt{2}}$, $\Re(\alpha_2)=0>-\frac{1}{\sqrt{2}}$, and  $\Re(w_\epsilon) = -2\epsilon \in (-\frac{1}{2},0)$. Thus, we can apply Lemma \ref{zerofinaltwo} to conclude that for any $x\in (-2\epsilon, 0)$, 
\begin{align*}
C(\balpha_\epsilon, \bx, 2^{-1/2},\mu,c) = \frac{e^{2\alpha_{1,\epsilon}\alpha_2}}{2\pi \i}\int_{\Re(z)=x}\Gamma(-z)f_\epsilon(z) (\mu e^{\sqrt{2}c})^z dz,
\end{align*}
where 
\[
f_\epsilon(z) := \frac{(4\pi)^z e^{\frac{1}{2}z^2  -(\frac{3}{2}+w_\epsilon)z}\Gamma(z+1) G(z+\beta_{1,\epsilon})G(z + \beta_2)}{(\Gamma(z-w_\epsilon))^zG(\beta_{1,\epsilon})G(\beta_2)}.
\]
The following lemma allows us to replace $x$ by $q\in (0,1)$, fixing $\ep$, after incurring an extra term due to crossing a pole at zero.
\begin{lmm}\label{limlmm1}
For any $q\in (0,1)$, we have
\begin{align*}
C(\balpha_\epsilon, \bx, 2^{-1/2},\mu,c)  &= \frac{e^{2\alpha_{1,\epsilon}\alpha_2}}{2\pi}\int_{-\infty}^\infty\Gamma(-q-\i y) f_\epsilon(q+\i y) (\mu e^{\sqrt{2}c})^{q+\i y}  dy + e^{2\alpha_{1,\epsilon}\alpha_2}f_\ep(0),
\end{align*}
and the integral on the right is absolutely convergent.
\end{lmm}
\begin{proof}
Let $g(z) := \Gamma(-z)f_\ep(z)(\mu e^{2\sqrt{c}})^z$. Take any $x\in (-2\ep,0)$. The only pole of $g$ in the strip $-2\ep < \Re(z)< 1$ is at $0$, arising due the pole of the Gamma function at $0$, and 
\[
\Res(g,0) = -f_\ep(0)\Res(\Gamma,0) = -f_\ep(0).
\]
By the bounds on $|f_\ep(z)|$ and $|\Gamma(-z)|$ from Lemma \ref{fzfinal2} and Lemma \ref{gammalmm}, this allows us to shift the contour of integration below, to get
\begin{align*}
C(\balpha_\epsilon, \bx, 2^{-1/2},\mu,c) &= \frac{e^{2\alpha_{1,\epsilon}\alpha_2}}{2\pi \i}\int_{\Re(z)=x}g(z) dz\\
&= \frac{e^{2\alpha_{1,\epsilon}\alpha_2}}{2\pi \i}\int_{\Re(z)=q}g(z) dz -e^{2\alpha_{1,\epsilon}\alpha_2}\Res(g,0)\\
&= \frac{e^{2\alpha_{1,\epsilon}\alpha_2}}{2\pi}\int_{-\infty}^\infty g(q+\i y) dy + e^{2\alpha_{1,\epsilon}\alpha_2}f_\ep(0).
\end{align*}
This completes the proof.
\end{proof}
We are now ready to prove Theorem \ref{zerocorthm}, by taking $\ep\to 0$. Since $w$ is purely imaginary, let us write $w = \i a$ for some $a\in \R$, so that $w_\epsilon = -2\epsilon + \i a$. 
By Lemma \ref{convergence}, we know that 
\[
C(\balpha, \bx, 2^{-1/2},\mu,c) = \lim_{\ep\to 0} C(\balpha_\epsilon, \bx, 2^{-1/2},\mu,c).
\]
Thus, we need to compute the limit of the right side in Lemma \ref{limlmm1}. An easy verification shows that for any $y\ne a$ and $x\in [0,1)$, 
\begin{align*}%\label{feplimit}
\lim_{\epsilon \to 0} f_\epsilon(x+ \i y) &=   f(x+\i y).
\end{align*}
In particular, since $a\ne 0$, we have 
\begin{align*}
\lim_{\epsilon \to 0} f_\epsilon(0) &=   f(0)=1.
\end{align*}
So, it only remains to show that we can move $\lim_{\ep\to 0}$ inside the integral displayed in Lemma~\ref{limlmm1}. But again, this is easy from the bounds given by Lemma~\ref{fzfinal2} and Lemma~\ref{gammalmm}, which give $\ep$-independent control on the integrand at large $y$, and the fact that $q>0$, which gives $\ep$-independent control on the integrand at small $y$.

\subsubsection{Proof of Theorem \ref{distthm}}\label{distthmpf}
Let $a := -\sqrt{2}(P_1+P_2)$, so that $w = \i a$. Recall that 
\[
\beta_1 = 1+\sqrt{2}\alpha_1 = \frac{1}{2} + \i\sqrt{2} P_1,\ \ \ \beta_2 = \frac{1}{2} + \i\sqrt{2} P_2 = \frac{1}{2} + \i (-a-\sqrt{2}P_1).
\]
Now, we have 
\begin{align*}
2\alpha_1\alpha_2 &= 2\biggl(-\frac{1}{2\sqrt{2}} + \i P_1\biggr) \biggl(-\frac{1}{2\sqrt{2}} +\i P_2\biggr)\\
&= \frac{1}{4} -\frac{\i}{\sqrt{2}}(P_1+P_2) -2P_1P_2\\
&= \frac{1}{4} + \frac{\i a}{2} - 2P_1\biggl(-\frac{a}{\sqrt{2}} - P_1 \biggr)= \frac{1}{4} + \frac{\i a}{2} + \sqrt{2}P_1 a + 2P_1^2.
\end{align*}
Thus, fixing $P_1$, we have $e^{2\alpha_1\alpha_2} = J(a)$, where 
\[
J(a) := \exp\biggl(\frac{1}{4} + \frac{\i a}{2} + \sqrt{2}P_1 a + 2P_1^2\biggr).
\]
We will continue with $P_1$ fixed. Define a function $F$ of two complex variables as 
\begin{align*}
F(u,v) &:= (4\pi\mu)^{\i u}\exp\biggl(\frac{1}{4} + \frac{\i v}{2} + \sqrt{2}P_1 v + 2P_1^2+\frac{1}{2}u^2  -\biggl(\frac{3}{2}+\i v\biggr)\i u\biggr)  \\
&\qquad \cdot (\Gamma(1+\i (u-v)))^{-\i u}\Gamma(1-\i u) \Gamma(1+\i u)\\
&\qquad \cdot\frac{G(\frac{1}{2} + \i(u+\sqrt{2} P_1))G(\frac{1}{2} + \i (u-v-\sqrt{2}P_1))}{G(\frac{1}{2} + \i\sqrt{2} P_1)G(\frac{1}{2} + \i (-v-\sqrt{2}P_1))},
\end{align*}
wherever the right side makes sense. This function arises due to the following lemma.
\begin{lmm}\label{cnewlmm}
We have, for any $q\in (0,1)$, 
\begin{align*}
C(\balpha, \bx, 2^{-1/2}, \mu,c) &=  -\frac{1}{2\pi}\int_{-\infty}^\infty  \frac{(q + \i (y-a))^{q+\i y}}{q+\i y}F(y-\i q,a) e^{\sqrt{2}c(q+\i y)}dy + J(a).
\end{align*}
\end{lmm}
\begin{proof}
By Theorem \ref{zerocorthm}, we have 
\begin{align*}
C(\balpha, \bx, 2^{-1/2}, \mu,c) &= \frac{J(a)}{2\pi }\int_{-\infty}^\infty \Gamma(-q-\i y)f(q+\i y) (\mu e^{\sqrt{2}c})^{q+\i y} dy + J(a),%\\
%&=  -\frac{1}{2\pi}\int_{-\infty}^\infty  \frac{(q + \i (y-a))^{\i y}}{q+\i y}F(y-\i q,a) e^{\sqrt{2}c(q+\i y)}dy + J(a).
\end{align*}
where
\begin{align*}
f(z) = \frac{(4\pi)^z e^{\frac{1}{2}z^2  -(\frac{3}{2}+w)z}\Gamma(z+1) G(z+\beta_1)G(z + \beta_2)}{(\Gamma(z-w))^zG(\beta_1)G(\beta_2)}.
\end{align*}
Now, from the formula for $F$, we get
\begin{align*}
F(u, a) &= (4\pi \mu)^{\i u} J(a) e^{-\frac{1}{2}(\i u)^2  -(\frac{3}{2}+w)\i u}  \\
&\qquad \cdot (\Gamma(1+\i u-w))^{-\i u}\Gamma(1-\i u) \Gamma(1+\i u)\frac{G(\i u+\beta_1)G(\i u+\beta_2)}{G(\beta_1)G(\beta_2)}.
\end{align*}
By  Lemma \ref{gammalmm2}, this shows that 
\begin{align}\label{fuvid}
F(u, a) &= (4\pi \mu)^{\i u} J(a) e^{-\frac{1}{2}(\i u)^2  -(\frac{3}{2}+w)\i u}  \notag\\
&\qquad \cdot (\i u - w)^{-\i u}(\Gamma(\i u-w))^{-\i u}(-\i u)\Gamma(-\i u) \Gamma(1+\i u)\frac{G(\i u+\beta_1)G(\i u+\beta_2)}{G(\beta_1)G(\beta_2)}\notag \\
&= -\frac{\i u}{(\i u-w)^{\i u}} J(a) \Gamma(-\i u) f(\i u)\mu^{\i u}.
\end{align}
Taking $u = y-\i q$, we get 
\[
J(a)\Gamma(-q-\i y)f(q+\i y) \mu^{q+\i y} = - \frac{(q+\i(y-a))^{q+\i y}}{q+\i y}F(q+\i y, a).
\]
This completes the proof.
\end{proof}
We need certain bounds for $F$ and its derivatives before proceeding further.
\begin{lmm}\label{annoying}
Let $D$ be any mixed partial derivative operator of any order (including the identity operator). Let $K$ be a positive real number and take any $v\in (-K,K)$. Then there are positive constants $C_1,C_2$ depending only on $\mu$, $P_1$, $K$, and $D$ such that for any $u\in \C$ with $|\Im(u)|< \frac{1}{10}$, we have $|DF(u,v)|\le C_1 e^{-C_2|u|}$. 
\end{lmm}
\begin{proof}
Define the region 
\[
\Omega_0 := \biggl\{(u,v)\in \C^2: |\Re(v)|<K, \, |\Im(v)|< \frac{1}{8}, \, |\Im(u)|< \frac{1}{8}\biggr\}.
\]
Let $\Omega_1$ be the domain obtained by replacing the $\frac{1}{8}$'s with $\frac{1}{9}$'s above, and let $\Omega_2$ be obtained by replacing the $\frac{1}{8}$'s with $\frac{1}{10}$'s. 

Take any $(u,v)\in \Omega_0$. Then  $|\Re(\i u)|< \frac{1}{8}$ and $|\Re(\i v)|< \frac{1}{8}$; thus, $|\Re(\i(u-v))|< \frac{1}{4}$. From this and the formula for $F$, we see that $F$ is analytic on the domain $\Omega_0$. Thus, if we can show that the required bound holds for $|F(u,v)|$ on $\Omega_1$, then Cauchy's estimates for derivatives of analytic functions will complete the proof of the bound for $|DF(u,v)|$ on~$\Omega_2$.

The analyticity of $F$ on $\Omega_0$ also implies that it $F$ is uniformly bounded in any bounded subset of $\Omega_1$. Thus, suffices to prove that there exist $C_1,C_2,C_3$ such that $|F(u,v)|\le C_1 e^{-C_2|u|}$ for any $(u,v)\in \Omega_1$ with $|u|>C_3$. In other words, we have to show that  $|F(y-\i x,v)|$ decays exponentially in $|y|$ as $|y|\to \infty$, with a decay rate that is uniform over all $v\in \C$ with $|\Re(v)|< K$ and $|\Im(v)|<\frac{1}{9}$, and over all $x\in \R$ with $|x|<\frac{1}{9}$.

To understand the asymptotics of $F(y-\i x,v)$ as $|y|\to \infty$, we will appeal to Lemma~\ref{fzfinal2} and Lemma \ref{gammalmm}. Throughout the following, implicit constants in $O$'s are uniform over all $v,x$ in the above regions. First note that by the identity~\eqref{fuvid}, 
\begin{align*}
F(y-\i x,v) &= -e^{2\tilde{\alpha}_1\tilde{\alpha}_2}(x+\i y) (x + \i(y-a))^{-x-\i y} \Gamma(-x-\i y) \tilde{f}(x+\i y)\mu^{x+\i y},
\end{align*}
where 
\[
\tilde{\alpha}_1 := \alpha_1 = -\frac{1}{2\sqrt{2}}+\i P_1,\ \ \ \tilde{\alpha}_2 := -\frac{1}{2\sqrt{2}} + \i \biggl(-\frac{v}{\sqrt{2}}-P_1\biggr),
\]
and $\tilde{f}$ is the version of $f$ obtained with $\tilde{\alpha}_1,\tilde{\alpha}_2$ in place of $\alpha_1,\alpha_2$. Note that then $w$ is replaced by 
\[
\tilde{w} = -1-\sqrt{2}(\tilde{\alpha}_1+\tilde{\alpha_2}) = \i v.
\]
Note also that $\Re(\tilde{\alpha}_1)=\Re(\alpha_1)= -\frac{1}{2\sqrt{2}} >-\frac{1}{\sqrt{2}}$, and 
\[
\Re(\tilde{\alpha}_2) = -\frac{1}{2\sqrt{2}}+ \frac{1}{\sqrt{2}} \Im(v)>-\frac{1}{\sqrt{2}}.
\]
Also, $\Re(x+\i y) = x > -\frac{1}{2}$. Thus, by Lemma \ref{fzfinal2}, 
\[
|\tilde{f}(x+\i y)|=\exp\biggl(-\frac{\pi}{4}|y|+O(\ln|y|)\biggr) \ \ \text{ as $|y|\to\infty$.}
\]
By Lemma \ref{gammalmm}, 
\[
\Gamma(-x-\i y) = \exp\biggl(-\frac{\pi}{2}|y|+O(\ln|y|)\biggr) \ \ \text{ as $|y|\to\infty$.}
\]
Lastly, 
\[
|(x + \i(y-v))^{-x-\i y}| = \exp\biggl(\frac{\pi}{2}|y|+O(\ln|y|)\biggr) \ \ \text{ as $|y|\to \infty$.}
\]
Combining, we get that 
\[
|F(y+\i x,v)| = \exp\biggl(-\frac{\pi}{4}|y|+O(\ln|y|)\biggr) \ \ \text{ as $|y|\to\infty$.}
\]
This completes the proof.
\end{proof}
We will now apply Lemma \ref{cnewlmm} with $q=\ep$. Lemma \ref{annoying} allows us to exchange the order of integration below, to get
\begin{align*}
C_\ep(\balpha, \bx, 2^{-1/2},\mu) &= \int_{-\infty}^\infty e^{-\sqrt{2}\i a c - \ep^2 c^2} C(\balpha, \bx,2^{-1/2},\mu, c) dc\\
&=- \frac{1}{2\pi}\int_{-\infty}^\infty  \int_{-\infty}^\infty \frac{(\ep + \i (y-a))^{\ep+\i y}}{\ep+\i y}F(y-\i \ep,a) e^{\sqrt{2}c(\ep+\i (y-a)) -\ep^2 c^2}dc dy \\
&\qquad + \int_{-\infty}^\infty J(a)e^{-\sqrt{2}\i a c - \ep^2 c^2}dc.
\end{align*}
Evaluating the integrals over $c$, we get 
\begin{align*}
C_\ep(\balpha, \bx, 2^{-1/2},\mu) &=- \frac{1}{2\sqrt{\pi}\ep}\int_{-\infty}^\infty  \frac{(\ep + \i (y-a))^{\ep+\i y}}{\ep+\i y}F(y-\i \ep,a) e^{(\ep+\i (y-a))^2/(2\ep^2)}dy \\
&\qquad +\frac{\sqrt{\pi}}{\ep} J(a)e^{-a^2/(2\ep^2)}. 
\end{align*}
Making the change of variable $u=(y-a)/\ep$, we arrive at 
\begin{align}\label{cesplit}
C_\ep(\balpha, \bx, 2^{-1/2},\mu) &= -\frac{1}{2\sqrt{\pi}}\int_{-\infty}^\infty  \frac{(\ep + \i \ep u)^{\ep+\i a+\i \ep u}}{\ep+\i a +\i \ep u}F(a + \ep u-\i \ep,a) e^{\frac{1}{2}(1+\i u)^2}du\notag \\
&\qquad +\frac{\sqrt{\pi}}{\ep} J(a)e^{-a^2/(2\ep^2)}. 
\end{align}
Let $C_\ep(a)$ denote the above quantity. Let the two terms on the right be denoted by $D_\ep(a)$ and $E_\ep(a)$. 
\begin{lmm}\label{dbase}
For any smooth function $\varphi:\R\to\R$ with compact support, we have 
\[
\lim_{\ep\to 0} \int_{-\infty}^\infty D_\ep(a) \varphi(a) da  = 0.
\]
Moreover, %there is some constant $C$ depending only on $\varphi$ such that for any $\ep\in (0,1)$,
\[
\sup_{0< \ep< \frac{1}{10}}\biggl| \int_{-\infty}^\infty D_\ep(a) \varphi(a) da \biggr|\le C,
\]
where $C$ is a finite constant that depends on $\varphi$ only through the size of the support of $\varphi$ and  upper bounds on $\varphi$ and its first two derivatives.
\end{lmm}
\begin{proof}
Throughout this proof, $C,C_1,C_2,\ldots$ will denote arbitrary constants that may depend only on $\alpha_1,\alpha_2,\varphi$, whose values may change from line to line. 
Take any $\ep\in (0,\frac{1}{10})$. 
Note that
\begin{align*}
&\int_{-\infty}^\infty D_\ep(a) \varphi(a) da \\
&= -\frac{1}{2\sqrt{\pi}}\int_{-\infty}^\infty  \int_{-\infty}^\infty  \frac{(\ep + \i \ep u)^{\ep+\i a+\i \ep u}}{\ep+\i a +\i \ep u}F(a + \ep u-\i \ep,a) e^{\frac{1}{2}(1+\i u)^2}\varphi(a) du da\\
&= -\frac{1}{2\sqrt{\pi}}\int_{-\infty}^\infty  \int_{-\infty}^\infty  \frac{(\ep + \i \ep u)^{\ep+\i x}}{\ep+\i x}F(x-\i \ep,x-\ep u) e^{\frac{1}{2}(1+\i u)^2}\varphi(x-\ep u) dx du,
%&= \frac{-1}{2\sqrt{\pi}}\int_{-\infty}^\infty  \int_{-\infty}^\infty  \frac{(\ep + \i \ep u)^{\i \ep x}}{1+\i x}\xi_\ep(x,u) e^{\frac{1}{2}(1+\i u)^2} du dx,
\end{align*}
where the second identity is obtained by substituting $x = a+\ep u$ and changing order of integration using Fubini's theorem (which applies because $F$ is a bounded function, by Lemma \ref{annoying}, and $\varphi$ is compactly supported). Substituting
\begin{align*}
\frac{1}{\ep+\i x} = \int_0^\infty e^{-t(\ep+\i x)} dt,
\end{align*}
and again using Fubini's theorem, we get
\begin{align*}
&\int_{-\infty}^\infty D_\ep(a) \varphi(a) da \\
&= \frac{-1}{2\sqrt{\pi}} \int_{-\infty}^\infty \int_0^\infty \int_{-\infty}^\infty  (\ep + \i \ep u)^{\ep+\i x} e^{-t(\ep+\i x)} F(x-\i \ep,x-\ep u) e^{\frac{1}{2}(1+\i u)^2}\varphi(x-\ep u) dx dt du\\
&= \frac{-1}{2\sqrt{\pi}} \int_{-\infty}^\infty \int_0^\infty e^{-\ep t} e^{\frac{1}{2}(1+\i u)^2}(\ep + \i \ep u)^{\ep} \psi_{\ep,u}(t-\log(\ep+\i \ep u) )dt du,
%&= \frac{-1}{2\sqrt{\pi}}\int_{-\infty}^\infty  \int_{-\infty}^\infty  \frac{(\ep + \i \ep u)^{\i \ep x}}{1+\i x}\xi_\ep(x,u) e^{\frac{1}{2}(1+\i u)^2} du dx,
\end{align*}
where 
\[
\psi_{\ep,u}(\theta) := \int_{-\infty}^\infty e^{-\i \theta x}F(x-\i\ep, x-\ep u) \varphi(x-\ep u) dx.
\]
Integration by parts gives 
\begin{align*}
\psi_{\ep,u}(\theta) = -\frac{1}{\theta^2}\int_{-\infty}^\infty e^{-\i \theta x}\frac{d^2}{d x^2}(F(x-\i\ep, x-\ep u) \varphi(x-\ep u)) dx.
\end{align*}
Let $[-K,K]$ be an interval containing the support of $\varphi$. Then the integrand above is zero if $|x-\ep u|>K$; otherwise, by Lemma \ref{annoying}, we have
\begin{align*}
\biggl|e^{-\i \theta x}\frac{d^2}{d x^2}(F(x-\i\ep, x-\ep u) \varphi(x-\ep u))\biggr|&\le C_1e^{x \Im(\theta)-C_2|x|}\\
&\le C_1 e^{(x-\ep u) \Im(\theta)   + \ep u\Im(\theta)  + C_2 |x-\ep u|- C_2\ep|u|}\\
&\le C_1 e^{C_3 +(C_4+\ep u) \Im(\theta)}.
\end{align*}
This shows that 
\[
|\psi_{\ep,u}(\theta)|\le \frac{C_1}{|\theta|^2}e^{C_2 +(C_3+\ep u) \Im(\theta)}.
\]
Note that 
\begin{align*}
\Im(t- \log(\ep+\i \ep u)) &= \arg(1+\i u)\in (-\pi,\pi).
\end{align*}
Thus, we arrive at the inequality
\begin{align*}
\biggl|\int_{-\infty}^\infty D_\ep(a) \varphi(a) da\biggr| &\le C_1\int_{-\infty}^\infty\int_0^\infty \frac{|1+\i u|^\ep e^{\frac{1}{2}(1-u^2) + C_2\ep |u|} }{|t - \log(\ep + \i \ep u)|^2}dt du.
\end{align*}
Now, we have
\begin{align*}
|t - \log(\ep + \i \ep u)|^2 &= |t - \ln|\ep + \i \ep u| - \i \arg(\ep+\i \ep u)|^2\\
&= |t - \ln \ep - \ln |1+\i u| - \i \arg(1+\i u)|^2\\
&= (t-\ln \ep-\ln (1+\i u))^2 + (\arg(1+\i u))^2\\
&\ge (t+\ln 10-\ln (1+\i u))^2 + (\arg(1+\i u))^2.
\end{align*}
This shows that
\begin{align*}
\biggl|\int_{-\infty}^\infty D_\ep(a) \varphi(a) da\biggr| &\le C_1\int_{-\infty}^\infty\int_0^\infty \frac{|1+\i u|^{\frac{1}{10}} e^{\frac{1}{2}(1-u^2)+C_2|u|} }{(t+\ln 10-\ln (1+\i u))^2 + (\arg(1+\i u))^2}dt du,
\end{align*}
and it is easy to see that the right side is finite. Moreover, it also allows to apply the dominated convergence theorem and see that the limit of the left side, as $\ep \to 0$, is zero.
\end{proof}

Let us now start writing $D_\ep(P_1,P_2)$ instead of $D_\ep(a)$ to emphasize that $D_\ep$ is a function of $(P_1,P_2)$. Lemma \ref{dbase} implies the following result. 
\begin{lmm}\label{dlmm} 
For any smooth function $\varphi:\R^2\to [-1,1]$ with compact support,
\[
\lim_{\ep \to 0} \iint D_\ep(P_1,P_2) \varphi(P_1,P_2)dP_1dP_2 = 0.
\]
\end{lmm}
\begin{proof}
Let $L$ be a number so large that $\varphi(P_1,P_2)=0$ if $|P_1|> L$ or $|P_2|>L$. Take any $P_1\in [-L,L]$. Then 
\begin{align*}
&\int_{-\infty}^\infty D_\ep(P_1,P_2)\varphi(P_1,P_2) dP_2 = \frac{1}{\sqrt{2}}\int_{-\infty}^\infty D_\ep(P_1, -2^{-1/2}a -P_1)\varphi(P_1, -2^{-1/2}a -P_1) da.
\end{align*}
According to our previous notation, the right side is
\[
\frac{1}{\sqrt{2}}\int_{-\infty}^\infty D_\ep(a)\psi(a) da,
\]
where $\psi(a) := \varphi(P_1, -2^{-1/2}a -P_1)$. By Lemma \ref{dbase}, this quantity, which is a function of $P_1$, tends to zero as $\ep\to0$, and its absolute value is bounded above by a number that does not depend on $P_1$. By the dominated convergence theorem, this completes the proof.
\end{proof}

We need one final lemma to complete the proof of Theorem \ref{distthm}. Recall the function $E_\ep$ defined below equation \eqref{cesplit}. Let us write it as $E_\ep(P_1,P_2)$ to emphasize the dependence on~$(P_1,P_2)$. 
\begin{lmm}\label{elmm}
For any smooth compactly supported function $\varphi:\R^2 \to \R$,
\[
\lim_{\ep\to 0} \int_{-\infty}^\infty \int_{-\infty}^\infty E_\ep(P_1,P_2)\varphi(P_1,P_2) dP_1 dP_2 = \pi \int_{-\infty}^\infty e^{\frac{1}{4}+ 2P_1^2}\varphi(P_1,-P_1)dP_1. 
\]
%\frac{\sqrt{\pi}}{2\sqrt{\ep}}\int_{\R^2} e^{2\alpha_{1}\alpha_2}\exp\biggl(-\frac{a^2}{2\ep}\biggr) \varphi(P_1,P_2) dP_1dP_2 
\end{lmm}
\begin{proof}
Take any $P_1\in \R$. Then  
\begin{align*}
&\int_{-\infty}^\infty E_\ep(P_1,P_2)\varphi(P_1,P_2) dP_2 = \frac{1}{\sqrt{2}}\int_{-\infty}^\infty E_\ep(P_1, -2^{-1/2}a -P_1)\varphi(P_1, -2^{-1/2}a -P_1) da\\
&= \frac{\sqrt{\pi}}{\sqrt{2}\ep}\int_{-\infty}^\infty  J(a) e^{-a^2/(2\ep^2)}\varphi(P_1, 2^{-1/2}a -P_1) da
\end{align*}
Making the change of variable $u = a/\ep$, this becomes
\begin{align*}
\frac{\sqrt{\pi}}{\sqrt{2}}\int_{-\infty}^\infty J(\ep u) e^{-\frac{1}{2}u^2}\varphi(P_1, 2^{-1/2}\ep u -P_1)du.
\end{align*}
Since $\varphi$ is compactly supported, we can now integrate the above over $P_1$ and get 
\begin{align*}
&\iint E_\ep(P_1,P_2)\varphi(P_1,P_2) dP_2 dP_1\\
&=\frac{\sqrt{\pi}}{\sqrt{2}}\int_{-L}^L\int_{-\infty}^\infty J(\ep u) e^{-\frac{1}{2}u^2} \varphi(P_1, 2^{-1/2}\ep u -P_1)du dP_1,
\end{align*}
where $L$ is a number so large that $\varphi(P_1,P_2)=0$ when $|P_1|> L$. The integrand decays sufficiently fast as $|u|\to \infty$ to ensure that we can now apply the dominated convergence theorem to evaluate the limit of the above as $\ep\to 0$ by moving the limit inside the double integral. After taking the limit, the integrals over $u$ and $P_1$ factorize, to give
\begin{align*}
&\lim_{\ep\to 0} \iint E_\ep(P_1,P_2)\varphi(P_1,P_2) dP_2 dP_1\\
&=\frac{\sqrt{\pi}}{\sqrt{2}}J(0) \biggl(\int_{-L}^L \varphi(P_1,  -P_1)dP_1\biggr)\biggl(\int_{-\infty}^\infty  e^{-\frac{1}{2}u^2}du\biggr)\\
&= \pi J(0)\int_{-L}^L \varphi(P_1,  -P_1)dP_1.
\end{align*}
This completes the proof.
\end{proof}

Clearly, Lemma \ref{dlmm} and Lemma \ref{elmm} together complete the proof of Theorem \ref{distthm}.

\subsubsection{Proof of Theorem \ref{distthm2}}\label{distthm2pf}
From Lemma \ref{zeroform1two}, we have
\begin{align}\label{repeat}
C(\balpha, \bx, 2^{-1/2}, \mu, c) =e^{2\alpha_1\alpha_2}\sum_{n=0}^\infty \frac{(-\mu e^{\sqrt{2}c})^n}{n!}a_n,
\end{align}
where $a_n$ are the coefficients defined in equation \eqref{cn2}. Although Lemma \ref{zeroform1two} is stated for $c\in \R$, it is easy to see from the proof that the formula continues to be valid for any complex $c$. By the estimates observed in the proof of Lemma \ref{convergence}, we can integrate $c$ from $0$ to $\sqrt{2}\pi \i$ by moving the integral inside the infinite sum, to get
\begin{align}\label{horizontal}
 -\i \int_0^{\sqrt{2}\pi} e^{-\sqrt{2}\i wt} C(\balpha, \bx, 2^{-1/2}, \mu, \i t) dt &= -\i e^{2\alpha_1\alpha_2}\sum_{n=0}^\infty \frac{(-\mu )^n}{n!}a_n\int_0^{\sqrt{2}\pi} e^{\sqrt{2}\i (n-w)t} dt\notag\\
&= e^{2\alpha_1\alpha_2}\sum_{n=0}^\infty \frac{(-\mu )^n(1-e^{2\pi \i (n-w)})}{\sqrt{2} n! (n-w)}a_n\notag \\
&= (1-e^{-2\pi \i w}) e^{2\alpha_1\alpha_2}\sum_{n=0}^\infty \frac{(-\mu )^n}{\sqrt{2} n! (n-w)}a_n.
\end{align}
Let us now specialize to the case $\alpha_j = \frac{1}{2}Q+\i P_j$, $j=1,2$. Then $w =\i a$, where $a:= -\sqrt{2}(P_1+P_2)$.  Note that
\begin{align*}
&(1-e^{-2\pi \i w})\int_0^\infty e^{-\sqrt{2}wc - \ep^2 c^2} C(\balpha, \bx, 2^{-1/2},\mu ,c) dc\\
&= (1-e^{2\pi a}) \int_{-\infty}^\infty e^{-\sqrt{2}\i a c - \ep^2 c^2} C(\balpha, \bx, 2^{-1/2},\mu ,c) dc \\
&\qquad \qquad - (1-e^{2\pi a}) \int_{-\infty}^0 e^{-\sqrt{2}\i a c - \ep^2 c^2} C(\balpha, \bx, 2^{-1/2},\mu ,c) dc.
\end{align*}
Let us call the first term on the right $H_\ep(P_1,P_2)$. By Theorem \ref{distthm}, we know that as $\ep \to 0$, $H_\ep(P_1,P_2)$ converges to $(1-e^{-2\sqrt{2}\pi(P_1+P_2)}) \pi e^{\frac{1}{4}+2P_1^2} \delta(P_1+P_2)$. But this is equal to zero. So we only have to worry about the second term on the right in the above display. By equation \eqref{repeat} and an application of the estimates from the proof of Lemma \ref{convergence}, we get
\begin{align*}
&(1-e^{2\pi a}) \int_{-\infty}^0 e^{-\sqrt{2}\i a c - \ep^2 c^2} C(\balpha, \bx, 2^{-1/2},\mu ,c) dc\\
&= (1-e^{2\pi a}) e^{2\alpha_1\alpha_2} \sum_{n=0}^\infty \frac{(-\mu )^n}{n!}a_n\int_{-\infty}^0 e^{\sqrt{2}c(n-\i a) - \ep^2 c^2 } dc.
\end{align*}
Now, for $n\ge 1$, 
\[
\int_{-\infty}^0 |e^{\sqrt{2}c(n-\i a) - \ep^2 c^2 }|  dc\le \int_{-\infty}^0 e^{\sqrt{2}cn } dc = \frac{1}{\sqrt{2} n}. 
\]
This bound allows us to apply the dominated convergence theorem and conclude that 
\begin{align*}
&\lim_{\ep\to 0} (1-e^{2\pi a}) e^{2\alpha_1\alpha_2} \sum_{n=1}^\infty \frac{(-\mu )^n}{n!}a_n\int_{-\infty}^0 e^{\sqrt{2}c(n-\i a) - \ep^2 c^2 } dc \\
&= (1-e^{2\pi a}) e^{2\alpha_1\alpha_2} \sum_{n=1}^\infty \frac{(-\mu )^n}{n!}a_n\int_{-\infty}^0 e^{\sqrt{2}c(n-\i a)  } dc\\
&= (1-e^{2\pi a}) e^{2\alpha_1\alpha_2} \sum_{n=1}^\infty \frac{(-\mu )^n}{\sqrt{2}n!(n-\i a)}a_n.
\end{align*}
This cancels with the corresponding sum from equation \eqref{horizontal}. Thus, we arrive at the conclusion that $\tilde{C}(\balpha,\bx, 2^{-1/2},\mu)$ is the limit (in the sense of distributions), as $\ep\to 0$, of 
\begin{align*}
J_\ep(P_1,P_2) := (1-e^{2\pi a}) e^{2\alpha_1\alpha_2}\biggl\{ -\frac{1}{\sqrt{2}\i a} - \int_{-\infty}^0 e^{-\sqrt{2}\i a c - \ep^2 c^2 } dc\biggr\}.
\end{align*}
We need the following lemma, which is proved in \textsection\ref{heavisidepf}.
\begin{lmm}\label{heaviside}
Let $\varphi:\R \to \R$ be a smooth function with compact support. Then 
\begin{align*}
\lim_{\ep\to 0}\int_{-\infty}^\infty \int_0^\infty \varphi(x) e^{\i tx-\ep^2 t^2} dt dx = \pi \varphi(0) + \i \int_0^\infty \frac{\varphi(x)-\varphi(-x)}{x} dx.
\end{align*}
Moreover, for any $\ep>0$,
\[
\biggl|\int_{-\infty}^\infty\int_0^{\infty}\varphi(x) e^{\i tx-\ep^2 t^2} dt dx\biggr| \le\int_0^{\infty}|\hat{\varphi}(t)|dt,
\]
where $\hat{\varphi}$ is the Fourier transform of $\varphi$.
\end{lmm}

Take any smooth function $\psi:\R^2 \to \R$ with compact support. Then 
\begin{align*}
&\iint J_\ep(P_1,P_2)\psi(P_1,P_2)dP_2 dP_1 \\
&= \frac{1}{\sqrt{2}}\iint J_\ep(P_1,-2^{-1/2}a - P_1)\psi(P_1,-2^{-1/2}a - P_1)da dP_1.
\end{align*}
Fix some $P_1$, and define 
\[
\varphi(a) := \exp\biggl(\frac{1}{4} + \frac{\i a}{2} + \sqrt{2}P_1 a + 2P_1^2\biggr)(1-e^{2\pi a}) \psi(P_1, -2^{-1/2}a -P_1).
\]
Note that $\varphi$ is a smooth function with compact support, and $\varphi(0)=0$. Thus, 
\begin{align*}
&\int J_\ep(P_1,-2^{-1/2}a - P_1)\psi(P_1,-2^{-1/2}a - P_1)da \\
&= -\int_{-\infty}^\infty \frac{\varphi(a)}{\sqrt{2}\i a} da - \int_{-\infty}^\infty\int_{-\infty}^0 \varphi(a) e^{-\sqrt{2}\i a c - \ep^2 c^2} dc da\\
&= -\int_{-\infty}^\infty \frac{\varphi(a)}{\sqrt{2}\i a} da - \frac{1}{\sqrt{2}}\int_{-\infty}^\infty\int_0^\infty \varphi(a) e^{\i t a  -\frac{1}{2} \ep^2 t^2} dt da.
\end{align*}
But by Lemma \ref{heaviside} and the facts that $\varphi$ is smooth and $\varphi(0)=0$, we get 
\begin{align*}
\lim_{\ep\to 0} \int_{-\infty}^\infty \int_0^\infty \varphi(a) e^{\i t a  -\frac{1}{2} \ep^2 t^2} dt da &= \pi \varphi(0) + \i \int_0^\infty \frac{\varphi(a)-\varphi(-a)}{a} da\\
&= \i \int_{-\infty}^\infty \frac{\varphi(a)}{a} da.
\end{align*}
Thus, we get
\begin{align*}
\lim_{\ep\to 0} \int J_\ep(P_1,-2^{-1/2}a - P_1)\psi(P_1,-2^{-1/2}a - P_1)da =0.
\end{align*}
Moreover, by the inequality from Lemma \ref{heaviside},
\begin{align*}
&\biggl|\int J_\ep(P_1,-2^{-1/2}a - P_1)\psi(P_1,-2^{-1/2}a - P_1)da\biggr| \\
&\le \frac{1}{\sqrt{2}}\int_{-\infty}^\infty \biggl|\frac{\varphi(a)}{a}\biggr| da + \frac{1}{\sqrt{2}}\int_{-\infty}^\infty |\hat{\varphi}(a)| da.
\end{align*}
This allows us to apply the dominated convergence theorem and deduce that 
\begin{align*}
\lim_{\ep\to 0}\iint J_\ep(P_1,P_2)\psi(P_1,P_2)dP_2 dP_1 =0,
\end{align*}
which completes the proof.

\subsection{Three-point function}
\subsubsection{Proof of Theorem \ref{threecthm}}\label{threecthmpf}
First, let us state the analogue of Lemma \ref{zeroform1one} and Lemma \ref{zeroform1two} for the three-point function.
\begin{lmm}\label{zeroform1three}
Let  $k=3$, $x_1=-e_3$, $x_2 = e_1$, and $x_3=e_3$. Suppose that $b\in (0,1)$ and $\Re(\alpha_j) >-\frac{1}{2b}$ for $j=1,2,3$. Then for any $\mu>0$ and $c\in \R$,
\begin{align*}
C(\balpha, \bx, b, \mu, c)  &= e^{2\sum_{1\le j<j'\le 3} \alpha_j\alpha_{j'}}2^{-2\alpha_2(\alpha_1+\alpha_3) } \biggl\{1+ \sum_{n=1}^\infty \frac{(-\mu e^{2bc})^n}{n!} a_n \biggr\},
\end{align*}
where
\begin{align}\label{cn3}
a_n &= 4^n e^{b^2 n(n+1-2w)-2n} 2^{-2bn\alpha_2}\int_{\C^n}\prod_{i=1}^n (1+|z_i|^2)^{-2b^2(n-w)} \notag\\
&\hskip 1in \cdot \prod_{i=1}^n (|z_i|^{4b\alpha_1}|z_i -1|^{4b\alpha_2}) \prod_{1\le i<j\le n}|z_i-z_j|^{4b^2} d^2z_1\cdots d^2z_n.
\end{align}
\end{lmm}
Since the proof is similar to the proofs of Lemma \ref{zeroform1one}  and Lemma \ref{zeroform1two}, it is shifted to \textsection\ref{zeroform1threepf}. Specializing the above lemma to the case $b=\alpha_2 = \frac{1}{\sqrt{2}}$, we get  the following analogue of Lemma~\ref{zeroform1}, Lemma \ref{oneform1}, and Lemma \ref{twoform1}. The proof is in \textsection\ref{threeform1pf}. %It turns out that if $\Re(c-a-b)>0$, then the above series converges absolutely even if $|z|=1$ (see \cite[Theorem 2.1.2]{andrewsetal99}, or Lemma \ref{hyperlemma} in the Appendix). This will suffice for us, since we will be exclusively using $z=-1$. In fact, under this condition, ${_2F_1}(a,b;c;z)$ is analytic in $(a,b,c)$. This is verified in Lemma \ref{hyperlemma}, together with some estimates that will be useful later.
\begin{lmm}\label{threeform1}
Let  $k=3$, $x_1=-e_3$, $x_2 = e_1$, and $x_3=e_3$. Suppose that $\Re(\alpha_1),\Re(\alpha_3)>-\frac{1}{\sqrt{2}}$, and $b = \alpha_2 =\frac{1}{\sqrt{2}}$. Let $w := -1-\sqrt{2}(\alpha_1+\alpha_2+\alpha_3)$. Then for any $\mu>0$ and $c\in \R$,  
\[
C(\balpha, \bx, 2^{-1/2}, \mu, c) =e^{\sqrt{2}(\alpha_1+\alpha_3) + 2\alpha_1\alpha_3}2^{-\sqrt{2}(\alpha_1+\alpha_3) } \sum_{n=0}^\infty \frac{(-\mu e^{\sqrt{2}c})^n}{n!}f(n),
\]
where $f$ is the function defined in the statement of Theorem \ref{threecthm}.%analytic function
%\begin{align*}%\label{fdef2}
%f(z) &:=\frac{(2\pi)^z e^{\frac{1}{2}z(z+1-2w)-2z} \Gamma(z+1)G(z+2+\sqrt{2}\alpha_1) G(z+2+\sqrt{2}\alpha_3)}{\Gamma(z-w)^zG(1+\sqrt{2}\alpha_1)G(1+\sqrt{2}\alpha_3)}\\
%&\qquad \cdot \biggl\{\frac{{_2F_1}(1, z-w-1; 1+\sqrt{2}\alpha_1; \frac{1}{2})}{2\Gamma(1+\sqrt{2}\alpha_1) \Gamma(z+1+\sqrt{2}\alpha_3)} - \frac{\sqrt{2}\alpha_3\, {_2F_1}(1,z-w-1;z+2+\sqrt{2}\alpha_1;\frac{1}{2})}{2\Gamma(z+2+\sqrt{2}\alpha_1)\Gamma(1+\sqrt{2}\alpha_3)}\biggr\}
%\end{align*}
%defined on the domain $\Omega$ where $z-w\notin (-\infty,0]$, $z+1$ is not a nonpositive integer, and $z+2+\sqrt{2}\alpha_1$ is not a nonpositive integer. %, where ${_2F_1}$ denotes the Gauss hypergeometric function.
%\[
%\Omega := \{z\in \C: z-w\notin (-\infty,0], \, z+1\notin \{0,-1,-2,\ldots\}, \, z+2+\sqrt{2}\alpha_1
%\]
\end{lmm} 
%Since the proof is similar to the proofs of Lemma \ref{zeroform1} and Lemma \ref{oneform1}, it is relegated to Subsection \ref{twoform1pf} of the Appendix.
%We remark that the domain of $f$ is determined from the following considerations. We need $z-w\notin (-\infty,0]$ because we need to define $\Gamma(z-w)^z$. The condition $z+1\notin \{0,-1,-2,\ldots\}$ arises due to the $\Gamma(z+1)$ factor in the numerator. The quantity inside the braces has no singularities, because the only way singularities can arise is if $1+\sqrt{2}\alpha_1$ is a nonpositive integer or $z+2+\sqrt{2}\alpha_1$ is a nonpositive integer. In both cases, both the numerator and the denominator would have simple poles, which would cancel out with each other.
Next, we need some estimates for $|f(z)|$, in analogy with similar estimates in Lemma~\ref{fzfinal}, Lemma~\ref{fzfinalone}, and Lemma~\ref{fzfinal2}. The proof is in \textsection\ref{fzfinal3pf}.
\begin{lmm}\label{fzfinal3}
In addition to the assumptions from Lemma \ref{threeform1}, suppose that $\Re(w)>-\frac{1}{2}$. Take any $z$ such that $z-w\ne 1$, and for some $\delta>0$, we have $\Re(z)\ge \Re(w)+\delta$ and $|\Re(z+1+\sqrt{2}\alpha_3)| \ge \delta$. Then 
\[
|f(z)|\le C_1 (1+|z|)^{C_2} e^{C_3|\Re(z)|} \exp\biggl(\frac{1}{2}\Re(z)\ln |z-w| - \frac{1}{2}\Im(z)\arg(z-w)\biggr),
\]
where $C_1,C_2,C_3$ depend only on $\alpha_1$, $\alpha_3$, and $\delta$.
\end{lmm}
For non-integer $x>\Re(w)$, define the contour integral
\begin{align}\label{fcdefthree}
F(x) := \frac{1}{2\pi \i} \int_{\Re(z) = x} \Gamma(-z) f(z)(\mu e^{\sqrt{2}c})^z dz,
\end{align}
where the contour goes from $x-\i \infty$ to $x+\i \infty$.
\begin{lmm}\label{fconvlmmthree}
For any $x\in \R \setminus \Z$ which satisfies the conditions $x> \Re(w)$, $x-\Re(w)\ne 1$, and $x\ne -\Re(1+\sqrt{2}\alpha_3)$, the integral defining $F(x)$ in equation~\eqref{fcdefthree} is absolutely convergent.
\end{lmm}
\begin{proof}
Since $\arg(x+\i y) \to \pm \frac{\pi}{2}$ as $y\to \pm \infty$, Lemma \ref{gammalmm} shows that $\Gamma(-x-\i y)$ decays exponentially in $|y|$ as $|y|\to \infty$. Similarly by Lemma \ref{fzfinal3}, $|f(x+\i y)|$ also decays exponentially in $|y|$ as $|y|\to \infty$. Since $x$ is not an integer, these functions remain bounded near $y=0$. Lastly, $|(\mu e^{\sqrt{2}c})^z| = (\mu e^{\sqrt{2}c})^x$ remains bounded. This shows that the integral in equation \eqref{fcdefthree} is absolutely convergent.
\end{proof}
The next lemma is the analogue of Lemma \ref{zerofinal0}, Lemma \ref{zerofinal0one}, and Lemma \ref{twofinal0}. 
The proof is in \textsection\ref{zerofinal0threepf}. 
\begin{lmm}\label{zerofinal0three}
For any $x_0\in (\Re(w),0)$ and $N\ge 1$, 
\begin{align*}%\label{contoureq}
F(x_0) = F(N+x_0) + \sum_{n=0}^{N-1}\frac{(-1)^n}{n!}f(n) (\mu e^{\sqrt{2}c})^n.
\end{align*}
\end{lmm}
Next, we have the analogue of Lemma~\ref{zerofinal}, Lemma \ref{zerofinalone} and Lemma \ref{zerofinaltwo}. The proof is in \textsection\ref{zerofinalthreepf}.
\begin{lmm}\label{zerofinalthree}
For any $x_0\in (\Re(w),0)$, $\lim_{N\to \infty} F(N+x_0) = 0$. Thus, by Lemma \ref{zerofinal0three},
\[
F(x_0) =  \sum_{n=0}^\infty \frac{(-\mu e^{\sqrt{2}c})^n}{n!}f(n). % = C(\balpha, \bx, 2^{-1/2}, \mu,c).
\]
\end{lmm}
%By Lemma \ref{threeform1}, this completes the proof of Theorem \ref{threecthm}.

We are now ready to complete the proof of Theorem \ref{threecthm}. We will use $C, C_1, C_2,\ldots$ to denote arbitrary positive constants that may depend only on $\alpha_1$ and $\alpha_3$, whose values may change from line to line. Let $a$ and $r$ denote the real and imaginary parts of $w$. Combining Lemma~\ref{threeform1} and Lemma~\ref{zerofinalthree}, we see that for any $x_0\in (a,0)$, 
\begin{align*}
C(\balpha, \bx, 2^{-1/2}, \mu,c) &= \frac{A}{2\pi}\int_{-\infty}^\infty \Gamma(-x_0-\i y)f(x_0+\i y)(\mu e^{\sqrt{2}c})^{x_0+\i y}dy,
\end{align*}
where 
\[
A:= e^{\sqrt{2}(\alpha_1+\alpha_3) + 2\alpha_1\alpha_3}2^{-\sqrt{2}(\alpha_1+\alpha_3) }.
\]
Consequently,
\begin{align}\label{cmucthree}
C(\balpha, \bx, 2^{-1/2}, \mu,c) = \frac{A}{2\pi}\lim_{x_0\downarrow a} \int_{-\infty}^\infty \Gamma(-x_0-\i y)f(x_0+\i y)(\mu e^{\sqrt{2}c})^{x_0+\i y}dy.
\end{align}
On the other, Theorem \ref{threecthm} claims that 
\begin{align}\label{zeroaltthree}
C(\balpha, \bx, 2^{-1/2}, \mu,c) &= \frac{A}{2\pi} \int_{-\infty}^\infty \Gamma(-w-\i y)f(w+\i y)(\mu e^{\sqrt{2}c})^{w+\i y}dy\notag \\
&= \frac{A}{2\pi} \int_{-\infty}^\infty \Gamma(-a-\i y)f(a+\i y)(\mu e^{\sqrt{2}c})^{a+\i y}dy.
\end{align}
Thus, all we have to show is that the limit in equation \eqref{cmucthree} can be moved inside the integral. To show this, take any $x_0\in (a,\frac{1}{2}a)$ and any $y\in \R$. Let $z := x_0+\i y$. We will use certain observations from the proof of Lemma \ref{fzfinal3} in \textsection\ref{fzfinal3pf}. Let $f_0$ and $h$ be as in that proof. Recall that $f_0$ is the same as the $f$ from Lemma \ref{fzfinal2}, but with $\alpha_1$ replaced by $\alpha_1' = \alpha_1+\frac{1}{\sqrt{2}}$, and $\alpha_2$ replaced by $\alpha_3$. Thus, following the same steps as in the proof of equation \eqref{fcibd}, we arrive at the inequality
\begin{align}
|f_0(z)| &\le C_1 (1+|y|)^{C_2} e^{-C_3|y|}(1+|y-r|^{a}).\label{fcibdthree}
\end{align}
%where $C_1,C_2,C_3$ depend only on $\alpha_1,\alpha_3$ (and crucially, not on $x_0$).
Next, consider the function $h$. Let
\[
\delta := \min\biggl\{\frac{1}{4},\, \frac{1}{2}(1+\sqrt{2}\alpha_3)\biggr\}.
\]
Let $b:= \Re(z-w-1)$, $t:= \Im(z-w+1)$, and $c:= 1+\sqrt{2}\alpha_1$. Then note that $b$ and $c-1$ are not nonpositive integers, since $b>-1$ and $b\ne 0$, and $c-1>-1$ and $c- 1=\sqrt{2}\alpha_1\ne 0$, where the last inequality holds because
\[
\Re(\sqrt{2}\alpha_1) = \Re(-w-2-\sqrt{2}\alpha_3) < \frac{1}{2} - 2 - \sqrt{2}\biggl(-\frac{1}{\sqrt{2}}\biggr)=-\frac{1}{2}.
\]
Also, $\Re(b) > -1$, 
\begin{align*}
\Re(1+b-c) &= \Re(z-w-1-\sqrt{2}\alpha_1)\\
&= \Re(z+1+\sqrt{2}\alpha_3) > \Re(z)+\delta > \Re(w)+\delta>-1+\delta,
\end{align*}
and since $\Re(z)<0$ and $\Re(1+\sqrt{2}\alpha_3)\ge 2\delta$, %since $\Re(1+\sqrt{2}\alpha_1)>2\delta$, 
\begin{align}\label{f20}
|\Re(1+b-c)| = |\Re(z+1+\sqrt{2}\alpha_3)|\ge \delta.
\end{align}
(Note that the last inequality above is an assumption of Lemma \ref{fzfinal3}, but not an assumption here; it follows from the definition of $\delta$ and the location of $z$.) Thus, we may apply Lemma~\ref{hyperasymp2} to get
\begin{align}\label{f21new}
\biggl|{_2F_1}\biggl(1, z-w-1; 1+\sqrt{2}\alpha_1; \frac{1}{2}\biggr)\biggr| &= \biggl|{_2F_1}\biggl(1, b+\i t; c; \frac{1}{2}\biggr)\notag \\
&\le C_1 (|b|+|t|+1)^{C_2} e^{C_3|b|}\notag \\
&\le C_1(1+|y|)^{C_2}.
%&\le C_1 (1+|z|)^{C_2} e^{C_3|\Re(z)|}.
\end{align}
Next, let $c:= \Re(z+2+\sqrt{2}\alpha_1)$, $t := \Im(z+2+\sqrt{2}\alpha_1)$, and $b := z-w-1-\i t$. Since $\delta \le \frac{1}{4}$ and $\Re(z)\in (a, \frac{1}{2}a)$, we have that for any nonnegative integer $n$,
\begin{align*}
|c+n| &\ge c+n = \Re(z+2+\sqrt{2}\alpha_1+n)\\
&> \Re(z+1+n)\ge \delta+n\ge \delta, 
\end{align*}
and $c\ge \Re(z+1)\ge \delta$. Lastly, note that 
\begin{align*}
|b| &= |z-w-1-\i\Im(z+2+\sqrt{2}\alpha_1)| \le C.
\end{align*}
Thus, by Lemma \ref{hyperasymp},
\begin{align}\label{f22new}
\biggl|{_2F_1}\biggl(1,z-w-1;z+2+\sqrt{2}\alpha_1;\frac{1}{2}\biggr)\biggr| &= \biggl|{_2F_1}(1,b+\i t; c+\i t;\frac{1}{2}\biggr)\biggr|\notag \\
&\le C_1 e^{C_2(|b|+|c|)}\le C_3.%\notag \\
%&\le C_1 e^{C_2|\Re(z)|}. 
\end{align}
Finally, note that for $j=1,3$, 
\[
\Re(z+2+\sqrt{2}\alpha_j) > \Re(w+2+\sqrt{2}\alpha_j)> \Re(w+1)\ge \delta.
\]
Thus, by Lemma \ref{gammalmm3} and the inequality \eqref{f20}, %assumption that $|\Re(z+1+\sqrt{2}\alpha_3)|\ge \delta$,
\begin{align}\label{f23new}
\biggl|\frac{\Gamma(z+1+\sqrt{2}\alpha_3)}{\Gamma(z+2+\sqrt{2}\alpha_1)}\biggr| &= \biggl|\frac{\Gamma(z+2+\sqrt{2}\alpha_3)}{(z+1+\sqrt{2}\alpha_3)\Gamma(z+2+\sqrt{2}\alpha_1)}\biggr|\notag \\
&\le C_1(1+|y|)^{C_2}. %}{|z+1+\sqrt{2}\alpha_3|}.
\end{align}
By the inequalities \eqref{f21new}, \eqref{f22new}, and \eqref{f23new}, we get
\begin{align*}
|h(z)| &\le C_1(1+|y|)^{C_2}.
\end{align*}
Combining this with equation \eqref{fcibdthree}, we get
\begin{align}
|f(z)| = |f_0(z)h(z)| &\le C_1 (1+|y|)^{C_2} e^{-C_3|y|}(1+|y-r|^{a}).\label{fcibdthree1}
\end{align}
Also, by equation \eqref{gammacibd}, we have 
\begin{align}\label{gammacibdthree}
|\Gamma(-z)|\le C_1(1+|y|)^{C_2} e^{-C_3|y|}. 
\end{align}
Since $a\in (-\frac{1}{2},0)$ and the constants do not depend on $x_0$, the above bounds suffice to apply the dominated convergence theorem and complete the proof.

\subsubsection{Proof of Theorem \ref{threethm}}\label{threethmpf}
Let $a$ and $r$ denote the real and imaginary parts of $w$. By Theorem \ref{threecthm}, we have 
\begin{align*}
&C_\ep(\balpha, \bx, 2^{-1/2}, \mu) \\
&= \frac{A}{2\pi}\int_{-\infty}^\infty \int_{-\infty}^\infty \Gamma(-a-\i (y+r))f(a+\i (y+r))\mu^{a+\i (y+r)}e^{\sqrt{2}\i c y- \ep^2 c^2}dy dc\\
&=  \frac{A}{2\pi}\int_{-\infty}^\infty \int_{-\infty}^\infty \Gamma(-a-\i y)f(a+\i y)\mu^{a+\i y}e^{\sqrt{2}\i c (y-r)- \ep^2 c^2}dy dc,
\end{align*}
where
\[
A := e^{\sqrt{2}(\alpha_1+\alpha_3) + 2\alpha_1\alpha_3}2^{-\sqrt{2}(\alpha_1+\alpha_3) }.
\]
The bounds \eqref{fcibdthree1} and \eqref{gammacibdthree} (which remain valid as we take $x_0\to a$) allow us to interchange the order of integration above. Since
\[
 \int_{-\infty}^\infty e^{\sqrt{2}\i c (y-r)- \ep^2 c^2}dc = \frac{\sqrt{\pi}}{\ep}\exp\biggl(-\frac{(y-r)^2}{2\ep^2}\biggr),
\]
this gives
\begin{align*}
&C_\ep(\balpha,\bx, 2^{-1/2}, \mu) = \frac{A}{2\sqrt{\pi}\ep}\int_{-\infty}^\infty \Gamma(-a-\i y)f(a+\i y)\mu^{a+\i y}\exp\biggl(-\frac{(y-r)^2}{2\ep^2}\biggr)dy\\
&= \frac{A}{2\sqrt{\pi}}\int_{-\infty}^\infty \Gamma(-a-\i (r+\ep u))f(a+\i(r+ \ep u))\mu^{a+\i (r+\ep u)}e^{-\frac{1}{2}u^2}du\\
&= \frac{A}{2\sqrt{\pi}}\int_{-\infty}^\infty \Gamma(-w-\i\ep u)f(w+\i\ep u)\mu^{w+\i \ep u}e^{-\frac{1}{2}u^2}du.
\end{align*}
But, we have (with a simple applications of Lemma \ref{gammalmm2}), 
\begin{align*}%\label{fdef2}
f(w + \i \ep u) &=(2\pi)^{w+\i \ep u} e^{\frac{1}{2}(w+\i \ep u)(-3+\i \ep u-w)+(w+\i \ep u)\log(\i \ep u)} \Gamma(w+1+\i \ep u) \\
&\qquad \cdot \frac{G(w+\i \ep u+2+\sqrt{2}\alpha_1) G(w+\i \ep u+2+\sqrt{2}\alpha_3)}{\Gamma(1+\i \ep u)^{w+\i \ep u}G(1+\sqrt{2}\alpha_1)G(1+\sqrt{2}\alpha_3)}\\
&\qquad \cdot \biggl\{\frac{{_2F_1}(1, -1+\i \ep u; 1+\sqrt{2}\alpha_1; \frac{1}{2})}{2\Gamma(1+\sqrt{2}\alpha_1) \Gamma(w+\i \ep u+1+\sqrt{2}\alpha_3)} \\
&\qquad - \frac{\sqrt{2}\alpha_3\, {_2F_1}(1,-1+\i \ep u;w+\i \ep u+2+\sqrt{2}\alpha_1;\frac{1}{2})}{2\Gamma(w+\i\ep u+2+\sqrt{2}\alpha_1)\Gamma(1+\sqrt{2}\alpha_3)}\biggr\}.
\end{align*}
Now note that 
\[
\log(\i \ep u) = \ln \ep + \ln |u| + \frac{\i \pi}{2}\sign(u) = \ln \ep + \log(\i u). 
\]
From these observations and the bounds provided by equations  \eqref{fcibdthree1} and \eqref{gammacibdthree}, it is now easy to apply the dominated convergence theorem and conclude that 
\begin{align*}
&\lim_{\ep\to 0} \ep^{-w} C_\ep(\balpha, \bx, 2^{-1/2}, \mu) \\
&= \frac{A (2\pi\mu )^w\Gamma(-w)\Gamma(w+1)e^{-\frac{1}{2}w(w+3)}G(-\sqrt{2}\alpha_1)G(-\sqrt{2}\alpha_3)}{2\sqrt{\pi} G(1+\sqrt{2}\alpha_1) G(1+\sqrt{2}\alpha_3)} \\
&\qquad \cdot \biggl\{\frac{{_2F_1}(1, -1; 1+\sqrt{2}\alpha_1; \frac{1}{2})}{2\Gamma(1+\sqrt{2}\alpha_1) \Gamma(-1-\sqrt{2}\alpha_1)} - \frac{\sqrt{2}\alpha_3\, {_2F_1}(1,-1;-\sqrt{2}\alpha_3;\frac{1}{2})}{2\Gamma(-\sqrt{2}\alpha_3)\Gamma(1+\sqrt{2}\alpha_3)}\biggr\}\\
&\qquad \cdot \int_{-\infty}^\infty(\i u)^we^{-\frac{1}{2}u^2}du.
\end{align*}
Now, note that 
\[
{_2F_1}(1,-1;c;z) = 1 - \frac{z}{c}.
\]
Thus, by the identity $\Gamma(z)\Gamma(1-z) = \pi/\sin(\pi z)$, 
\begin{align*}
&\frac{{_2F_1}(1, -1; 1+\sqrt{2}\alpha_1; \frac{1}{2})}{2\Gamma(1+\sqrt{2}\alpha_1) \Gamma(-1-\sqrt{2}\alpha_1)} - \frac{\sqrt{2}\alpha_3\, {_2F_1}(1,-1;-\sqrt{2}\alpha_3;\frac{1}{2})}{2\Gamma(-\sqrt{2}\alpha_3)\Gamma(1+\sqrt{2}\alpha_3)}\\
&= \frac{1+2\sqrt{2}\alpha_1}{4(1+\sqrt{2}\alpha_1)\Gamma(1+\sqrt{2}\alpha_1) \Gamma(-1-\sqrt{2}\alpha_1)} - \frac{1+2\sqrt{2}\alpha_3}{4\Gamma(-\sqrt{2}\alpha_3)\Gamma(1+\sqrt{2}\alpha_3)}\\
&= \frac{1+2\sqrt{2}\alpha_1}{4\Gamma(2+\sqrt{2}\alpha_1) \Gamma(-1-\sqrt{2}\alpha_1)}+ \frac{1+2\sqrt{2}\alpha_3}{4\Gamma(-1-\sqrt{2}\alpha_3)\Gamma(2+\sqrt{2}\alpha_3)}\\
&= \frac{(1+2\sqrt{2}\alpha_1)\sin(\sqrt{2}\pi \alpha_1)}{4\pi}+ \frac{(1+2\sqrt{2}\alpha_3)\sin(\sqrt{2}\pi \alpha_3)}{4\pi}.
\end{align*}
The proof of the main assertion of the theorem is now completed by invoking equation~\eqref{cosid} and combining the above calculations. 

To show that the limit is nonzero, first recall that by assumption, $\Re(\sqrt{2}\alpha_j) > -1$ for $j=1,3$. On the other hand, since $\Re(w)>-\frac{1}{2}$ and $w = -2-\sqrt{2}(\alpha_1+\alpha_3)$, we have 
\[
\Re(\sqrt{2}\alpha_1) + \Re(\sqrt{2}\alpha_3) < -\frac{3}{2}.
\]
Thus,
\[
\Re(\sqrt{2}\alpha_1) < -\frac{3}{2}-\Re(\sqrt{2}\alpha_3) < -\frac{3}{2}-(-1) = -\frac{1}{2}.
\]
We conclude that $-1< \Re(\sqrt{2}\alpha_1)<-\frac{1}{2}$. Similarly, $-1< \Re(\sqrt{2}\alpha_3)<-\frac{1}{2}$. In particular, neither $\sqrt{2}\alpha_1$ nor $\sqrt{2}\alpha_3$ can be a nonnegative integer. This implies that $G(-\sqrt{2}\alpha_1)$ and $G(-\sqrt{2}\alpha_3)$ are nonzero, since the zeros of $G$ are at the nonpositive integers. Since $w\in (-\frac{1}{2},0)$, $\cos(\frac{\pi w}{2})$ is nonzero. Thus, we conclude that the only way the limit can be zero is if 
\[
(1+2\sqrt{2}\alpha_1)\sin(\sqrt{2}\pi \alpha_1)+ (1+2\sqrt{2}\alpha_3)\sin(\sqrt{2}\pi \alpha_3)= 0.
\]
But since $-1< \sqrt{2}\alpha_j <-\frac{1}{2}$ for $j=1,3$, the following lemma shows that the above sum has  strictly positive real part, and therefore, is nonzero. This completes the proof.
\begin{lmm}
Take any $z\in \C$ with $-1< \Re(z)< -\frac{1}{2}$. Then $\Re((1+2z)\sin(\pi z)) > 0$.
\end{lmm}
\begin{proof}
Let $x$ and $y$ be the real and complex parts of $z$. Then 
\begin{align*}
\sin(\pi z) &= \sin (\pi x + \i \pi y) \\
&= \sin(\pi x) \cos(\i \pi y) + \cos(\pi x) \sin(\i \pi y)\\
&= \sin (\pi x) \cosh(\pi y) + \i \cos(\pi x) \sinh(\pi y).
\end{align*}
Thus,
\begin{align*}
\Re((1+2z)\sin(\pi z)) &= (1+2x) \sin (\pi x) \cosh(\pi y) - 2y \cos(\pi x) \sinh(\pi y).
\end{align*}
Since $-1< x< -\frac{1}{2}$, we have $1+2x< 0$, $\sin (\pi x) < 0$, and $\cos (\pi x) < 0$. Also, irrespective of the value of $y$, $\cosh(\pi y) > 0$ and $y\sinh(\pi y) \ge 0$. Thus, the above expression is strictly positive.
\end{proof}

\subsubsection{Proof of Theorem \ref{distthm3}}\label{distthm3pf}
From Lemma \ref{zeroform1three}, we have
\begin{align}\label{repeatthree}
C(\balpha, \bx, 2^{-1/2}, \mu, c) =A\sum_{n=0}^\infty \frac{(-\mu e^{\sqrt{2}c})^n}{n!}a_n,
\end{align}
where 
\[
A := e^{\sqrt{2}(\alpha_1+\alpha_3) + 2\alpha_1\alpha_3}2^{-\sqrt{2}(\alpha_1+\alpha_3) }.%e^{2\sum_{1\le j<j'\le 3} \alpha_j\alpha_{j'}}2^{-2\alpha_2(\alpha_1+\alpha_3) }
\]
and $a_n$ are the coefficients defined in equation \eqref{cn3} with $\alpha_2 = \frac{1}{\sqrt{2}}$. Although Lemma \ref{zeroform1three} is stated for $c\in \R$, it is easy to see from the proof that the formula continues to be valid for any complex $c$. By the inequality \eqref{wagner2}, we can integrate $c$ from $0$ to $\sqrt{2}\pi \i$ by moving the integral inside the infinite sum, to get
\begin{align*}%\label{horizontalthree}
 -\i \int_0^{\sqrt{2}\pi} e^{-\sqrt{2}\i wt} C(\balpha, \bx, 2^{-1/2}, \mu, \i t) dt &= -\i A\sum_{n=0}^\infty \frac{(-\mu )^n}{n!}a_n\int_0^{\sqrt{2}\pi} e^{\sqrt{2}\i (n-w)t} dt\notag\\
&= A\sum_{n=0}^\infty \frac{(-\mu )^n(1-e^{2\pi \i (n-w)})}{\sqrt{2} n! (n-w)}a_n\notag \\
&= A (1-e^{-2\pi \i w}) \sum_{n=0}^\infty \frac{(-\mu )^n}{\sqrt{2} n! (n-w)}a_n.
\end{align*}
By inequality \eqref{wagner2}, it follows that this quantity is finite. Moreover, it has no dependence on $\ep$. Since $\Re(w)<0$, we see that the product of the above quantity and $\ep^{-w}$ tends to zero as $\ep \to 0$. Next, note that
\begin{align*}
&(1-e^{-2\pi \i w})\int_0^\infty e^{-\sqrt{2}wc - \ep^2 c^2} C(\balpha, \bx, 2^{-1/2},\mu ,c) dc\\
&= (1-e^{-2\pi\i w}) \int_{-\infty}^\infty e^{-\sqrt{2}w c - \ep^2 c^2} C(\balpha, \bx, 2^{-1/2},\mu ,c) dc \\
&\qquad \qquad - (1-e^{-2\pi \i w}) \int_{-\infty}^0 e^{-\sqrt{2}w c - \ep^2 c^2} C(\balpha, \bx, 2^{-1/2},\mu ,c) dc.
\end{align*}
Theorem \ref{threethm} tells us the behavior of the first term as $\ep \to 0$. We will now show that the second term multiplied by $\ep^{-w}$ tends to zero as $\ep\to 0$. Since $\Re(w)<0$, it suffices to show that the second term remains bounded as $\ep \to 0$. To this end, first note that by equation~\eqref{repeatthree} and an application of the bound \eqref{wagner2} (to move the integration inside the infinite sum), we get
\begin{align*}%\label{horizontalone2}
&(1-e^{-2\pi \i w}) \int_{-\infty}^0 e^{-\sqrt{2}w c - \ep^2 c^2} C(\balpha, \bx, 2^{-1/2},\mu ,c) dc\notag\\
&= (1-e^{-2\pi \i w}) \sum_{n=0}^\infty \frac{(-\mu )^n}{n!}a_n\int_{-\infty}^0 e^{\sqrt{2}c(n-w) - \ep^2 c^2 } dc.
\end{align*}
But since $\Re(w)<0$, we have that for any $n\ge 0$,
\begin{align*}
\int_{-\infty}^0 |e^{\sqrt{2}c(n-w) - \ep^2 c^2 }| dc &\le \int_{-\infty}^0 e^{\sqrt{2}c(n-\Re(w)) } dc = \frac{1}{\sqrt{2}(n- \Re(w))}. 
\end{align*}
By another application of the inequality \eqref{wagner2}, this completes the proof.

\section*{Acknowledgements}
The author is indebted to Edward Witten for introducing him to this body of problems, and numerous helpful discussions. The author also thanks Ioannis Tsiares and Beatrix M\"uhlmann for helpful correspondence, in particular for pointing out recent bootstrap results at rational central charge and for discussions of the contour/saddle-point interpretation of the $b=2^{-1/2}$ cancellation. The author's research was supported in part by NSF grant DMS-2450608 and the Simons Collaboration grant on `Probabilistic Paths to QFT'.

\bibliographystyle{abbrvnat}

\addcontentsline{toc}{section}{References}

\bibliography{myrefs}

\appendix

\setcounter{section}{0}

\refstepcounter{section}   % create section A internally
\section*{Appendix}        % print without the "A"
\addcontentsline{toc}{section}{Appendix}

%\section{Appendix} 
%\addtocontents{toc}{\protect\setcounter{tocdepth}{1}}

%\begin{proof}[Proof of Lemma \ref{convergence}]

\subsection{Hypergeometric function toolbox}%{Some properties of the hypergeometric function}
The following well-known result shows that the hypergeometric function ${_2F_1}(a,b;c;z)$ is well-defined through its usual power series expansion and analytic in $(a,b,c,z)$ when $|z|< 1$ and $c$ is not a nonpositive integer. % and $a,b$ are not nonpositive integers. 
\begin{lmm}\label{hyperlemma}
The power series
\[
{_2F_1}(a,b;c;z) := \sum_{n=0}^\infty \frac{(a)_n(b)_n}{(c)_nn!} z^n
\]
is absolutely convergent and defines an analytic function of $(a,b,c,z)$ in the region 
\[
\Omega:= \{(a,b,c,z)\in \C^4: c\notin\{0,-1,-2,\ldots\},\, |z|<1\}.
\] 
\end{lmm}
\begin{proof}
Take any $(a,b,c,z)\in \Omega$. For each $n$, let
\[
q_n := \frac{(a)_n(b)_n}{(c)_nn!}z^n.
\]
Since $c$ is not a nonpositive integer, each $q_n$ is well-defined, and 
\[
\frac{q_{n+1}}{q_n} = \frac{(a+n)(b+n)z}{(c+n)(n+1)}.
\]
Since $c$ is not a nonpositive integer, we also have
\begin{align*}
\biggl|\frac{(a+n)(b+n)}{(c+n)(n+1)}-1\biggr| &= \frac{|ab-c+ (a+b-c-1)n|}{|c+n|(n+1)}\le \frac{C}{n},
\end{align*}
where $C$ depends continuously only on $a,b,c$.  Thus, 
\[
\biggl|\frac{q_{n+1}}{q_n}\biggr| \le \biggl(1+\frac{C}{n}\biggr) |z|. 
\]
Since $|z|<1$, this bound suffices to show that the analytic function
\[
f_N(a,b,c,z) := \sum_{n=0}^N \frac{(a)_n(b)_n}{(c)_nn!} z^n
\]
converges uniformly on any compact subset of $\Omega$ as $N\to \infty$. This proves the claim.
\end{proof}
The next result is a special case of a theorem of Euler.
\begin{lmm}\label{eulerlmm}
Suppose that $\Re(c) > \Re(b)>0$ and $|z|<1$. Then 
\[
{_2F_1}(a,b;c;z) = \frac{\Gamma(c)}{\Gamma(b)\Gamma(c-b)} \int_0^1 t^{b-1} (1-t)^{c-b-1} (1-zt)^{-a} dt,
\]
where $(1-zt)^{-a} := \exp(-a\log(1-zt))$, where $\log$ denotes the analytic branch of logarithm on $\C\setminus(-\infty,0]$ that is real-valued on $(0,\infty)$.
\end{lmm}
\begin{proof}
Since $|zt|<1$, we can expand
\[
(1-zt)^{-a} = \sum_{n=0}^\infty {-a \choose n} (-zt)^n = \sum_{n=0}^\infty \frac{(a)_n(zt)^n}{n!}.
\]
From the given conditions and the above identity, it is now easy to apply the dominated convergence theorem and get
\begin{align*}
&\frac{\Gamma(c)}{\Gamma(b)\Gamma(c-b)} \int_0^1 t^{b-1} (1-t)^{c-b-1} (1-zt)^{-a} dt\\
&=  \frac{\Gamma(c)}{\Gamma(b)\Gamma(c-b)} \sum_{n=0}^\infty \frac{(a)_nz^n}{n!}\int_0^1 t^{n+b-1} (1-t)^{c-b-1} dt\\
&=  \frac{\Gamma(c)}{\Gamma(b)\Gamma(c-b)} \sum_{n=0}^\infty \frac{(a)_nz^n}{n!}\frac{\Gamma(n+b)\Gamma(c-b)}{\Gamma(n+c)}.
\end{align*}
But $\Gamma(n+x)/\Gamma(x) = (x)_n$ for any $x$. This completes the proof.
\end{proof}
A consequence of Euler's lemma is the following special case of an identity of Pfaff.
\begin{lmm}\label{pfafflmm}
Suppose that $c\in \C\setminus\{0,-1,-2,\ldots\}$, $|z|<1$, and $|z|<|z-1|$. Then for any $a,b\in \C$, 
\[
{_2F_1}(a,b;c;z) = (1-z)^{-a} {_2F_1}\biggl(a, c-b;c; \frac{z}{z-1}\biggr).
\]
\end{lmm}
\begin{proof}
First, suppose that $\Re(c) > \Re(b)>0$. 
Take any $|z|<1$ and $t\in [0,1]$. Since $|zt|<1$ and $|z|<1$, we have $\Re(1-zt) > 0$ and $\Re(1-z)>0$. Since the reciprocal of an element of the open right half-plane is also in the open right half-plane, we have $\Re(1/(1-z))>0$. Now, the branch of logarithm we chose in the statement of Lemma \ref{eulerlmm} has the property that $\log(uv) = \log u + \log v$ and $\log(1/u) = -\log u$ for any $u,v$ in the open right half-plane. Thus, 
\[
\log\frac{1-zt}{1-z} = \log(1-zt) - \log (1-z). 
\]
Consequently,
\begin{align*}
(1-zt)^{-a} &= \exp(-a\log(1-zt))\\
&= \exp\biggl(-a\log \frac{1-zt}{1-z} - a\log (1-z)\biggr) = \biggl(\frac{1-zt}{1-z}\biggr)^{-a} (1-z)^{-a}.
\end{align*}
Thus, Lemma \ref{eulerlmm} can be rewritten as
\begin{align*}
{_2F_1}(a,b;c;z) &= \frac{(1-z)^{-a}\Gamma(c)}{\Gamma(b)\Gamma(c-b)} \int_0^1 t^{b-1} (1-t)^{c-b-1}  \biggl(\frac{1-zt}{1-z}\biggr)^{-a} dt.
\end{align*}
Making the change of variable $s=1-t$, this becomes
\begin{align*}
{_2F_1}(a,b;c;z) &= \frac{(1-z)^{-a}\Gamma(c)}{\Gamma(b)\Gamma(c-b)} \int_0^1 (1-s)^{b-1} s^{c-b-1}  \biggl(1-\frac{zs}{z-1}\biggr)^{-a} ds.
\end{align*}
Since $|z/(z-1)|<1$, this allows us to once again apply Lemma \ref{eulerlmm} to get the desired result under the assumption that $\Re(c) > \Re(b)>0$. But by Lemma \ref{eulerlmm}, both sides of the claimed identity are analytic functions of $(a,b,c)$ in the region where $c$ is not a nonpositive integer (for a fixed $z$). Since they are equal in the region where $\Re(c)>\Re(b)>0$, we conclude that they are equal everywhere in this domain. 
\end{proof}
We will now use the above results to prove the following lemma, which will be crucial in the analysis of the three-point function.
\begin{lmm}\label{hyperlmm}
For any positive integer $n$ and any $a,b\in \C\setminus\{0,-1,-2,\ldots\}$, 
\begin{align*}
\sum_{j=0}^n \frac{1}{\Gamma(j+a)\Gamma(n-j+b)} &= \frac{{_2F_1}(1, n+a+b-1; a; \frac{1}{2})}{2\Gamma(a) \Gamma(n+b)}\\
&\qquad - \frac{(b-1)\, {_2F_1}(1,n+a+b-1;n+a+1;\frac{1}{2})}{2\Gamma(n+a+1)\Gamma(b)}.
\end{align*}
\end{lmm}
\begin{proof}
Fix $n$, and define the analytic function
\[
h(a,b,z) := \sum_{j=0}^n \frac{z^j}{\Gamma(j+a)\Gamma(n-j+b)} 
\]
for $a,b\in \C \setminus\{0,-1,-2,\ldots\}$ and $z\in \C$. 
Note that 
\[
\Gamma(j+a) = (a)_j\Gamma(a),
\]
and 
\begin{align*}
\Gamma(n-j+b) &= \frac{\Gamma(n+b)}{(n+b-1)(n+b-2)\cdots(n+b-j)}\\
&= (-1)^j \frac{\Gamma(n+b)}{(1-n-b)(2-n-b)\cdots(j-n-b)} = (-1)^j\frac{\Gamma(n+b)}{(1-n-b)_j}.
\end{align*}
Combining, we have 
\begin{align*}
h(a,b,z) &=\frac{1}{\Gamma(a)\Gamma(n+b)}\sum_{j=0}^n \frac{(-1)^j(1-n-b)_j}{(a)_j}z^j\\
&= \frac{1}{\Gamma(a)\Gamma(n+b)}\sum_{j=0}^n \frac{(1)_j(1-n-b)_j}{(a)_jj!}(-z)^j.
\end{align*}
%where the right side is interpreted as zero of $n+b\in \{0,-1,-2,\ldots\}$. 
Take any $z$ with $|z|<1$. Then by Lemma \ref{hyperlemma} and the assumption that $a$ is not a nonpositive integer, we have 
\[
\sum_{j=0}^\infty \frac{(1)_j(1-n-b)_j}{(a)_jj!}(-z)^j = {_2F_1}(1, 1-n-b; a; -z),
\]
where the left side is absolutely convergent. For the same reason, we also have 
\begin{align*}
\sum_{j=n+1}^\infty \frac{(1)_j(1-n-b)_j}{(a)_jj!}(-z)^j &= \sum_{j=n+1}^\infty \frac{(1-n-b)_j}{(a)_j}(-z)^j\\
&= \frac{(1-n-b)_{n+1}(-z)^{n+1}}{(a)_{n+1}} \sum_{k=0}^\infty \frac{(2-b)_k}{(a+n+1)_k}(-z)^k\\
&=\frac{(1-n-b)_{n+1}(-z)^{n+1}}{(a)_{n+1}} \sum_{k=0}^\infty \frac{(1)_k(2-b)_k}{(a+n+1)_kk!}(-z)^k\\
&= \frac{(1-n-b)_{n+1}(-z)^{n+1}}{(a)_{n+1}} {_2F_1}(1,2-b;a+n+1;-z).
\end{align*}
Finally, note that since $b$ is not a nonpositive integer,
\begin{align*}
\frac{(1-n-b)_{n+1}(-z)^{n+1}}{\Gamma(a) \Gamma(n+b)(a)_{n+1}}  &=\frac{(n+b-1)(n+b-2)\cdots(b-1)}{\Gamma(a+n+1) \Gamma(n+b)} z^{n+1}\\
&= \frac{(b-1)z^{n+1}}{\Gamma(n+a+1)\Gamma(b)}.
\end{align*}
Combining, we get that for any $z$ with $|z|<1$ and any $a,b\in \C \setminus\{0,-1,-2,\ldots\}$, 
\begin{align*}
h(a,b,z) &= \frac{{_2F_1}(1, 1-n-b; a; -z)}{\Gamma(a) \Gamma(n+b)} - \frac{(b-1)z^{n+1}\, {_2F_1}(1,2-b;a+n+1;-z)}{\Gamma(n+a+1)\Gamma(b)}.
\end{align*}
The result now follows from Lemma \ref{pfafflmm}, since neither $a$ nor $a+n+1$ is a nonpositive integer, and then taking $z\to 1$ through the open unit disk (and noting that $|-z/(-z-1)|<1$ when $z$ is close enough to $1$). 
\end{proof}

We will need the following lemma about a certain kind of asymptotic behavior of the hypergeometric function.
\begin{lmm}\label{hyperasymp}
Take any $b\in \C$, $c\in \R \setminus\{0,-1,-2,\ldots\}$, and $t\in \R$. Let $\delta$ and $K$ be positive real numbers such that $|c+n|\ge \delta$ for every nonnegative integer $n$, and $c\ge -K$. Then there is a  universal constant $C_1$ and a positive constant $C_2$ depending only on $\delta$ and $K$ such that 
\[
\biggl|{_2F_1}\biggl(1, b+\i t;c+\i t;\frac{1}{2}\biggr)\biggr| \le C_1 e^{C_2(|b|+|c|)}. 
\]
\end{lmm}
\begin{proof}
Note that for any nonnegative integer $n$,
\begin{align*}
\biggl|\frac{b+\i t + n}{c+ \i t + n}-1\biggr| &= \frac{|b-c|}{|c+\i t + n|} \le \frac{|b-c|}{|c+n|}.%=  \frac{|b-c|}{|n+c^+-c^-|}
\end{align*}
Thus, 
\begin{align}\label{bceq}
\biggl|\frac{b+\i t + n}{c+ \i t + n}\biggr| &\le 1 + \frac{|b-c|}{|c+n|}\notag \\
&= \frac{|c+n|+|b-c|}{|c+n|}\le \frac{|b| + 2|c|+n}{|c+n|}. 
\end{align}
Now, if $n+1\ge 2(K+1)$, then 
\[
|c+n|\ge c+n \ge n+1-K-1\ge n+1 - \frac{n+1}{2}=\frac{n+1}{2}.
\]
On the other hand, if $n+1< 2(K+1)$, then 
\[
|c+n| \ge \delta \ge \frac{\delta (n+1)}{2(K+1)}.
\]
Thus, for every nonnegative integer $n$,
\begin{align*}
|c+n|\ge C(n+1),
\end{align*}
where 
\[
C := \min\biggl\{\frac{\delta}{2(K+1)}, \, \frac{1}{2}\biggr\}.
\]
Using this in equation \eqref{bceq}, we get 
\begin{align}\label{bceq2}
\biggl|\frac{b+\i t + n}{c+ \i t + n}\biggr| &\le \frac{|b| + 2|c|+n}{C(n+1)}.
\end{align}
Now, if $n\ge k := \lceil20(|b|+|c|)\rceil$, then by equation \eqref{bceq},
\begin{align*}
\biggl|\frac{b+\i t + n}{c+ \i t + n}\biggr| &\le \frac{2(|b|+|c|) + n}{n - |c|}\\
&\le \frac{\frac{n}{10} + n}{n - \frac{n}{20}} = \frac{22}{19}.
\end{align*}
On the other hand, if $n < k$, then by equation \eqref{bceq2},
\begin{align*}
\biggl|\frac{b+\i t + n}{c+ \i t + n}\biggr| &\le \frac{22(|b|+|c|)}{C(n+1)} = \frac{22}{19}\cdot \frac{19(|b|+|c|)}{C(n+1)}.
\end{align*}
Since $C\le \frac{1}{2}$ and $|c|\ge \delta$, we have that for any $n< k$,
\begin{align*}
\frac{19(|b|+|c|)}{C(n+1)} &\ge \frac{38(|b|+|c|)}{n+1}\ge \frac{38(|b|+|c|)}{k}\\
&\ge \frac{38(|b|+|c|)}{20(|b|+|c|)+1}= \frac{19}{10}\cdot \frac{|b|+|c|}{|b|+|c|+\frac{1}{20}}\ge \frac{19}{10}\cdot \frac{\delta}{\delta+\frac{1}{20}} =:C_1.
\end{align*}
Thus, for any $n$,
\begin{align*}
\biggl|\frac{(b+\i t)_n}{(c+\i t)_n}\biggr| &= \prod_{j=0}^{n-1} \biggl|\frac{b+\i t + j}{c+\i t + j}\biggr|\\
&\le \biggl(\frac{22}{19}\biggr)^n\prod_{j=0}^{\min\{n-1, k-1\}}\frac{19(|b|+|c|)}{C(j+1)}\\
&\le \biggl(\frac{22}{19}\biggr)^n \frac{(19(|b|+|c|))^k}{C^kk!}C_1^{-\max\{k-n, 0\}}\\
&\le \biggl(\frac{22}{19}\biggr)^n \frac{(19(|b|+|c|))^k}{C^kk!}C_2^k,
\end{align*}
where $C_2 := 1+C_1^{-1}$. 
This gives
\begin{align*}
\biggl|{_2F_1}\biggl(1, b+\i t;c+\i t;\frac{1}{2}\biggr)\biggr| &\le \sum_{n=0}^\infty\biggl|\frac{(b+\i t)_n}{(c+\i t)_n}\biggr| 2^{-n}\\
&\le\frac{C_2^k(19(|b|+|c|))^k}{C^kk!}\sum_{n=0}^\infty  \biggl(\frac{22}{19}\biggr)^n  2^{-n}.
\end{align*}
The infinite sum of the right is a finite universal constant, and  since $19(|b|+|c|)\le k$, 
\begin{align*}
\frac{(19(|b|+|c|))^k}{k!} \le \frac{k^k}{k!}\le e^k,
\end{align*}
where the last step follows from the standard lower bound $k! \ge (k/e)^k$. Since $k\le 20(|b|+|c|)+1$, this completes the proof.
\end{proof}
We will also need a version of Lemma \ref{hyperasymp} where instead of $c+\i t$ we only have $c$. For that, we will use the following well-known identity for the hypergeometric function~\cite[Corollary 2.3.3]{andrewsetal99}:
\begin{align}\label{f2linear}
&{_2F_1}(a,b;c;z) \notag \\
&= \frac{\Gamma(c)\Gamma(c-a-b)}{\Gamma(c-a)\Gamma(c-b)}\, {_2F_1}(a, b; a+b+1-c; 1-z)\notag \\
&\qquad + \frac{\Gamma(c)\Gamma(a+b-c)}{\Gamma(a)\Gamma(b)}(1-z)^{c-a-b}\,{_2F_1}(c-a, c-b; 1+c-a-b; 1-z).
\end{align}
This holds, for example, when $|z|<1$, the numbers $a,b,c,c-a,c-b$ are not nonpositive integers, and $a+b-c$ is not an integer. 
\begin{lmm}\label{hyperasymp2}
Take any $b,c\in \C$ such that $\Re(b) > -1$, $b$ and $c-1$ are not nonpositive integers, and $1+b-c$ is not an integer. Further, assume that there is some $\delta>0$ such that $\Re(1+b-c)\ge-1+ \delta$, and $|\Re(1+b-c)|\ge \delta$. Then there are constants $C_1,C_2,C_3$ depending only on $c$ and $\delta$, such that for any $t\in \R$,
\[
\biggl|{_2F_1}\biggl(1,b+\i t;c;\frac{1}{2}\biggr)\biggr|\le C_1(|b|+|t|+1)^{C_2} e^{C_3|b|}. 
\]
\end{lmm}
\begin{proof}
Throughout this proof, $C, C_1,C_2,\ldots$ will denotes constants depending only on $c$ and $\delta$, whose values may change from line to line. 
When $a=1$ and $z=\frac{1}{2}$, the identity \eqref{f2linear} reduces to
\begin{align*}
{_2F_1}\biggl(1,b;c;\frac{1}{2}\biggr) &= \frac{\Gamma(c)\Gamma(c-1-b)}{\Gamma(c-1)\Gamma(c-b)} \, {_2F_1}\biggl(1, b; 2+b-c; \frac{1}{2}\biggr) \\
&\qquad + \frac{\Gamma(c)\Gamma(1+b-c)}{\Gamma(b)}2^{1+b-c}\, {_2F_1}\biggl(c-1, c-b; c-b; \frac{1}{2}\biggr)\\
&= \frac{c-1}{c-b-1} \, {_2F_1}\biggl(1, b; 2+b-c; \frac{1}{2}\biggr) \\
&\qquad + \frac{\Gamma(c)\Gamma(1+b-c)}{\Gamma(b)}2^{1+b-c}\, {_2F_1}\biggl(c-1, c-b; c-b; \frac{1}{2}\biggr).
\end{align*}
(This holds because by the assumptions of the lemma and the chosen values of $a$ and $z$, we have $|z|<1$, the numbers $a,b,c,c-a,c-b$ are not nonpositive integers, and the number $a+b-c$ is not an integer.) 
Now, note that 
\begin{align*}
{_2F_1}\biggl(c-1, c-b; c-b; \frac{1}{2}\biggr) &= \sum_{n=0}^\infty \frac{(c-1)_n}{n!}2^{-n}\\
&= \sum_{n=0}^\infty {1-c\choose n} (-1)^n 2^{-n} = \biggl(1-\frac{1}{2}\biggr)^{1-c} = 2^{c-1}.
\end{align*}
Thus,
\begin{align*}
{_2F_1}\biggl(1,b;c;\frac{1}{2}\biggr) &= \frac{c-1}{c-b-1} \, {_2F_1}\biggl(1, b; 2+b-c; \frac{1}{2}\biggr) \\
&\qquad + \frac{\Gamma(c)\Gamma(1+b-c)}{\Gamma(b)}2^{b}.
\end{align*}
Replacing $b$ by $b+\i t$ above, we get
\begin{align*}
{_2F_1}\biggl(1,b+\i t;c;\frac{1}{2}\biggr) &= \frac{c-1}{c-b-1-\i t} \, {_2F_1}\biggl(1, b+\i t; 2+b+\i t-c; \frac{1}{2}\biggr) \\
&\qquad + \frac{\Gamma(c)\Gamma(1+b+\i t-c)}{\Gamma(b+\i t)}2^{b+\i t}.
\end{align*}
Since $\Re(b)> -1$, $\Re(1+b-c)\ge -1+\delta$, and $|\Re(1+b-c)|\ge \delta$, Lemma \ref{gammalmm3} gives 
\begin{align*}
\biggl|\frac{\Gamma(1+b+\i t-c)}{\Gamma(b+\i t)}\biggr|&= \biggl|(b+\i t)(b+1+\i t)\frac{\Gamma(2+b+\i t-c)}{(1+b+\i t - c)\Gamma(2+b+\i t)}\biggr|\\
&\le \frac{C(|b|^2+|t|^2+1)}{|\Re(1+b-c)|} \biggl|\frac{\Gamma(2+b+\i t-c)}{\Gamma(2+b+\i t)}\biggr|\le C_1 e^{C_2\ln(1+|b|+|t|)}.
\end{align*}
Also,
\begin{align*}
\biggl|\frac{c-1}{c-b-1-\i t}\biggr|&\le \frac{|c-1|}{|\Re(c-b-1)|} \le C.
\end{align*}
Thus, we get 
\begin{align}\label{f2ineq0}
\biggl|{_2F_1}\biggl(1,b+\i t;c;\frac{1}{2}\biggr) \biggr| &\le C_1\bigg| \, {_2F_1}\biggl(1, b+\i t; 2+b+\i t-c; \frac{1}{2}\biggr)\biggr| \notag\\
&\qquad + C_2(|b|+|t|+1)^{C_3}e^{C_4 |b|}.
\end{align}
Let $c' := \Re(2+b-c)$, $t' := t + \Im (2+b-c)$, and $b' := b - \i \Im (2+b-c)$. Then 
\begin{align*}
2+b+\i t - c = c' + \i t', \ \ \ b+ \i t = b' + \i t'. 
\end{align*}
Thus,
\begin{align*}
{_2F_1}\biggl(1, b+\i t; 2+b+\i t-c; \frac{1}{2}\biggr) &= {_2F_1}\biggl(1, b'+\i t'; c'+\i t'; \frac{1}{2}\biggr).
\end{align*}
Since $c'$ and $t'$ are real, and $c'$ is not a nonpositive integer, the above identity allows us to invoke Lemma \ref{hyperasymp} and deduce that 
\begin{align*}
\bigg| \, {_2F_1}\biggl(1, b+\i t; 2+b+\i t-c; \frac{1}{2}\biggr)\biggr| &\le C_1 e^{C_2|b|}. 
\end{align*}
Plugging this into equation \eqref{f2ineq0} completes the proof.
\end{proof}

\subsection{Convergence and basic sphere estimates}
\subsubsection{Proof of Lemma \ref{convergence}}\label{convergencepf}
It is known by a result of~\citet[Theorem 2]{wagner90} that for any $n\ge 2$ and any $x_1,\ldots,x_n \in \S^2$, %(see also \cite[Theorem 4]{wagner89}), 
\begin{align}\label{wagner}
\sum_{1\le i<j \le n} G_{\S^2}(x_i,x_j) \ge -\frac{1}{4}n\ln n - Cn,
\end{align}
where $C$ is a positive constant that does not depend on $n$. By the inequality \eqref{wagner} and the formula  \eqref{gform} for $G_{\S^2}$, we get
\begin{align*}
&\biggl|\exp\biggl(-4b\sum_{j=1}^k\sum_{l=1}^n \alpha_j G_{\S^2}(x_j, y_l)-4b^2 \sum_{1\le l < l'\le n} G_{\S^2}(y_l,y_{l'}) \biggr)\biggr|\\
&\le C_1^n n^{b^2 n} \exp\biggl(-4b\sum_{j=1}^k\sum_{l=1}^n\Re( \alpha_j)G_{\S^2}(x_j, y_l)\biggr)\le C_2^n n^{b^2n} \prod_{j=1}^k\prod_{l=1}^n \|x_j - y_l\|^{4b \Re(\alpha_j)},
\end{align*}
where $C_1$ is a universal constant and $C_2$ depends only on $b$ and $\alpha_1,\ldots,\alpha_k$. Now,
\begin{align*}
\int_{(\S^2)^n}  \prod_{j=1}^k\prod_{l=1}^n \|x_j - y_l\|^{4b \Re(\alpha_j)} da(y_1)\cdots da(y_n)&= \biggl(\int_{\S^2}  \prod_{j=1}^k \|x_j - y\|^{4b \Re(\alpha_j)} da(y)\biggr)^n.
\end{align*}
The integral in the second line is finite, since $x_1,\ldots,x_k$ are distinct and $4b\Re(\alpha_j) > -2$ for each $j$. This proves the inequality \eqref{wagner2}. It is easy to see equation \eqref{wagner2} implies the summability of the series. The continuity and complex differentiability also follow easily from the above observations and the dominated convergence theorem.

\subsubsection{Proof of Lemma \ref{zeroform1two}}\label{zeroform1twopf}
First, let us fix $x_1=-e_3$ and let $x_2$ be any point in $\S^2 \setminus \{-e_3, e_3\}$. Let $a_n(u_2)$ denote the formula displayed in equation \eqref{ancor} with $k=2$ and $u_j = \sigma(x_j)$ for $j=1,2$. Note that $u_1=0$. We have to show that as $|u_2|\to \infty$, $a_n(u_2)$ approaches the formula displayed in equation~\eqref{cn2}. To do that, take any $M>0$, and define $a_n(u_2,M)$ using the same expression as in equation~\eqref{ancor}, but restricting the integration to the region $(\Omega_M)^n$ where $\Omega_M:=\{z\in \C: |z|\le M\}$. Then it follows easily by the dominated convergence theorem that for any $M$, $\lim_{|u_2|\to \infty} a_n(u_2, M)$ is given by the formula in equation \eqref{cn2}, but with the domain of integration restricted to $(\Omega_M)^n$. Then, by the monotone convergence theorem, we conclude that the double limit
\[
\lim_{M\to \infty} \lim_{|u_2|\to \infty} a_n(u_2,M)
\]
equals the right side of equation \eqref{cn2}. Thus, to complete the proof, we need to show that
\begin{align}\label{cn2show}
\lim_{M\to \infty} \lim_{|u_2|\to \infty} a_n(u_2,M) =\lim_{|u_2|\to \infty}    a_n(u_2).
\end{align}
Now, converting the integration back to the sphere, it is easy to see that 
\begin{align*}
a_n(u_2,M) &= \int_{(\S^2_M)^n}\exp\biggl\{-4b\sum_{j=1}^2\sum_{l=1}^n \alpha_j G(x_j, y_l) \\
&\qquad \qquad -4b^2 \sum_{1\le l < l'\le n} G(y_l,y_{l'}) \biggr\} da(y_1)\cdots da(y_n),
\end{align*}
where $\S^2_M$ is the set of all $y\in \S^2$ such that $|\sigma(y)|\le M$. Thus, by equation \eqref{wagner}, 
\begin{align}\label{ananbound}
&|a_n(u_2)-a_n(u_2,M)| = \int_{(\S^2)^n \setminus (\S^2_M)^n}\exp\biggl\{-4b\sum_{j=1}^2\sum_{l=1}^n \alpha_j G(x_j, y_l) \notag \\
&\hskip2in -4b^2 \sum_{1\le l < l'\le n} G(y_l,y_{l'}) \biggr\} da(y_1)\cdots da(y_n)\notag \\
&\le C_n\int_{(\S^2)^n\setminus (\S^2_M)^n}\prod_{j=1}^2\prod_{l=1}^n\|y_l - x_j\|^{4b\alpha_j} da(y_1)\cdots da(y_n)\notag \\
&\le C_n\sum_{i=1}^n\int_{y_i\notin \S^2_M}\prod_{j=1}^2\prod_{l=1}^n\|y_l - x_j\|^{4b\alpha_j} da(y_1)\cdots da(y_n)\notag \\
&= C_nn\biggl(\int_{\S^2 \setminus \S^2_M}\prod_{j=1}^2\|y - x_j\|^{4b\alpha_j} da(y)\biggr)\biggl(\int_{\S^2}\prod_{j=1}^2\|y - x_j\|^{4b\alpha_j} da(y)\biggr)^{n-1},
\end{align}
where $C_n$ does not depend on $M$ or $x_2$. Since $\Re(\alpha_j) > -\frac{1}{2b}$ for $j=1,2$, it follows that the supremum of the right side over all $x_2$ in the upper hemisphere is bounded by a number $\epsilon(M)$ that tends to zero as $M\to\infty$. Consequently,
\[
\limsup_{|w_2|\to \infty}|a_n(u_2,M)-a_n(u_2)|\le \epsilon(M).
\]
From this, it is easy to deduce the claim \eqref{cn2show}. The proof is now completed by appealing to equation \eqref{calphabasic}. %Lastly, note that the $e^{2\alpha_1\alpha_2}$ factor comes from the observation that $e^{-4\alpha_1\alpha_2 G_{\S^2}(-e_3,e_3)} = e^{2\alpha_1\alpha_2}$.

\subsubsection{Proof of Lemma \ref{zeroform1three}}\label{zeroform1threepf}
First, let us fix $x_1=-e_3$ and $x_2 = e_1$, and let $x_3$ be any point in $\S^2 \setminus \{-e_3, e_3, e_1\}$. Let $a_n(u_j)$ denote the formula displayed in equation \eqref{ancor} with $k=3$ and $u_j = \sigma(x_j)$ for $j=1,2,3$. Note that $u_1=0$ and $u_2 = 1$. We have to show that as $|u_3|\to \infty$, $a_n(u_3)$ approaches the formula displayed in equation~\eqref{cn3}. To do that, take any $M>0$, and define $a_n(u_3,M)$ using the same expression as in equation~\eqref{ancor}, but restricting the integration to the region $(\Omega_M)^n$ where $\Omega_M:=\{z\in \C: |z|\le M\}$. Then it follows easily by the dominated convergence theorem that for any $M$, $\lim_{|u_3|\to \infty} a_n(u_3, M)$ is given by the formula in equation \eqref{cn3}, but with the domain of integration restricted to $(\Omega_M)^n$. Then, by the monotone convergence theorem, we conclude that the double limit
\[
\lim_{M\to \infty} \lim_{|u_3|\to \infty} a_n(u_3,M)
\]
equals the right side of equation \eqref{cn3}. Thus, to complete the proof, we need to show that
\begin{align}\label{cn3show}
\lim_{M\to \infty} \lim_{|u_2|\to \infty} a_n(u_3,M) =\lim_{|u_3|\to \infty}    a_n(u_3).
\end{align}
Now, converting the integration back to the sphere, it is easy to see that 
\begin{align*}
a_n(u_3,M) &= \int_{(\S^2_M)^n}\exp\biggl(-4b\sum_{j=1}^3\sum_{l=1}^n \alpha_j G(x_j, y_l) \\
&\qquad \qquad -4b^2 \sum_{1\le l < l'\le n} G(y_l,y_{l'}) \biggr) da(y_1)\cdots da(y_n),
\end{align*}
where $\S^2_M$ is the set of all $y\in \S^2$ such that $|\sigma(y)|\le M$. Thus, proceeding as in the proof of the inequality \eqref{ananbound}, we get
\begin{align*}
&|a_n(u_3)-a_n(u_3,M)| \\%= \int_{(\S^2)^n \setminus (\S^2_M)^n}\exp\biggl(-4b\sum_{j=1}^3\sum_{l=1}^n \alpha_j G(x_j, y_l) \\
%&\qquad \qquad -4b^2 \sum_{1\le l < l'\le n} G(y_l,y_{l'}) \biggr) da(y_1)\cdots da(y_n)\\
%&\le C_n\int_{(\S^2)^n\setminus (\S^2_M)^n}\prod_{j=1}^3\prod_{l=1}^n\|y_l - x_j\|^{4b\alpha_j} da(y_1)\cdots da(y_n)\\
%&\le C_n\sum_{i=1}^n\int_{y_i\notin \S^2_M}\prod_{j=1}^3\prod_{l=1}^n\|y_l - x_j\|^{4b\alpha_j} da(y_1)\cdots da(y_n)\\
&\le C_nn\biggl(\int_{\S^2 \setminus \S^2_M}\prod_{j=1}^3\|y - x_j\|^{4b\alpha_j} da(y)\biggr)\biggl(\int_{\S^2}\prod_{j=1}^3\|y - x_j\|^{4b\alpha_j} da(y)\biggr)^{n-1},
\end{align*}
where $C_n$ does not depend on $M$ or $x_3$. Since $\Re(\alpha_j) > -\frac{1}{2b}$ for $j=1,2,3$, it follows that the supremum of the right side over all $x_3$ above a sufficiently high latitude is bounded by a number $\epsilon(M)$ that tends to zero as $M\to\infty$. Consequently,
\[
\limsup_{|w_3|\to \infty}|a_n(u_3,M)-a_n(u_3)|\le \epsilon(M).
\]
From this, it is easy to deduce the claim \eqref{cn3show}. The proof is now completed by appealing to equation \eqref{calphabasic}. %Lastly, note that the $e^{2\alpha_1\alpha_2}$ factor comes from the observation that $e^{-4\alpha_1\alpha_2 G_{\S^2}(-e_3,e_3)} = e^{2\alpha_1\alpha_2}$.

%\end{proof}

%\begin{proof}[Proof of Lemma \ref{intlmm}]

\subsection{Determinantal and Barnes--Gamma computations}\label{app:special}
\subsubsection{Proof of Lemma \ref{intlmm}}\label{intlmmpf}
The conditions $\Re(\alpha)> -2$ and $\Re(2\beta -\alpha) >2$ imply that the integrand decays sufficiently fast as $|z|\to \infty$ and blows up sufficiently slowly as $|z|\to 0$, to ensure that the integral converges absolutely. Transforming to polar coordinates, we have
\[
\int_{\C} \frac{|z|^\alpha}{(1+|z|^2)^\beta} d^2z = 2\pi \int_0^\infty \frac{r^{\alpha+1}}{(1+r^2)^\beta} dr.
\]
Making the change of variable
\[
s = \frac{1}{1+r^2}, 
\]
we have
\[
r = \sqrt{\frac{1-s}{s}}, \ \ \ ds = -\frac{2r}{(1+r^2)^2} dr.
\]
This gives
\begin{align*}
\int_0^\infty \frac{r^{\alpha+1}}{(1+r^2)^\beta} dr &=\frac{1}{2}\int_0^1 \biggl(\frac{1-s}{s}\biggr)^{\frac{1}{2}\alpha}s^{\beta-2} ds\\
&= \frac{1}{2}\int_0^1 s^{\beta-\frac{1}{2}\alpha -2 } (1-s)^{\frac{1}{2}\alpha} ds = \frac{\Gamma(\beta - \frac{1}{2}\alpha-1) \Gamma(\frac{1}{2}\alpha+1)}{2\Gamma(\beta)},
\end{align*}
where the last step follows by the Beta integral formula, noting that $\Re(\beta-\frac{1}{2}\alpha -1) >0$ and $\Re(\frac{1}{2}\alpha+1) > 0$.
%\end{proof}

\subsubsection{Proof of Lemma \ref{oneform1}}\label{oneform1pf}
Note that by equation \eqref{an1form},
\begin{align*}
a_n &= 4^n e^{\frac{1}{2} n(n-3-2w)} \int_{\C^n} \prod_{1\le i<j\le n}|z_i-z_j|^{2} d\nu(z_1)\cdots d\nu(z_n),
\end{align*}
where $\nu$ is the measure on $\C$ that has density $(1+|z|^2)^{-(n-w)}$ with respect to Lebesgue measure. 
Repeating the steps in the proof of Lemma \ref{zeroform1}, we arrive at the identity, for $n\ge 1$, that 
\begin{align}\label{an1one}
a_n = 4^n e^{\frac{1}{2} n(n-3-2w)} n! \prod_{j=0}^{n-1} \int_{\C} |z|^{2j} d\nu(z).
\end{align}
By Lemma \ref{intlmm},
\begin{align*}
\int_{\C} |z|^{2j} d\nu(z) &= \int_{\C} \frac{|z|^{2j}}{(1+|z|^2)^{n-w}} d^2z \\
&= \frac{\pi \Gamma(n-w-j-1)\Gamma(j+1)}{\Gamma(n-w)},
\end{align*}
noting that $\Re(2j) >-2$ and 
\[
\Re(2(n-w)-2j) = 2n - 2j +  2+2\sqrt{2}\Re(\alpha) > 2n-2j\ge 2.
\]
Plugging this into equation \eqref{an1one}, we get that for each $n\ge 1$,
\begin{align}\label{an3one}
a_n&=  \frac{(4\pi)^n e^{\frac{1}{2}n(n-3-2w)}\Gamma(n+1)}{\Gamma(n-w)^{n}} \biggl(\prod_{j=0}^{n-1}\Gamma(j+1)\biggr)\biggl(\prod_{j=0}^{n-1}\Gamma(j-w)\biggr).
\end{align}
By the identity \eqref{gzid}, 
\begin{align*}
\prod_{j=0}^{n-1}\Gamma(j-w) &= \prod_{j=0}^{n-1}\frac{G(j-w+1)}{G(j-w)} = \frac{G(n-w)}{G(-w)}.
\end{align*}
Using this and the identity \eqref{gn1id} in equation \eqref{an3one}, we get that for any $n\ge 2$, 
\begin{align}\label{analt}
a_n&=  \frac{(4\pi)^n e^{\frac{1}{2}n(n-3-2w)}\Gamma(n+1)G(n+1)G(n-w)}{\Gamma(n-w)^{n}G(-w)}.
\end{align}
Taking $a_0=1$ shows that the above formula is valid also for $n=0$. This completes the proof.
%\end{proof}

\subsubsection{Proof of Lemma \ref{twoform1}}\label{twoform1pf}
By Lemma \ref{zeroform1two}, we get that for each $n\ge 1$, 
\begin{align*}
a_n &= 4^n e^{\frac{1}{2} n(n-3-2w)} \int_{\C^n}\prod_{i=1}^n \frac{|z_i|^{2\sqrt{2}\alpha_1}}{(1+|z_i|^2)^{n-w}}  \prod_{1\le i<j\le n}|z_i-z_j|^{2} d^2z_1\cdots d^2z_n.
\end{align*}
Applying the same sequence of steps as in the  proofs of Lemma \ref{zeroform1} and Lemma \ref{oneform1}, we get
\begin{align*}
a_n &= 4^n e^{\frac{1}{2} n(n-3-2w)} \sum_{\sigma,\tau \in S_n} \sign(\sigma)\sign(\tau) \prod_{i=1}^n \int_{\C} z^{\sigma(i)-1}\overline{z}^{\tau(i)-1}  \frac{|z|^{2\sqrt{2}\alpha_1}}{(1+|z|^2)^{n-w}} d^2 z\\
&= n! 4^n e^{\frac{1}{2} n(n-3-2w)} \prod_{i=1}^n \int_{\C}  \frac{|z|^{2\sqrt{2}\alpha_1 + 2i-2}}{(1+|z|^2)^{n-w}} d^2 z.
\end{align*}
Now, by the assumed conditions, we have that for each $1\le i\le n$, 
\begin{align*}
\Re(2\sqrt{2}\alpha_1 + 2i-2) \ge 2\sqrt{2}\Re(\alpha_1)> -2,
\end{align*}
and 
\begin{align*}
&2\Re(n-w) -  \Re(2\sqrt{2}\alpha_1 + 2i-2) \\
&= 2(n + 1+\sqrt{2}\Re(\alpha_1+\alpha_2)) -  \Re(2\sqrt{2}\alpha_1 + 2i-2)\\
&= 2(n-i + 2) +2\sqrt{2}\Re(\alpha_2) > 2(n-i+2) - 2\ge 2.
\end{align*}
These bounds allow us to apply Lemma \ref{intlmm} and get
\begin{align*}
\int_{\C}  \frac{|z|^{2\sqrt{2}\alpha_1 + 2i-2}}{(1+|z|^2)^{n-w}} d^2 z &= \frac{\pi\Gamma(n-i-w -\sqrt{2}\alpha_1)\Gamma(i + \sqrt{2}\alpha_1)}{\Gamma(n-w)}\\
&= \frac{\pi\Gamma(n-i+1 +\sqrt{2}\alpha_2)\Gamma(i + \sqrt{2}\alpha_1)}{\Gamma(n-w)}.
\end{align*}
Thus, we get
\begin{align}\label{cnexp}
a_n &= \frac{\Gamma(n+1) (4\pi)^n e^{\frac{1}{2} n(n-3-2w)}}{(\Gamma(n-w))^n} \prod_{j=0}^{n-1}\Gamma(j+ 1 + \sqrt{2}\alpha_1)\Gamma(j +1 + \sqrt{2}\alpha_2).
\end{align}
Now, by the identity \eqref{gzid}, we have that for any $z\in \C\setminus\{0,-1,-2,\ldots\}$,
\begin{align}\label{gammaprod}
\prod_{j=0}^{n-1} \Gamma(j +z ) &= \prod_{j=0}^{n-1}\frac{G(j+z+1)}{G(j+z)} = \frac{G(n+z)}{G(z)}.
\end{align}
By the assumed conditions,  $1+\sqrt{2}\alpha_j  > 0$ for $j=1,2$. Thus, we may use the identity~\eqref{gammaprod} in equation~\eqref{cnexp}, to get
\begin{align*}
a_n &=\frac{ (4\pi)^n e^{\frac{1}{2} n(n-3-2w)}\Gamma(n+1) G(n+1+\sqrt{2}\alpha_1)G(n + 1+\sqrt{2}\alpha_2)}{(\Gamma(n-w))^nG(1+\sqrt{2}\alpha_1)G(1+\sqrt{2}\alpha_2)}\\
&= \frac{ (4\pi)^n e^{\frac{1}{2} n(n-3-2w)}\Gamma(n+1) G(n+\beta_1)G(n + \beta_2)}{(\Gamma(n-w))^nG(\beta_1)G(\beta_2)}.
\end{align*}
Now, if we define $f$ as in the statement of the lemma, then the previous display shows that $a_n = f(n)$ for each integer $n\ge 1$.  Since $w\ne 0$ and $\Re(w)<1$ by the given conditions, $f(0)$ is well-defined and equal to $1$, which matches the value $a_0=1$. 
This completes the proof.

\subsubsection{Proof of Lemma \ref{threeform1}}\label{threeform1pf}
By Lemma \ref{zeroform1three}, we get that for each $n\ge 1$, 
\begin{align*}
a_n &= 2^n e^{\frac{1}{2}n(n-3-2w)} \int_{\C^n}\prod_{i=1}^n (1+|z_i|^2)^{-(n-w)} \notag\\
&\hskip 1in \cdot \prod_{i=1}^n (|z_i|^{2\sqrt{2}\alpha_1}|z_i -1|^{2}) \prod_{1\le i<j\le n}|z_i-z_j|^{2} d^2z_1\cdots d^2z_n.
\end{align*}
We can rewrite this as 
\begin{align}\label{anspecial3}
a_n &= 2^n e^{\frac{1}{2}n(n-3-2w)} \int_{\C^n} \prod_{i=1}^n |z_i -1|^{2}\prod_{1\le i<j\le n}|z_i-z_j|^{2} d\nu(z_1)\cdots d\nu(z_n),
\end{align}
where $\nu$ is the measure on $\C$ that has density $(1+|z|^2)^{-(n-w)}|z|^{2\sqrt{2}\alpha_1}$ with respect to Lebesgue measure. Now take any $n\ge 1$, and note that by the Vandermonde determinant formula, we have 
\begin{align*}
\prod_{i=1}^n (1-z_i ) \prod_{1\le i<j\le n} (z_j-z_i) &= \det 
\begin{pmatrix}
1 & z_1 & z_1^2 & \cdots & z_1^{n}\\
1 & z_2 & z_2^2 & \cdots & z_2^{n}\\
\vdots & \vdots & \vdots &\ddots & \vdots\\
1 & z_n & z_n^2 & \cdots & z_n^{n}\\
1 & 1 & 1 & \cdots & 1\\
\end{pmatrix}= \sum_{\sigma \in S_{n+1}} \sign(\sigma)\prod_{i=1}^n z_i^{\sigma(i)-1}.
\end{align*}
This gives 
\begin{align*}
\prod_{i=1}^n|z_i - 1|^2\prod_{1\le i<j\le n}|z_i-z_j|^{2} &= \prod_{i=1}^n (1-z_i)\prod_{1\le i<j\le n}(z_j-z_i) \prod_{i=1}^n (1-\overline{z}_i) \prod_{1\le i<j\le n}(\overline{z}_j-\overline{z}_i)\\
&= \biggl(\sum_{\sigma \in S_{n+1}} \sign(\sigma)\prod_{i=1}^n z_i^{\sigma(i)-1}\biggr) \biggl(\sum_{\sigma \in S_{n+1}} \sign(\sigma)\prod_{i=1}^n \overline{z}_i^{\sigma(i)-1}\biggr)\\
&= \sum_{\sigma,\tau \in S_{n+1}} \sign(\sigma)\sign(\tau) \prod_{i=1}^n (z_i^{\sigma(i)-1}\overline{z}_i^{\tau(i)-1}).
\end{align*}
Plugging this into equation \eqref{anspecial3}, we get that for $n\ge 1$, 
\begin{align*}
a_n &=  2^n e^{\frac{1}{2}n(n-3-2w)} \sum_{\sigma,\tau \in S_{n+1}} \sign(\sigma)\sign(\tau) \int_{\C^n} \prod_{i=1}^n (z_i^{\sigma(i)-1}\overline{z}_i^{\tau(i)-1}) d\nu(z_1)\cdots d\nu(z_n)\\
&=  2^n e^{\frac{1}{2}n(n-3-2w)} \sum_{\sigma,\tau \in S_{n+1}} \sign(\sigma)\sign(\tau) \prod_{i=1}^n \int_{\C} z^{\sigma(i)-1}\overline{z}^{\tau(i)-1} d\nu(z).
\end{align*}
Since $\nu$ is a radially symmetric, we have 
\[
\int_{\C} z^k \overline{z}^l d\nu(z) =
\begin{cases}
\int_{\C} |z|^{2k} d\nu(z) &\text{ if } k=l,\\
0 &\text{ otherwise.}
\end{cases}
\]
Thus, only terms with $\sigma(i)=\tau(i)$ for $1\le i\le n$ survive. But this implies that $\sigma(n+1)=\tau(n+1)$.  Therefore, only terms with $\sigma=\tau$ survive. This gives
\begin{align*}%\label{an3}
a_n = 2^n e^{\frac{1}{2}n(n-3-2w)}  \sum_{\sigma \in S_{n+1}} \prod_{i=1}^n \int_{\C} |z|^{2\sigma(i)-2} d\nu(z).
\end{align*}
Splitting the above sum by values of $\sigma(n+1)$, we get
\begin{align*}
a_n &= 2^n e^{\frac{1}{2}n(n-3-2w)} n! \sum_{j=1}^{n+1} \prod_{\substack{1\le i\le n+1,\\
i\ne j}} \int_{\C} |z|^{2i-2} d\nu(z)\\
&= 2^n e^{\frac{1}{2}n(n-3-2w)} n! \sum_{j=1}^{n+1} \prod_{\substack{1\le i\le n+1,\\
i\ne j}} \int_{\C}  \frac{|z|^{2\sqrt{2}\alpha_1 + 2i-2}}{(1+|z|^2)^{n-w}} d^2 z.
\end{align*}
Now, by the assumed conditions, we have that for each $1\le i\le n+1$, 
\begin{align*}
\Re(2\sqrt{2}\alpha_1 + 2i-2) \ge 2\sqrt{2}\Re(\alpha_1)> -2,
\end{align*}
and 
\begin{align*}
&2\Re(n-w) -  \Re(2\sqrt{2}\alpha_1 + 2i-2) \\
&= 2(n + 1+\sqrt{2}\Re(\alpha_1+\alpha_2+\alpha_3)) -  \Re(2\sqrt{2}\alpha_1 + 2i-2)\\
&= 2(n-i + 3) +2\sqrt{2}\Re(\alpha_3) > 2(n-i+3) - 2\ge 2.
\end{align*}
These bounds allow us to apply Lemma \ref{intlmm} and get
\begin{align*}
\int_{\C}  \frac{|z|^{2\sqrt{2}\alpha_1 + 2i-2}}{(1+|z|^2)^{n-w}} d^2 z &= \frac{\pi\Gamma(n-i-w -\sqrt{2}\alpha_1)\Gamma(i + \sqrt{2}\alpha_1)}{\Gamma(n-w)}\\
&= \frac{\pi\Gamma(n-i+2 +\sqrt{2}\alpha_3)\Gamma(i + \sqrt{2}\alpha_1)}{\Gamma(n-w)}.
\end{align*}
Thus, we get
\begin{align*}%\label{cnexp}
a_n &= 2^n e^{\frac{1}{2}n(n-3-2w)}n! \sum_{j=1}^{n+1} \prod_{\substack{1\le i\le n+1,\\
i\ne j}} \frac{\pi\Gamma(n-i+2 +\sqrt{2}\alpha_3)\Gamma(i + \sqrt{2}\alpha_1)}{\Gamma(n-w)}\\
&=  \frac{(2\pi)^n e^{\frac{1}{2}n(n-3-2w)} \Gamma(n+1)}{\Gamma(n-w)^n}\biggl(\prod_{j=0}^n\Gamma(j+1+\sqrt{2}\alpha_3)\Gamma(j+1 + \sqrt{2}\alpha_1)\biggr)\\
&\qquad \cdot \sum_{j=0}^n \frac{1}{\Gamma(j+1+\sqrt{2}\alpha_1)\Gamma(n-j+1 + \sqrt{2}\alpha_3)}.
\end{align*}
Applying Lemma \ref{hyperlmm} and the identity \eqref{gzid} to the above, we get 
\begin{align*}
a_n &=\frac{(2\pi)^n e^{\frac{1}{2}n(n-3-2w)} \Gamma(n+1)G(n+2+\sqrt{2}\alpha_1) G(n+2+\sqrt{2}\alpha_3)}{\Gamma(n-w)^nG(1+\sqrt{2}\alpha_1)G(1+\sqrt{2}\alpha_3)}\\
&\qquad \cdot \biggl\{\frac{{_2F_1}(1, n-w-1; 1+\sqrt{2}\alpha_1; \frac{1}{2})}{2\Gamma(1+\sqrt{2}\alpha_1) \Gamma(n+1+\sqrt{2}\alpha_3)} - \frac{\sqrt{2}\alpha_3\, {_2F_1}(1,n-w-1;n+2+\sqrt{2}\alpha_1;\frac{1}{2})}{2\Gamma(n+2+\sqrt{2}\alpha_1)\Gamma(1+\sqrt{2}\alpha_3)}\biggr\}.
\end{align*}
This shows that $a_n = f(n)$ for each integer $n\ge 1$. It remains to show that $a_0=f(0)$, where $a_0:=1$. Note that by the identity \eqref{gzid},
\begin{align}\label{f0bad}
f(0)&= \frac{G(2+\sqrt{2}\alpha_1) G(2+\sqrt{2}\alpha_3)}{G(1+\sqrt{2}\alpha_1)G(1+\sqrt{2}\alpha_3)}\notag\\
&\qquad \cdot \biggl\{\frac{{_2F_1}(1, -w-1; 1+\sqrt{2}\alpha_1; \frac{1}{2})}{2\Gamma(1+\sqrt{2}\alpha_1) \Gamma(1+\sqrt{2}\alpha_3)} - \frac{\sqrt{2}\alpha_3\, {_2F_1}(1,-w-1;2+\sqrt{2}\alpha_1;\frac{1}{2})}{2\Gamma(2+\sqrt{2}\alpha_1)\Gamma(1+\sqrt{2}\alpha_3)}\biggr\}\notag \\
&= \frac{1}{2}\, {_2F_1}\biggl(1, -w-1; 1+\sqrt{2}\alpha_1; \frac{1}{2}\biggr) - \frac{\sqrt{2}\alpha_3\, {_2F_1}(1,-w-1;2+\sqrt{2}\alpha_1;\frac{1}{2})}{2(1+\sqrt{2}\alpha_1)}.
\end{align}
Since $\Re(\alpha_1)>-\frac{1}{\sqrt{2}}$, neither $1+\sqrt{2}\alpha_1$ nor $2+\sqrt{2}\alpha_1$ is a nonpositive integer. Thus, we can apply Lemma \ref{pfafflmm} to deduce that for any $z$ with $|z|<1$,
\begin{align*}
\frac{\sqrt{2}\alpha_3\, {_2F_1}(1,-w-1;2+\sqrt{2}\alpha_1;z)}{1+\sqrt{2}\alpha_1}&=  \frac{\sqrt{2}\alpha_3\, {_2F_1}(1,1-\sqrt{2}\alpha_3;2+\sqrt{2}\alpha_1;\frac{z}{z-1})}{(1-z)(1+\sqrt{2}\alpha_1)}.
\end{align*}
But, setting $u := z/(z-1)$ and assuming that $|u|<1$, we have 
\begin{align*}
\frac{\sqrt{2}\alpha_3\, {_2F_1}(1,1-\sqrt{2}\alpha_3;2+\sqrt{2}\alpha_1;u)}{1+\sqrt{2}\alpha_1} &= \sum_{j=0}^\infty \frac{(1)_j\sqrt{2}\alpha_3(1-\sqrt{2}\alpha_3)_ju^j}{(1+\sqrt{2}\alpha_1)(2+\sqrt{2}\alpha_1)_jj!}\\
&= -\frac{1}{u}\sum_{j=0}^\infty \frac{(1)_{j+1}(-\sqrt{2}\alpha_3)_{j+1}}{(1+\sqrt{2}\alpha_1)_{j+1}(j+1)!}u^{j+1}\\
&= -\frac{1}{u}({_2F_1}(1, -\sqrt{2}\alpha_3; 1+\sqrt{2}\alpha_1; u) -1).
\end{align*}
If $|u|<1$, then by Lemma \ref{pfafflmm},
\begin{align*}
{_2F_1}(1, -\sqrt{2}\alpha_3; 1+\sqrt{2}\alpha_1; u) &= \frac{1}{1-u}\, {_2F_1}\biggl(1, -w-1; 1+\sqrt{2}\alpha_1; \frac{u}{u-1}\biggr)\\
&= (1-z)\, {_2F_1}(1, -w-1; 1+\sqrt{2}\alpha_1; z).
\end{align*}
Combining the above steps, we get
\begin{align*}
\frac{\sqrt{2}\alpha_3\, {_2F_1}(1,-w-1;2+\sqrt{2}\alpha_1;z)}{1+\sqrt{2}\alpha_1}&= \frac{1}{z}((1-z)\, {_2F_1}(1, -w-1; 1+\sqrt{2}\alpha_1; z) - 1)
\end{align*}
for any $z$ such that $|z|<1$ and $|z|<|z-1|$. Both of these conditions are satisfied if $|z|<1$ and $\Re(z)<\frac{1}{2}$. Thus, taking $z\to \frac{1}{2}$ through this domain and applying Lemma \ref{hyperlemma}, we get
\begin{align*}
\frac{\sqrt{2}\alpha_3\, {_2F_1}(1,-w-1;2+\sqrt{2}\alpha_1;\frac{1}{2})}{1+\sqrt{2}\alpha_1}&= {_2F_1}\biggl(1, -w-1; 1+\sqrt{2}\alpha_1; \frac{1}{2}\biggr) - 2.
\end{align*}
Plugging this into equation \eqref{f0bad}, we get $f(0)=1$. Thus, $a_0 = f(0)$. 
This completes the proof.

\subsection{Contour-integral and asymptotic estimates}\label{app:contours}

%Our next goal is to prove Lemma \ref{fzfinal}. 
\subsubsection{Proof of Lemma \ref{fzfinal}}\label{fzfinalpf}
We need some preparation. First, note that by the formula \eqref{gformula} and the identity~\cite[Example 12.48]{whittakerwatson21}
\begin{align}\label{gzid}
G(z+1)=\Gamma(z)G(z),
\end{align}
we have $G(2)=\Gamma(1)G(1)=1$. Since the only zeros of the function $G$ are at the nonpositive integers, $G$ has an analytic logarithm $\Theta$ in the simply connected domain $\C\setminus(-\infty,0]$. (That is, $G(z) = e^{\Theta(z)}$ for any $z\in \C \setminus(-\infty,0]$.) If we further specify that this logarithm should agree with $\ln G$ on $(0,\infty)$, then it is uniquely determined, and given by
\[
\Theta(z+1) = \int_C \frac{G'(z+1)}{G(z+1)} dz
\]
for each $z\in \C \setminus(-\infty,-1]$, where $C$ is any contour from $1$ to $z$ lying entirely in $\C \setminus(-\infty,-1]$.

Now recall the functions $\Pi$ and $\psi$ introduced in \textsection\ref{zerosec}, and the relation \eqref{pirep}. It is known~\cite[Example 12.49]{whittakerwatson21} that for $z$ in the domain of $\psi$, 
\begin{align}\label{gprime}
\frac{G'(z+1)}{G(z+1)} = \frac{1}{2}\ln (2\pi) + \frac{1}{2} - z + z\psi(z).
\end{align}
Thus, for any $z\in \C \setminus(-\infty,0]$, $\Theta(z+1)$ can be represented as
\begin{align}\label{gmain}
\Theta(z+1) = \int_C \biggl( \frac{1}{2}\ln (2\pi) + \frac{1}{2} - w + w\psi(w)\biggr) dw,
\end{align}
where $C$ is any contour from $1$ to $z$ lying entirely in $\C \setminus(-\infty,0]$. Now, there is an exact formula for $\psi(z)$ when $\Re(z)>0$, given by~\cite[Section 12.31]{whittakerwatson21}
\begin{align}\label{psizform}
\psi(z) = \log z - \frac{1}{2z} -\int_0^\infty \biggl(\frac{1}{2}-\frac{1}{t}+\frac{1}{e^t-1}\biggr) e^{-tz} dt,
\end{align}
where $\log$ denotes the analytic  branch of logarithm on $\C \setminus(-\infty,0]$ that is real-valued on $(0,\infty)$. 
We obtain the following corollaries.
\begin{cor}\label{picor}
For any nonzero $z\in \C$ with $\Re(z)\ge 0$, 
\[
\Pi(z) = z\log z - z - \frac{1}{2}\log z + \frac{1}{2}\ln (2\pi) + R_1(z),
\]
where $|R_1(z)|\le C|z|^{-1}$ for some universal constant $C$.
\end{cor}
\begin{proof}
%This follows directly from the representation \eqref{pirep} and the upper bound provided by Lemma \ref{psibound}.
It suffices to prove the claim for $\Re(z)>0$, because we can then prove the claim for $\Re(z)=0$ by taking a limit and using the continuity of $\Pi$.  It can be shown using the formula \eqref{psizform}, as worked out in \cite[Section 12.31]{whittakerwatson21}, that for $\Re(z)>0$,
\begin{align*}
\Pi(z) = z\log z - z - \frac{1}{2}\log z + \frac{1}{2}\ln (2\pi) + \int_0^\infty \biggl(\frac{1}{2}-\frac{1}{t}+\frac{1}{e^t-1}\biggr) \frac{e^{-tz}}{t} dt.
\end{align*}
Define 
\[
\phi(t) := \biggl(\frac{1}{2}-\frac{1}{t}+\frac{1}{e^t-1}\biggr) \frac{1}{t}.
\]
Now, it is known~\cite[Equation (7.1)]{whittakerwatson21} that the function $t/(e^t-1)$ has an absolutely convergent power series expansion in a neighborhood of the origin, given by
\[
\frac{t}{e^t-1} = 1-\frac{t}{2} + \sum_{m=1}^\infty (-1)^{m-1}\frac{B_{m}}{(2m)!} t^{2m},
\]
where the $B_{m}$'s are the Bernoulli numbers. Thus, the power series expansion of $\phi$ at zero is 
\[
\phi(t) = \sum_{m=1}^\infty (-1)^{m-1}\frac{B_{m}}{(2m)!} t^{2m-2}.
\]
Since $B_1 = \frac{1}{6}$, this shows that  $\phi(t)\to \frac{1}{12}$ as $t\to 0$. Also, it is easy to check that $\phi(t)\to 0$ and $\phi'(t)\to 0$ as $t\to \infty$. 
Since $\Re(z) > 0$, this allows us to apply integration by parts and get
\begin{align}\label{phidiff}
\int_0^\infty\phi(t)e^{-tz} dt = \frac{1}{12z} + \frac{1}{z}\int_0^\infty \phi'(t) e^{-tz} dt.
\end{align}
Now, note that 
\begin{align*}
\phi'(t) &= \biggl(\frac{1}{t^2}-\frac{e^t}{(e^t-1)^2}\biggr) \frac{1}{t} - \biggl(\frac{1}{2}-\frac{1}{t}+\frac{1}{e^t-1}\biggr) \frac{1}{t^2}.%\\
%&= -\frac{1}{2t^2} + \frac{1}{2t^3} - \frac{e^t - 1 +te^t}{(e^t-1)^2 t^2} \\
%&= \frac{-(e^{2t}-2e^t + 1)t +(e^{2t}-2e^t+1)- 2t(e^t - 1- te^t)}{2(e^t-1)^2 t^3}\\
%&= \frac{(1-t)e^{2t} +2(t^2-1)e^t+1+t}{(e^t-1)^2 t^2}
\end{align*}
This is $O(t^{-2})$ as $t\to \infty$. Also, we know from the above series expansion that $\phi'$ is bounded near zero. Thus,
\[
\int_0^\infty|\phi'(t)|dt <\infty.
\]
By equation \eqref{phidiff} and the fact that $\Re(z)>0$, this shows that 
\[
\biggl|\int_0^\infty\phi(t)e^{-tz} dt\biggr| \le \frac{C}{|z|},
\]
where $C$ is a universal constant. This completes the proof.
\end{proof}
\begin{cor}\label{gcor}
%The function $\Theta=\log G$ defined above satisfies, for $\Re(z)\ge 0$, 
For any nonzero $z\in \C$ with $\Re(z)\ge 0$, 
\[
\Theta(z+1) = \frac{z-1}{2}\ln (2\pi)  +\frac{1}{2}z^2 \log z - \frac{3z^2}{4} -\frac{1}{12}\log z+ R_2(z),
\]
where $|R_2(z)|\le C(1+|z|^{-1})$ for some universal constant $C$.
\end{cor}
\begin{proof}
As in the previous corollary, it suffices to prove the claim for $\Re(z)>0$. So, assume that $\Re(z)>0$ and define 
\[
\chi(t) := \frac{1}{2}-\frac{1}{t}+\frac{1}{e^t-1}.
\]
Proceeding as in the proof of Corollary \ref{picor}, we see that $\chi(t)$ as the power series expansion
\[
\chi(t)  = \sum_{m=1}^\infty (-1)^{m-1}\frac{B_{m}}{(2m)!} t^{2m-1},
\]
converging absolutely in a neighborhood of the origin. Thus, $\chi(t)\to 0$ as $t\to 0$. Justifying integration by parts as before, we get
\begin{align*}
\int_0^\infty \chi(t) e^{-tz} dt &= \frac{1}{z}\int_0^\infty \chi'(t) e^{-tz} dt.
\end{align*}
The above expansion also shows that $\chi'(t) \to \frac{1}{12}$, $\chi''(t) \to 0$, and $\chi'''(t) \to -\frac{1}{720}$ as $t\to 0$. Thus, two further applications of integration by parts gives 
\begin{align*}
\int_0^\infty \chi(t) e^{-tz} dt &= \frac{1}{12z^2} -\frac{1}{720z^4} +  \frac{1}{z^4}\int_0^\infty \chi''''(t) e^{-tz} dt.
\end{align*}
Consequently, by equation \eqref{psizform}, we get
\begin{align*}
\psi(z) &= \log z - \frac{1}{2z}- \frac{1}{12z^2} +\frac{1}{720z^4} -  \frac{1}{z^4}\int_0^\infty \chi''''(t) e^{-tz} dt.
\end{align*}
It is easy to see that $\chi''''(t) = O(t^{-5})$ as $t\to\infty$. Also, the series expansion shows that $\chi''''(t)$ is bounded near the origin. Thus, $\int_0^\infty|\chi''''(t)|dt < \infty$. This shows that if we define
\begin{align*}
T(z) := \psi(z) - \biggl(\log z - \frac{1}{2z}- \frac{1}{12z^2}\biggr),
\end{align*}
then $|T(z)|\le C|z|^{-4}$ for  some universal constant $C$. 

Let $L$ be the straight line from $1$ to $z$. Then by equation~\eqref{gmain},
\begin{align*}
\Theta(z+1) &= \int_L  \biggl( \frac{1}{2}\ln (2\pi) + \frac{1}{2} - w + w\biggl(\log w - \frac{1}{2w}- \frac{1}{12w^2} + T(w)\biggr)\biggr) dw\\
&= \frac{1}{2}(z-1)\ln (2\pi) -\frac{1}{12}\log z -\frac{1}{2}(z^2-1)+ \int_L w\log w dw +\int_L wT(w) dw.
\end{align*}
Since 
\[
\frac{d}{dw}(w^2 \log w) = 2w\log w + w, 
\]
we have 
\[
\int_L w\log w dw  =\frac{1}{2} z^2\log z - \frac{z^2-1}{4}. 
\]
This gives 
\begin{align*}
\Theta(z+1) &= \frac{1}{2}(z-1)\ln (2\pi) -\frac{1}{12}\log z -\frac{3}{4}(z^2-1)+ \frac{1}{2}z^2 \log z +\int_L wT(w) dw.
\end{align*}
Finally, by the bound $|T(w)|\le C|w|^{-4}$, we get
\begin{align*}
\biggl|\int_L wT(w) dw\biggr| &\le \int_0^1 |(1-t+tz)T(1-t+tz)||z| dt\\
&\le C|z| \int_0^1 \frac{1}{|1-t+tz|^3} dt.
\end{align*}
Writing $z = r e^{\i \theta}$ and recalling that $\Re(z)>0$ (and hence, $\cos\theta >0$), we get
\begin{align*}
\int_0^1 \frac{1}{|1-t+tz|^3} dt &= \int_0^1 \frac{1}{((1-t)^2+2rt(1-t)\cos \theta + r^2 t^2)^{3/2}} dt\\
&\le  \int_0^1 \frac{1}{((1-t)^2+ r^2 t^2)^{3/2}} dt = \int_0^1 \frac{1}{(1-2t+ (1+r^2) t^2)^{3/2}} dt.
\end{align*}
To evaluate the last integral, let us make the change of variable
\[
u = \frac{(1+r^2)t -1}{r}. 
\]
This gives
\begin{align*}
\int_0^1 \frac{1}{(1-2t+ (1+r^2) t^2)^{\frac{3}{2}}} dt &= \frac{r}{1+r^2}\int_{-1/r}^r \biggl(1-\frac{2(ru+1)}{1+r^2}+ \frac{(ru+1)^2}{1+r^2}\biggr)^{-3/2} du\\
&= \frac{\sqrt{1+r^2}}{r^2} \int_{-1/r}^r (1+u^2)^{-3/2} du\\
&= \frac{\sqrt{1+r^2}}{r^2}\biggl[\frac{u}{\sqrt{1+u^2}}\biggr]^r_{-1/r} = \frac{1}{r} + \frac{1}{r^2}.
\end{align*}
Thus, we conclude that 
\[
\biggl|\int_L wT(w) dw\biggr|\le C(1+|z|^{-1}).
\]
This completes the proof.
\end{proof}
One final technical lemma that we need is the following. 
\begin{lmm}\label{gammalmm2}
For $z\in \C \setminus(-\infty,0]$, we have
\[
\Pi(z+1) = \Pi(z) + \log z,
\]
where $\log$ denotes the branch of logarithm on $ \C \setminus(-\infty,0]$ that is real-valued on $(0,\infty)$.
\end{lmm}

\begin{proof}
Note that the function $\Pi(z+1)-\Pi(z)-\log z$ is analytic on $ \C \setminus(-\infty,0]$. Also, it vanishes on $(0,\infty)$, since for $x\in(0,\infty)$, $\Pi(x)= \ln \Gamma(x)$ (which is a simple consequence of the definition of $\Pi$), $\log x = \ln x$, and $\Gamma(x+1)=x\Gamma(x)$. Thus, $\Pi(z+1)-\Pi(z)-\log z$ must vanish everywhere on $ \C \setminus(-\infty,0]$. 
\end{proof}

We will also need another technical lemma for future purposes, that we record here. The proof uses Corollary \ref{picor}.
\begin{lmm}\label{gammalmm3}
For any $z_1,z_2$ in the open right half-plane, we have
\[
\biggl|\frac{\Gamma(z_1)}{\Gamma(z_2)}\biggr| \le e^{(C_1+C_2) |z_1-z_2| + C_2C_3},
\]
where $C_1$ is the maximum value of $|\log z|$ along the line segment joining $z_1$ and $z_2$,  $C_2$ is the maximum value of $|z|^{-1}$ along the same line segment, and $C_3$ is a universal constant.
\end{lmm}
\begin{proof}
By Corollary \ref{picor} and the assumption that $z_1$ and $z_2$ are both in the open right half-plane, we have
\begin{align*}
\biggl|\frac{\Gamma(z_1)}{\Gamma(z_2)}\biggr| &= |\exp(\Pi(z_1)-\Pi(z_2))|\\
&= \exp\biggl\{\Re(z_1\log z_1 -z_1 - (z_2\log z_2-z_2)) - \frac{1}{2}\Re(\log z_1 - \log z_2) \\
&\qquad + \Re(R_1(z_1) - R_1(z_2)) \biggr\}.
\end{align*}
Now note that 
\begin{align*}
&z_1\log z_1 -z_1 - (z_2\log z_2-z_2) \\
&= \int_0^1 \frac{d}{dt}\{(tz_1 + (1-t)z_2)\log(tz_1+(1-t)z_2) - (tz_1+(1-t)z_2)\} dt\\
&=(z_1-z_2)\int_0^1\log(tz_1 + (1-t)z_2) dt.  
\end{align*}
This shows that 
\begin{align*}
|z_1\log z_1 -z_1 - (z_2\log z_2-z_2)| &\le C_1|z_1 - z_2|.
\end{align*}
Similarly, 
\[
|\log z_1 - \log z_2| \le |z_1-z_2|\int_0^1\frac{1}{|tz_1+(1-t)z_2|} dt \le C_2 |z_1-z_2|,
\]
and for $j=1,2$, $|R_1(z_j)|\le C_3 C_2$, where $C_3$ is a universal constant.
\end{proof}

We are now ready to prove Lemma \ref{fzfinal}.
%\begin{proof}[Proof of Lemma \ref{fzfinal}]
Observe that by the identity \eqref{gzid}, 
\[
G(z+1)^2 = \frac{G(z+3)^2}{\Gamma(z+1)^2 \Gamma(z+2)^2}.
\]
By Lemma \ref{gammalmm2},
\begin{align*}
\Gamma(z+1)^{z+1} &= e^{(z+1)\Pi(z+1)} \\
&= e^{(z+1)(\Pi(z+2)-\log(z+1))} = \Gamma(z+2)^{z+1} e^{-(z+1)\log(z+1)}. 
\end{align*}
Using these identities in the definition \eqref{fdef} of $f$, we get 
\begin{align*}
f(z) &= \frac{(4\pi)^z e^{\frac{1}{2}z(z-1)}G(z+1)^2}{\Gamma(z+1)^{z-1}}= \frac{(4\pi)^z e^{\frac{1}{2}z(z-1)}G(z+3)^2}{\Gamma(z+1)^{z+1}\Gamma(z+2)^2}\\
&= \frac{(4\pi)^z e^{(z+1)\log(z+1)} e^{\frac{1}{2}z(z-1)}G(z+3)^2}{(\Gamma(z+2))^{z+3}}\\
&= \exp\biggl\{z\log(4\pi) + (z+1)\log(z+1) + \frac{1}{2}z(z-1) + 2\Theta(z+3) - (z+3)\Pi(z+2)\biggr\}.
\end{align*}
Substituting the expansions of $\Theta$ and $\Pi$ from Corollary \ref{gcor} and Corollary \ref{picor}, we get that if $\Re(z) \ge -1$ and $z\ne -1$, 
\begin{align*}
f(z) &= \exp\biggl\{z\log(4\pi) + (z+1)\log(z+1) + \frac{1}{2}z(z-1)  + (z+1)\ln (2\pi)\\
&\qquad \qquad  +(z+2)^2 \log (z+2) - \frac{3(z+2)^2}{2}-\frac{1}{6}\log(z+2) + 2R_2(z+2)\\
&\qquad \qquad -(z+3)(z+2)\log(z+2) + (z+3)(z+2)+\frac{1}{2}(z+3)\log(z+2)\\
&\qquad \qquad - \frac{1}{2}(z+3)\ln(2\pi) - (z+3)R_1(z+2)\biggr\}.
\end{align*}
To simplify this, note that the terms that are quadratic in $z$ add up to
\begin{align*}
 \frac{1}{2}z(z-1) - \frac{3(z+2)^2}{2}+ (z+3)(z+2) = -\frac{3}{2}z, 
\end{align*}
and the linear terms have real coefficients. Thus, the linear and quadratic terms add up to $z$ times a universal real coefficient. The coefficients of $\log(z+2)$ add up to 
\[
(z+2)^2-\frac{1}{6} -(z+3)(z+2)+ \frac{1}{2}(z+3) = -\frac{1}{2}z - \frac{2}{3}.
\]
Lastly, note that by the bounds from Corollary \ref{picor} and Corollary \ref{gcor},
\[
|(z+3)R_1(z+2)| + |2R_2(z+2)|\le C,
\]
where $C$ is a universal constant. Putting it all together, we get the simplification
\begin{align}\label{fzsimple}
f(z) &= \exp\biggl\{(z+1)\log(z+1) - \frac{1}{2}(z+1)\log (z+2) +Az + R_3(z)\biggr\},
\end{align}
where $A$ is a real universal constant and $|R_3(z)|\le C \log(2+|z|)$ for some positive universal constant $C$. Now, since $\Re(z)\ge -1$ and $z\ne -1$, we have 
\begin{align*}
|\log (z+2) - \log(z+1)| &= \biggl|\int_0^1 \frac{1}{z+1 + t} dt\biggr|\\
&\le \int_0^1 \frac{1}{|z+1+t|} dt\le \frac{1}{|z+1|}. 
\end{align*} 
Therefore, replacing $\log(z+2)$ by $\log(z+1)$ in equation \eqref{fzsimple} incurs an error whose absolute value is bounded by a universal constant. 
This gives the desired expression for $f(z)$. For the bound on $|f(z)|$, note that if $z=x+\i y$, then
\begin{align*}
&\biggl|\exp\biggl\{\frac{1}{2}(z+1)\log (z+1)+ Az + R(z)\biggr\}\biggr|\\
&= \biggl|\exp\biggl\{\frac{1}{2}(x+1+\i y)(\ln |z+1| + \i \arg(z+1)) + Ax+\i A y + R(z)\biggr\}\biggr|\\
&= \exp\biggl\{\frac{1}{2}(x+1)\ln |z+1| - \frac{1}{2} y \arg(z+1) + Ax + \Re(R(z))\biggr\}.
\end{align*}
This completes the proof.
%\end{proof}

%\begin{proof}[Proof of Lemma \ref{gammalmm}]
\subsubsection{Proof of Lemma \ref{gammalmm}}\label{gammalmmpf}
Note that $n\ge 1$. By repeated applications of the relation $\Gamma(w+1)=w\Gamma(w)$, we have
\begin{align*}
\Gamma(-z) = \frac{ \Gamma(-z+n)}{-z(-z+1)\cdots (-z+n-1)}.  
\end{align*}
By the definition of $n$, $\Re(-z+n) \ge 1$. Thus, by Corollary \ref{picor},
\begin{align*}
\Gamma(-z+n) &= \exp\biggl\{ (-z+n)\log (-z+n) + z - n - \frac{1}{2}\log (-z+n) +R_1(-z+n)\biggr\},
\end{align*}
where $|R_1(-z+n)|\le C$ for some universal constant $C$. 
Recall that $a = \Re(-z+n)$, and $1\le a < 2$. Thus,
\begin{align*}
|\Gamma(-z+n)| &= \biggl|\exp\biggl\{(a-\i y) (\ln |a-\i y| + \i \arg(a-\i y)) - a +\i y\\
&\qquad \qquad  - \frac{1}{2}(\ln |a-\i y| + \i \arg(a-\i y)) + R_1(a-\i y)\biggr\}\biggr|\\
&= \exp\biggl\{a \ln |a-\i y| + y \arg(a-\i y) - a -\frac{1}{2}\ln|a-\i y| + \Re(R_1(a-\i y))\biggr\}.
\end{align*}
It is easy to check that this completes the proof.
%\end{proof}

%\begin{proof}

\subsubsection{Proof of Lemma \ref{zerofinal}}\label{zerofinalpf}
In this proof, $C, C_1, C_2,\ldots$ will denote arbitrary positive constants that may depend only on $x_0$, whose values may change from line to line. Take any $N\ge 3$ and let $x:= N+x_0$. Then note that for any $y\in \R$, 
\begin{align*}
\prod_{j=0}^{\lceil x+1\rceil-1}|j-x- \i y| &= \exp\biggl(\sum_{j=0}^{N }\ln |j-x- \i y| \biggr)\\
&= \exp\biggl(\frac{1}{2}\sum_{j=0}^{N }\ln ((j-x)^2+ y^2) \biggr)= \exp\biggl( \frac{1}{2}\sum_{j=0}^{N}\ln ((j+x_0)^2+ y^2) \biggr).
\end{align*}
Next, note that 
\begin{align*}
\sum_{j=0}^{N}\ln ((j+x_0)^2+ y^2) &= \ln(x_0^2 + y^2) + \ln((1+x_0)^2+y^2) + \sum_{j=2}^N \ln((j+x_0)^2+y^2)\\
&\ge -C_1 +  \int_{1}^N\ln((s+x_0)^2 + y^2) ds= -C_1+ \int_{1+x_0}^{x} \ln(s^2+y^2) ds. %\\
%&\ge -C_1- \int_0^{1+x_0} \log s^2 ds + \int_0^x \log(s^2+y^2) ds\\
%&= -C_2 + \int_0^x \log(s^2+y^2) ds.
\end{align*}
Integration by parts gives
\begin{align*}
\int_{1+x_0}^{x} \ln(s^2+y^2) ds &= [ s\ln(s^2+y^2)]^{x}_{1+x_0}- \int_{1+x_0}^{x} \frac{2s^2}{s^2 + y^2} ds\\
&\ge x\ln(x^2+y^2) - \ln(1+y^2) - 2x - C_1. % + 2y^2 \int_0^{x} \frac{1}{s^2 + y^2} ds\\
%&= x\log(x^2+y^2) - 2x + 2y\arctan(x/y)\\
%&\ge x\log(x^2+y^2) - 4x,
\end{align*}
%where the last inequality holds because $\arctan a \ge a$ for $a\ge 0$. 
Combining the last three displays, we get
\begin{align}\label{prodjineq}
\prod_{j=0}^{\lceil x+1\rceil-1}|j-x- \i y|  &\ge C_1\exp\biggl(\frac{1}{2}x \ln (x^2 + y^2) - C_2x\biggr)\notag\\
&= C_1\exp(x \ln |x+\i y| - C_2 x - \ln(1+y^2)).
\end{align}
Using this in equation \eqref{gxir}, we get
\begin{align*}
|g(x+\i y)|&\le C_1\exp\biggl\{C_2\ln(2+|x|+|y|)+ y \arg(\lceil x+1 \rceil -x- \i y) \notag\\
&\qquad - x\ln|x+\i y| + C_3 x + \frac{1}{2}(x+1)\ln |x+ \i y+1|  \notag \\
&\qquad \qquad \qquad - \frac{1}{2} y \arg(x+ \i y+1) \biggr\} (\mu e^{\sqrt{2}c})^x.
\end{align*}
Simplifying this, we arrive at the inequality
\begin{align}\label{gxymain}
|g(x+\i y) | &\le C_1 (\mu e^{\sqrt{2}c})^x \exp\biggl\{y\biggl(\arg(1-x_0-\i y) -\frac{1}{2}\arg(x+\i y+1)\biggr)\notag \\
&\hskip1in - x\ln |x+\i y | + \frac{1}{2}x \ln|x+\i y+1|\notag \\
&\hskip1in + C_2 x + C_3\ln(2+|x|+|y|)\biggr\}.
\end{align}
Since $x= N+x_0> 2$, we have
\[
\ln |x+\i y | \ge \ln |x+\i y+1| - C.
\]
Similarly,
\[
\ln(2+|x|+|y|) \le \ln|x+\i y + 1| + C.
\]
Also, since $1-x_0> 0$ and $x = N+x_0 \ge 3+x_0 > 1-x_0$, we have
\begin{align*}
y\biggl(\arg(1-x_0-\i y) -\frac{1}{2} \arg(x+\i y+1)\biggr) &= y\biggl(-\arg(1-x_0+\i y) -\frac{1}{2} \arg(x+\i y+1)\biggr) \\
&\le -\frac{3}{2}y \arg(x+\i y+1).
\end{align*}
Using these two observations in equation \eqref{gxymain}, we get
\begin{align}\label{gxybound}
|g(x+\i y) | &\le C_1 (\mu e^{\sqrt{2}c})^x \exp\biggl( - \frac{3}{2}y\arg(x+\i y+1)- \frac{1}{2}x \ln|x+\i y+1|\notag \\
&\hskip1in + C_2 x + C_3\ln|x+\i y+1|\biggr).
\end{align}
This gives 
\begin{align*}
|F(x)| &\le C_1  (\mu e^{(\sqrt{2}c+C_2)})^x \int_{-\infty}^\infty \exp\biggl(- \frac{3}{2}y\arg(x+\i y+1)- \frac{1}{2}x \ln|x+\i y+1|\notag\\
&\hskip1in + C_3\ln |x+\i y+1|\biggr)dy.
\end{align*}
Making the change of variable $ u = y/(x+1)$, we get 
\begin{align*}
|F(x)|&\le C_1 (x+1) (\mu e^{(\sqrt{2}c+C_2)})^x \int_{-\infty}^\infty \exp\biggl( - \frac{3}{2} (x+1)u\arg((x+1)(1+\i u))\\
&\hskip1.5in - \frac{1}{2}x \ln ((x+1)|1+\i u|)+ C_3\ln((x+1)|1+\i u|)\biggr)du\\
&= C_1 (x+1)^{C_2} \exp\biggl((\sqrt{2}c+C_3) x- \frac{1}{2}x\ln (x+1)\biggr)\\
&\qquad\cdot \int_{-\infty}^\infty \exp\biggl( - \frac{3}{2} (x+1)u\arg(1+\i u) +C_4 \ln |1+\i u|\biggr)du.
\end{align*}
Since
\begin{align*}
&\int_{-\infty}^\infty \exp\biggl( - \frac{3}{2} (x+1)u\arg(1+\i u) +C_4\ln |1+\i u|\biggr)du\\
&\le \int_{-\infty}^\infty \exp\biggl( - \frac{3}{2}u\arg(1+\i u)+C_4\ln |1+\i u|\biggr)du,
\end{align*}
which does not depend on $x$ and is finite, we conclude that
\[
|F(x)|\le C_1 (x+1)^{C_2}\mu^x \exp\biggl((\sqrt{2}c+C_3) x- \frac{1}{2}x\ln (x+1)\biggr).
\]
In particular, $F(N+x_0)\to 0$ as $N\to\infty$. This completes the proof.
%\end{proof}

\subsubsection{Proof of Lemma \ref{fzfinalone}}\label{fzfinalonepf}
%Note that $\Re(-w) = \Re(1+\sqrt{2}\alpha) > 0$. 
Recall that by assumption, $\Re(w)\in [-1,0)$. 
By the identity \eqref{gzid}, 
\[
G(z+1) = \frac{G(z+3)}{\Gamma(z+1) \Gamma(z+2)}, \ \ \ G(z-w) = \frac{G(z-w+2)}{\Gamma(z-w)\Gamma(z-w+1)}. 
\]
By Lemma \ref{gammalmm2}, we have that for any $z\in w+( \C\setminus(-\infty,0])$,
\begin{align*}
\Gamma(z-w)^{z+1} &= e^{(z+1)\Pi(z-w)} \\
&= e^{(z+1)(\Pi(z-w+1)-\log(z-w))} = \Gamma(z-w+1)^{z+1} e^{-(z+1)\log(z-w)}. 
\end{align*}
Thus, for all $z$ in the domain $\Omega:= (w+(\C\setminus(-\infty,0]))\setminus\{-1,-2,\ldots\}$, we have
\begin{align}\label{fznewalt}%\label{fdef}
f(z) &= \frac{(4\pi)^z e^{\frac{1}{2}z(z-3-2w)}\Gamma(z+1)G(z+1)G(z-w)}{\Gamma(z-w)^{z}G(-w)}\notag \\
&= \frac{(4\pi)^z e^{\frac{1}{2}z(z-3-2w)+(z+1)\log(z-w)}G(z+3)G(z-w+2)}{\Gamma(z-w+1)^{z+2}\Gamma(z+2)G(-w)}.
\end{align}
We can write this as
\begin{align*}
f(z) &= \exp\biggl\{z\log(4\pi) + (z+1)\log(z-w) + \frac{1}{2}z(z-3-2w) + \Theta(z+3) \\
&\qquad + \Theta(z-w+2) - (z+2)\Pi(z-w+1)- \Pi(z+2)-\Theta(-w)\biggr\}.
\end{align*}
Now take any $z\in \Omega$ which also satisfies $\Re(z+2)\ge 0$ and $\Re(z-w+2)\ge 0$. Since $\Re(-w)\in (0,1]$, this is equivalent to simply saying that $z\in \Omega$ and $\Re(z)\ge -2$. Then, substituting the expansions of $\Theta$ and $\Pi$ from Corollary \ref{gcor} and Corollary \ref{picor}, we get 
\begin{align*}
f(z) &= \exp\biggl\{z\log(4\pi) + (z+1)\log(z-w) + \frac{1}{2}z(z-3-2w)  + \frac{1}{2}(z+1)\ln (2\pi)\\
&\qquad  +\frac{1}{2}(z+2)^2 \log (z+2) - \frac{3(z+2)^2}{4}-\frac{1}{12}\log(z+2) + R_2(z+2)\\
&\qquad + \frac{1}{2}(z-w)\ln (2\pi)  +\frac{1}{2}(z-w+1)^2 \log (z-w+1) - \frac{3(z-w+1)^2}{4}\\
&\qquad  -\frac{1}{12}\log(z-w+1) + R_2(z-w+1) -(z+2)(z-w+1)\log(z-w+1) \\
&\qquad + (z+2)(z-w+1)+\frac{1}{2}(z+2)\log(z-w+1) - \frac{1}{2}(z+2)\ln(2\pi) \\
&\qquad - (z+2)R_1(z-w+1) - (z+2)\log(z+2) +(z+2) + \frac{1}{2}\log(z+2)\\
&\qquad  - \frac{1}{2}\log(2\pi) - R_1(z+2)-\Theta(-w)\biggr\}.
\end{align*}
To simplify this, note that the terms that are quadratic in $z$ add up to
\begin{align*}
& \frac{1}{2}z(z-3-2w)- \frac{3(z+2)^2}{4} - \frac{3(z-w+1)^2}{4} + (z+2)(z-w+1)\\
%&= \biggl(-\frac{3}{2}-w-3 -\frac{3}{2}(1-w)+(1-w+2)\biggr)z -3-\frac{3}{4}(1-w)^2 + 2(1-w)\\
&= \biggl(-\frac{1}{2}w - 3\biggr)z  -3-\frac{3}{4}(1-w)^2 + 2(1-w).
\end{align*}
and the linear terms have real coefficients. Thus, the linear and quadratic terms add up to a constant (depending on $w$) plus $(A-\frac{1}{2}w)z$, where $A$ is a universal real constant. 

Next, note that the coefficients of $\log(z+2)$ add up to 
\[
\frac{1}{2}(z+2)^2-\frac{1}{12} -(z+2)+\frac{1}{2} = \frac{1}{2}z^2 + z + \frac{5}{12}.
\]
The coefficients of $\log(z-w+1)$ add up to
\begin{align*}
&\frac{1}{2}(z-w+1)^2 -\frac{1}{12}-(z+2)(z-w+1)+\frac{1}{2}(z+2)\\
&= -\frac{1}{2}z^2 -\frac{3}{2}z + \frac{1}{2}(1-w)^2 -\frac{1}{12} -2(1-w)+1.
\end{align*}
Thus, the terms involving $\log(z+2)$ and $\log(z-w+1)$ add up to
\begin{align*}
-\frac{1}{2}z^2 \log \frac{z-w+1} {z+2}+ z\log (z+2) - \frac{3}{2}z\log(z-w+1) +O(1),
\end{align*}
where $O(1)$ denotes a term whose absolute value is bounded by a constant that depends only on $w$. Now recall that $\Re(z)\ge \Re(w)\ge -1$. Thus, we have (with the same convention about the big-$O$ notation)
\begin{align*}
-\frac{1}{2}z^2 \log \frac{z-w+1} {z+2} &= -\frac{1}{2}z^2 \log \biggl(1 - \frac{w+1} {z+2}\biggr)\\
&= \frac{z^2(w+1)}{2(z+2)} + O(1)\\
&= \frac{1}{2}(z+2)(w+1) + O(1) = \frac{1}{2}zw + O(1).
\end{align*}
The $\frac{1}{2}zw$ above cancels the $-\frac{1}{2}zw$ from the linear and quadratic terms, leaving us only with $Az$. Similarly, 
\[
z\log (z+2) - \frac{3}{2}z\log(z-w+1) = -\frac{1}{2}z\log(z-w+1) + O(1).
\]
Thus, the total contribution coming from the linear and quadratic terms, and the terms involving $\log(z+2)$ and $\log(z-w+1)$, is equal to
\[
Az -\frac{1}{2}z\log(z-w+1) + O(1).
\]
We have to also keep in mind the $(z+1)\log (z-w)$ term, which we have not combined with anything else.
Lastly, by the bounds from Corollary \ref{picor} and Corollary \ref{gcor}, the remainder terms add up to $O(1)$, since $\Re(z)\ge \Re(w)\ge -1$.  Putting it all together, we get the desired results.
%\end{proof}

\subsubsection{Proof of Lemma \ref{zerofinal0one}}\label{zerofinal0onepf}
Define the function 
\[
g(z) := \Gamma(-z) f(z) (\mu e^{\sqrt{2}c})^{z}
\]
on the domain $(w+(\C \setminus(-\infty,0]))\setminus\Z$. Fix some $x_0\in (\Re(w),0)$. 
Take any $N\ge 1$ and $R\ge 1$. Let $C_R$ be the rectangular contour with vertices $x_0\pm \i R$ and $N + x_0\pm \i R$, traversed counter-clockwise (see Figure \ref{fig:CR-with-poles}). Since $f$ is analytic in $w+(\C\setminus(-\infty,0])$, we deduce that the only poles of $g$ are at the nonnegative integers, arising due to the poles of the Gamma function. Since
\[
\Res(\Gamma, -n) = \frac{(-1)^n}{n!}
\]
for each nonnegative integer $n$, we get that 
\[
\Res(g, n) =- \frac{(-1)^n}{n!}f(n) (\mu e^{\sqrt{2}c})^n. % = \frac{(-1)^n}{n!}a_n t^n.
\]
Since $C_R$ encloses the poles at $0,1,\ldots,N-1$, Cauchy's theorem gives
\begin{align}\label{cr1one}
\frac{1}{2\pi \i} \oint_{C_R} g(z) dz = -\sum_{n=0}^{N-1}\frac{(-1)^n}{n!}f(n) (\mu e^{\sqrt{2}c})^n,
\end{align}
Now, by Lemma \ref{fzfinalone} and Lemma \ref{gammalmm}, for any $x\ge x_0$, 
\begin{align}\label{gxirone}
|g(x + \i R)| &\le \frac{C_1\exp(C_2\log(2+R)+ R \arg(\lceil x+1 \rceil -x- \i R))}{\prod_{j=0}^{\lceil x+1\rceil-1 }|j-x+ \i R|}\notag \\
&\quad \cdot \exp\biggl((x+1)\ln |x+ \i R-w| - R\arg(x+ \i R-w) - \frac{1}{2}x\ln|x+\i R -w +1|\notag\\
&\quad   + \frac{1}{2}R \arg(x+\i R-w+1)+ Ax + C\biggr) e^{\sqrt{2}cx},
\end{align}
where  $A, C, C_1, C_2$ are universal constants. Clearly, the bound is decreasing exponentially in $R$ as $R \to \infty$, uniformly over $x$ in any given bounded range. Thus,
\begin{align}\label{cr2one}
\lim_{R\to \infty} \max_{x_0\le x\le N+x_0} |g(x+ \i R)| = 0.
\end{align}
Similarly, 
\begin{align}\label{cr3one}
\lim_{R\to \infty} \max_{x_0\le x\le N+x_0} |g(x-\i R)| = 0.
\end{align}
By equations \eqref{cr2one} and \eqref{cr3one}, we can take $R\to \infty$ in equation \eqref{cr1one}. This completes the proof.

\subsubsection{Proof of Lemma \ref{zerofinalone}}\label{zerofinalonepf}
In this proof, $C, C_1, C_2,\ldots$ will denote arbitrary positive constants that may depend only on $x_0$ and $w$, whose values may change from line to line. Take any $N\ge 3$ and let $x:= N+x_0$. 
Using the inequality \eqref{prodjineq} in equation \eqref{gxirone}, we get
\begin{align*}
|g(x+\i y)|&\le C_1\exp\biggl(C_2\log(2+|x|+|y|)+ y \arg(\lceil x+1 \rceil -x- \i y) \notag\\
&\qquad - x\ln|x+\i y| + (x+1)\ln |x+ \i y-w| - y\arg(x+ \i y-w) \notag\\
&\quad   - \frac{1}{2}x\ln|x+\i y -w +1| + \frac{1}{2}y \arg(x+\i y-w+1)+ C_3x\biggr) (\mu e^{\sqrt{2}c})^x.
\end{align*}
Rearranging this and noting that $\lceil x+1\rceil - x= 1-x_0$, we arrive at the inequality
\begin{align}\label{gxymainone}
|g(x+\i y) | &\le C_1 (\mu e^{\sqrt{2}c})^x \exp\biggl\{y\biggl(\arg(1-x_0-\i y) -\arg(x+\i y-w)\notag \\
&\qquad +\frac{1}{2}\arg(x+\i y-w+1)\biggr) - x\ln |x+\i y | + (x+1) \ln|x+\i y-w|\notag \\
&\qquad -\frac{1}{2}x\ln|x+\i y - w+1| + C_2 x + C_3\ln(2+|x|+|y|)\biggr\}.
\end{align}
Since $x= N+x_0> 2$ and $\Re(w)\in [-1,0)$, we have
\begin{align*}
&\ln |x+\i y | \ge \ln |x+\i y-w+1| - C,\\
&\ln|x+\i y -w|\le \ln|x+\i y-w+1|+C,\\
&\ln(2+|x|+ |y|) \le \ln|x+\i y -w + 1| + C.
\end{align*}
Also, since $1-x_0> 0$ and $\Re(x-w+1) \ge x\ge 1-x_0$, we have
\begin{align*}
&y\biggl(\arg(1-x_0-\i y) -\arg(x+\i y-w) +\frac{1}{2}\arg(x+\i y-w+1)\biggr) \\
&= y\biggl(-\arg(1-x_0+\i y)  -\arg(x+\i y-w) +\frac{1}{2}\arg(x+\i y-w+1) \biggr) \\
&\le -\frac{3}{2}y \arg(x+\i y-w+1).
\end{align*}
Using these two observations in equation \eqref{gxymainone}, we get
\begin{align}\label{gxyboundone}
|g(x+\i y) | &\le C_1 (\mu e^{\sqrt{2}c})^x \exp\biggl( - \frac{3}{2}y\arg(x+\i y-w+1)- \frac{1}{2}x \ln|x+\i y-w+1|\notag \\
&\hskip1in + C_2 x + C_3\ln|x+\i y-w+1|\biggr).
\end{align}
Let $a := \Re(w)$. The above inequality gives 
\begin{align*}
|F(x)| &\le C_1  (\mu e^{(\sqrt{2}c+C_2)})^x \int_{-\infty}^\infty \exp\biggl(- \frac{3}{2}y\arg(x+\i y-w+1)- \frac{1}{2}x \ln|x+\i y-w+1|\notag\\
&\hskip1in + C_3\ln |x+\i y-w+1|\biggr)dy\\
&= C_1  (\mu e^{(\sqrt{2}c+C_2)})^x \int_{-\infty}^\infty \exp\biggl(- \frac{3}{2}y\arg(x+\i y-a+1)- \frac{1}{2}x \ln|x+\i y-a+1|\notag\\
&\hskip1in + C_3\ln |x+\i y-a+1|\biggr)dy,
\end{align*}
where the second line is obtained by changing variable $y \to y+\Im(w)$. 
Making the change of variable $ u = y/(x-a+1)$, we get 
\begin{align*}
|F(x)|&\le C_1 (x-a+1) (\mu e^{(\sqrt{2}c+C_2)})^x \int_{-\infty}^\infty \exp\biggl( - \frac{3}{2} (x-a+1)u\arg((x-a+1)(1+\i u))\\
&\hskip1in - \frac{1}{2}x \ln ((x-a+1)|1+\i u|)+ C_3\ln((x-a+1)|1+\i u|)\biggr)du\\
&= C_1 x^{C_2} \exp\biggl((\sqrt{2}c+C_3) x- \frac{1}{2}x\ln (x-a+1)\biggr)\\
&\hskip1in \cdot \int_{-\infty}^\infty \exp\biggl( - \frac{3}{2} (x-a+1)u\arg(1+\i u) + C_4 \ln |1+\i u|\biggr)du.
\end{align*}
Since
\begin{align*}
&\int_{-\infty}^\infty \exp\biggl( - \frac{3}{2} (x-a+1)u\arg(1+\i u) +C_4\ln |1+\i u|\biggr)du\\
&\le \int_{-\infty}^\infty \exp\biggl( - \frac{3}{2}u\arg(1+\i u)+C_4\ln |1+\i u|\biggr)du,
\end{align*}
which does not depend on $x$ and is finite, we conclude that
\[
|F(x)|\le C_1 x^{C_2}\mu^x \exp\biggl((\sqrt{2}c+C_3) x- \frac{1}{2}x\ln (x-a+1)\biggr).
\]
In particular, $F(N+x_0)\to 0$ as $N\to\infty$. This completes the proof.

\subsubsection{Proof of Lemma \ref{fzfinal2}}\label{fzfinal2pf}
The identities \eqref{gzid} and $\Gamma(z+1)=z\Gamma(z)$  show that for any $z\in \C$,
\[
G(z) = \frac{zG(z+1)}{\Gamma(z+1)}.
\] 
Iterating this, we get
\begin{align*}
G(z) = \frac{z(z+1)G(z+2)}{\Gamma(z+1)\Gamma(z+2)} = \frac{zG(z+2)}{(\Gamma(z+1))^2}.
\end{align*}
In particular, 
%Note that $\Re(\beta_j)=1+\sqrt{2}\Re(\alpha_j) >0$ for $j=1,2$. Since $\Re(w)>-1$, this shows that for any $z$ with $\Re(z)\ge \Re(w)$, we have 
\begin{align}\label{gziden}
G(z+\beta_j) = \frac{(z+\beta_j)G(z+\beta_j+2)}{\Gamma(z+\beta_j+1)^2}
\end{align}
for $j=1,2$. 
Also, by Lemma \ref{gammalmm2}, we have, for $z\in w + (\C \setminus(-\infty,0])$, 
\begin{align}\label{gammaiden}
(\Gamma(z-w))^{-z} &= \exp(-z\Pi(z-w)) \notag\\
&= \exp(-z\Pi(z-w+1)+ z\log (z-w))\notag\\
&= \exp(-z\Pi(z-w+2)+z\log(z-w+1)+z\log(z-w)).
\end{align}
By the identities \eqref{gziden} and \eqref{gammaiden}, and the definition of $f$, we have that for any $z$ in the domain $\Omega:= (w + (\C \setminus(-\infty,0]))\setminus\{-1,-2,\ldots\}$, 
\begin{align*}
f(z) %&= (z+\beta_1)(z+\beta_2)\frac{ (4\pi)^z e^{\frac{1}{2} z(z+1-2w)-2z}\Gamma(z+1) G(z+\beta_1+2)G(z + \beta_2+2)}{(\Gamma(z-w))^zG(\beta_1)G(\beta_1)(\Gamma(z+\beta_1+1)\Gamma(z+\beta_2+1))^2}\\
&= (z+\beta_1)(z+\beta_2)\exp\biggl\{z\log(4\pi) + \frac{1}{2}z(z-3-2w)  + \Theta(z+\beta_1+2) \\
&\qquad \qquad + \Theta(z + \beta_2 +2)-\Theta(\beta_1)-\Theta(\beta_2) + \Pi(z+1) - z\Pi(z-w+2)\\
&\qquad \qquad  + z\log(z-w)+ z\log(z-w+1)- 2\Pi(z+\beta_1+1) - 2\Pi(z+\beta_2+1)\biggr\}.
\end{align*}
Now take $z\in \Omega$ with $\Re(z)>-\frac{1}{2}$. Then by the given conditions,
\begin{align}\label{rbj}
\Re(z+\beta_j +1) = \Re(z)+ 2+\sqrt{2}\Re(\alpha_j) > \Re(z)+1 > \frac{1}{2}
\end{align}
for $j=1,2$. Also, $\Re(z+1)>\frac{1}{2}$, and 
\begin{align}\label{rwj}
\Re(z-w+2) > \frac{3}{2}-\Re(w)=\frac{5}{2} + \sqrt{2}(\Re(\alpha_1)+\Re(\alpha_2)) > \frac{1}{2}.
\end{align}
Thus, substituting the expansions of $\Theta$ and $\Pi$ from Corollary \ref{gcor} and Corollary \ref{picor}, we get 
\begin{align*}
f(z) &= (z+\beta_1)(z+\beta_2) \exp\biggl\{z\log(4\pi) + \frac{1}{2}z(z-3-2w)  + \frac{1}{2}(z+\beta_1)\ln (2\pi)\\
&\qquad   + \frac{1}{2}(z+\beta_2)\ln (2\pi)+\frac{1}{2}(z+\beta_1 + 1)^2 \log (z+\beta_1 + 1) \\
&\qquad  + \frac{1}{2}(z+\beta_2 + 1)^2 \log (z+\beta_2 + 1)- \frac{3}{4}(z+\beta_1+1)^2 - \frac{3}{4}(z+\beta_2+1)^2  \\
&\qquad   + (z+1)\log(z+1)- z - z(z-w+2) \log(z-w+2)+ z(z-w+2)   \\
&\qquad  + \frac{1}{2} z\log (z-w+2)-\frac{1}{2}z \ln (2\pi) + z\log(z-w)+z\log(z-w+1)  \\
&\qquad - 2(z+\beta_1+1)\log(z+\beta_1+1) + 2(z+\beta_1+1) \\
&\qquad - 2(z+\beta_2+1)\log(z+\beta_2+1) + 2(z+\beta_2+1) + Q_0(z)  \biggr\},
\end{align*}
where $|Q_0(z)|\le C_0\log(2+|z|)$ for some constant $C_0$ that depends only on $\alpha_1,\alpha_2$. Expanding the squares and collecting all terms of the form `constant times $z$' into a single term, we get the simplification 
\begin{align*}
f(z) &= (z+\beta_1)(z+\beta_2) \exp\biggl\{\biggl(A-\frac{1}{2}w\biggr)z + \frac{1}{2}z^2 +\frac{1}{2}(z^2+2(\beta_1 + 1)z) \log (z+\beta_1 + 1)\\
&\qquad  + \frac{1}{2}(z^2+2(\beta_2 + 1)z) \log (z+\beta_2 + 1) - \frac{3}{2}z^2  + z\log(z+1) \\
&\qquad  - (z^2-wz+2z) \log(z-w+2) + z^2 + \frac{1}{2} z\log (z-w+2) + z\log(z-w)  \\
&\qquad +z\log(z-w+1) - 2z\log(z+\beta_1+1) - 2z\log(z+\beta_2+1) + Q_1(z)  \biggr\},
\end{align*}
where $|Q_1(z)|\le C_1\log(2+|z|)$ for a constant $C_1$ that depends only on $\alpha_1,\alpha_2$, and $A$ is a real universal constant. In the above, terms of the form `constant times $z^2$' cancel out, and terms of the form `constant times $z^2\log(\ldots)$' add up to (using the lower bound \eqref{rwj})
\begin{align*}
&\frac{1}{2}z^2 \log \frac{(z+\beta_1+1)(z+\beta_2+1)}{(z-w+2)^2} = \frac{1}{2}z^2 \log \frac{z^2 + (\beta_1+\beta_2+2) z + (\beta_1+1)(\beta_2+1)}{(z-w+2)^2}\\
&= \frac{1}{2}z^2 \log \frac{z^2 + (3-w) z + (\beta_1+1)(\beta_2+1)}{(z-w+2)^2}\\
&= \frac{1}{2}z^2 \log \biggl(1+\frac{(w-1)z + (\beta_1+1)(\beta_2+1) -(2-w)^2}{(z-w+2)^2}\biggr)= \frac{(w-1)z^3}{2(z-w+2)^2} + O(1),
%&= \frac{1}{2}z^2 \biggl[\log \biggl(1 + \frac{-w+3}{z} + \frac{(\beta_1+1)(\beta_2+1)}{z^2}\biggr) - \log \biggl(1-\frac{2(w-2)}{z} + \frac{(w-2)^2}{z^2}\biggr)\biggr]\\
%&= \frac{1}{2}(w-1) z + O(1),
\end{align*}
where $O(1)$ denotes a term whose absolute value is bounded above by a constant that depends only on $\alpha_1,\alpha_2$. (We will follow this convention about $O$ in the remainder of this proof.) But 
\begin{align*}
\frac{(w-1)z^3}{2(z-w+2)^2} &= \frac{1}{2}(w-1)z - \frac{(w-1)z(2(2-w)z + (2-w)^2)}{2(z-w+2)^2}\\
&= \frac{1}{2}(w-1)z + O(1).
\end{align*}
Thus, we get the further simplification
\begin{align*}
f(z) &= (z+\beta_1)(z+\beta_2)\exp\biggl\{Bz  +(\beta_1 + 1)z\log (z+\beta_1 + 1)\\
&\qquad  + (\beta_2 + 1)z \log (z+\beta_2 + 1)  + z\log(z+1)+(w-2)z \log(z-w+2)  \\
&\qquad  + \frac{1}{2} z\log (z-w+2) + z\log(z-w)+z\log(z-w+1) \\
&\qquad - 2z\log(z+\beta_1+1)  - 2z\log(z+\beta_2+1) + O(\log(2+|z|))  \biggr\},
\end{align*}
where $B$ is a real universal constant. Now, by the inequalities \eqref{rbj} and \eqref{rwj}, we have 
\begin{align*}
&\log(z+\beta_j + 1) = \log(z-w+2) + O((1+|z|)^{-1}),\\
&\log(z+1) = \log(z-w+2)+O((1+|z|)^{-1}).
\end{align*}
%and the same holds for $\log(z+1)$ as well. 
This gives %us our final simplification
\begin{align*}
f(z) &= (z+\beta_1)(z+\beta_2) \exp\biggl\{Bz  +(\beta_1 + 1)z\log (z-w+2) + (\beta_2 + 1)z \log (z-w+2)  \\
&\qquad   + z\log(z-w+2) +(w-2)z \log(z-w+2) + \frac{1}{2} z\log (z-w+2)  \\
&\qquad  + z\log(z-w) +z\log(z-w+1)- 4z\log(z-w+2) + O(\log(2+|z|))  \biggr\}.
\end{align*}
But the coefficients of $z\log (z-w+2)$ add up to 
\[
(\beta_1 + 1) + (\beta_2 + 1) + 1 + (w-2)+ \frac{1}{2} - 4 = -\frac{3}{2}. %\sqrt{2}(\alpha_1+\alpha_2) + w +\frac{1}{2} = -\frac{1}{2}.
\]
This proves the desired expression for $f(z)$. For the bound on $|f(z)|$, note that if we have $\Re(z)\ge \Re(w)+\delta$ for some $\delta>0$, then
\[
|z\log (z-w+2) - z\log (z-w)| \le C
\]
for some constant $C$ that depends only on $\alpha_1$, $\alpha_2$, and $\delta$, and the same holds for the difference $|z\log (z-w+2) - z\log (z-w+1)| $. Consequently, 
\begin{align*}
&\biggl|\exp\biggl(z\log(z-w) +z\log(z-w+1)-\frac{3}{2}z\log (z-w+1)\biggr)\biggr|\\
&=\exp\biggl(\Re\biggl\{z\log(z-w)+z\log(z-w+1) -\frac{3}{2}z\log (z-w+1)\biggr\}\biggr)\\
&\le \exp\biggl(\Re\biggl\{\frac{1}{2}z\log (z-w)\biggr\}+C\biggr)\\
&= \exp\biggl(\frac{1}{2}\Re(z)\ln |z-w|-\frac{1}{2}\Im(z)\arg(z-w)+C\biggr).
%&= \biggl|\exp\biggl(\frac{1}{2}(x-w+\i y)(\log |z-w| + \i \arg(z-w))\biggr)\biggr|\\
%&= \exp\biggl(\frac{1}{2}(x-w)\log |z-w| - \frac{1}{2} y \arg(z-w) \biggr).
\end{align*}
Evidently, this completes the proof.

\subsubsection{Proof of Lemma \ref{twofinal0}}\label{twofinal0pf}
Let $g(z) := \Gamma(-z) f(z) (\mu e^{\sqrt{2}c})^z$ on the domain $(w+(\C \setminus(-\infty,0]))\setminus\Z$. Fix some $x_0\in (\Re(w),0)$. Take any $N\ge 1$ and $R\ge 1$. Let $C_R$ be the rectangular contour with vertices $x_0\pm \i R$, $N + x_0\pm \i R$, traversed counter-clockwise (see Figure \ref{fig:CR-with-poles}). Since $f$ is analytic in $\{z:\Re(z) > \Re(w)\}$, the only poles of $g$ are at the nonnegative integers, arising due to the poles of the Gamma function. As before, this gives
\begin{align}\label{cr1two}
\frac{1}{2\pi \i} \oint_{C_R} g(z) dz = -\sum_{n=0}^{N-1}\frac{(-1)^n}{n!}f(n)(\mu e^{\sqrt{2}c})^n.
\end{align}
Now, by Lemma \ref{fzfinal2} and Lemma \ref{gammalmm}, for any $x\ge x_0$ and $R\ge 0$, 
\begin{align}\label{gxirtwo}
|g(x + \i R)| &\le \frac{C_1(|x+\i R|+1)^2\exp(C_2\log(2+R)+ R \arg(\lceil x+1 \rceil -x- \i R))}{\prod_{j=0}^{\lceil x+1\rceil-1 }|j-x+ \i R|}\notag \\
&\quad \cdot \exp\biggl(\frac{1}{2}x\ln |x+ \i R-w| - \frac{1}{2}R\arg(x+ \i R-w) + Ax + C\biggr) e^{\sqrt{2}cx},
\end{align}
where  $A, C_1, C_2$ are universal constants and $C$ may depend only on $\alpha_1,\alpha_2, x_0$. Clearly, the bound is decreasing exponentially in $R$ as $R \to \infty$, uniformly over $x$ in any given bounded range.  The proof can now be completed just like the proofs of Lemma \ref{zerofinal0} and Lemma \ref{zerofinal0one}.

\subsubsection{Proof of Lemma \ref{zerofinaltwo}}\label{zerofinaltwopf}
Take any $N\ge 3$ and let $x:= N+x_0$. 
Using the inequality \eqref{prodjineq} in equation \eqref{gxirtwo}, we get
\begin{align*}
|g(x+\i y)|&\le C_1(\mu e^{\sqrt{2}c})^x\exp\biggl(C_2\log(2+|x|+| y|)+ y \arg(\lceil x+1 \rceil -x- \i y) \notag\\
&\quad - x\ln|x+\i y| + \frac{1}{2}x\ln |x+ \i y-w| - \frac{1}{2}y\arg(x+ \i y-w) + Ax\biggr),
\end{align*}
where $C_1, C_2$ are positive constants that depend only on $\alpha_1,\alpha_2, x_0$. The rest of the proof now goes through exactly as  the proof of Lemma \ref{zerofinalone}.

\subsubsection{Proof of Lemma \ref{heaviside}}\label{heavisidepf}
Note that 
\begin{align*}
\int_{-\infty}^\infty\int_0^{\infty}\varphi(x) e^{\i tx-\ep^2 t^2} dt dx &= \int_0^{\infty}\int_{-\infty}^\infty  \varphi(x) e^{\i tx-\ep^2 t^2} dx dt\\
&= \int_0^{\infty}e^{ - \ep^2 t^2 } \hat{\varphi}(t)dt,
\end{align*}
where $\hat{\varphi}$ is the Fourier transform of $\varphi$. 
Since $\varphi$ is smooth and compactly supported, $\hat{\varphi}$ is a Schwartz function. This allows us to apply the dominated convergence theorem and get that the limit of the above integral as $\ep \to 0$ is 
\[
\lim_{\ep\to 0} \int_0^{\infty}e^{ - \ep^2 t^2 } \hat{\varphi}(t)dt = \int_0^{\infty}  \hat{\varphi}(t)dt.
\]
Moreover, it also gives us the bound
\[
\biggl|\int_{-\infty}^\infty\int_0^{\infty}\varphi(x) e^{\i tx-\ep^2 t^2} dt dx\biggr| \le\int_0^{\infty}|\hat{\varphi}(t)|dt.
\]
But for the same reason, 
\begin{align*}
\int_0^{\infty}  \hat{\varphi}(t)dt &= \lim_{\ep \to 0} \int_0^{\infty}  e^{-\ep t} \hat{\varphi}(t)dt\\
&= \lim_{\ep \to 0}  \int_0^{\infty}\int_{-\infty}^\infty  \varphi(x)e^{\i t x - \ep t}dx dt\\
&= \lim_{\ep \to 0}  \int_{-\infty}^\infty  \int_0^{\infty} \varphi(x)e^{\i t x - \ep t}dt dx\\
&= \lim_{\ep\to 0} \int_{-\infty}^\infty  \frac{\varphi(x)}{-\i x + \ep }dx.
\end{align*}
Now, note that 
\begin{align*}
\int_{-\infty}^\infty  \frac{\varphi(x)}{-\i x + \ep }dx &= \int_{-\infty}^\infty  \frac{\varphi(x)(\ep +\i x)}{\ep^2 +  x^2 }dx\\
&= \int_{-\infty}^\infty  \frac{\varphi(\ep y)}{1+y^2 }dy + \i \int_0^\infty  \frac{x(\varphi(x)-\varphi(-x))}{\ep^2 + x^2 }dx.
\end{align*}
It is now easy to apply the dominated convergence theorem and show that the right side converges to the claimed limit as $\ep \to 0$.

\subsubsection{Proof of Lemma \ref{fzfinal3}}\label{fzfinal3pf}
Throughout this proof, $C, C_1,C_2,\ldots$ will denote arbitrary positive constants depending only on $\alpha_1$, $\alpha_3$, and $\delta$, whose values may change from line to line.  For the reader's convenience, let us recall that 
\begin{align*}%\label{fdef2}
f(z) &=\frac{(2\pi)^z e^{\frac{1}{2}z(z-3-2w)} \Gamma(z+1)G(z+2+\sqrt{2}\alpha_1) G(z+2+\sqrt{2}\alpha_3)}{\Gamma(z-w)^zG(1+\sqrt{2}\alpha_1)G(1+\sqrt{2}\alpha_3)}\\
&\qquad \cdot \biggl\{\frac{{_2F_1}(1, z-w-1; 1+\sqrt{2}\alpha_1; \frac{1}{2})}{2\Gamma(1+\sqrt{2}\alpha_1) \Gamma(z+1+\sqrt{2}\alpha_3)} - \frac{\sqrt{2}\alpha_3\, {_2F_1}(1,z-w-1;z+2+\sqrt{2}\alpha_1;\frac{1}{2})}{2\Gamma(z+2+\sqrt{2}\alpha_1)\Gamma(1+\sqrt{2}\alpha_3)}\biggr\}.
\end{align*}
Let $\alpha_1' := \alpha_1 + \frac{1}{\sqrt{2}}$. Then $G(z+2+\sqrt{2}\alpha_1) = G(z+1+\sqrt{2}\alpha'_1)$. Also, by equation~\eqref{gzid}, 
\[
G(z+2+\sqrt{2}\alpha_3) = G(z+1+\sqrt{2}\alpha_3) \Gamma(z+1+\sqrt{2}\alpha_3),
\]
and 
\[
G(1+\sqrt{2}\alpha_1) = \frac{G(1+\sqrt{2}\alpha_1')}{\Gamma(1+\sqrt{2}\alpha_1)}.
\]
Plugging these into the definition of $f$, we get
\begin{align*}%\label{fdef2}
f(z) &:= f_0(z) h(z),
\end{align*}
where 
\begin{align*}
f_0(z) &:= \frac{(4\pi)^z e^{\frac{1}{2}z(z-3-2w)} \Gamma(z+1)G(z+1+\sqrt{2}\alpha_1') G(z+1+\sqrt{2}\alpha_3)}{\Gamma(z-w)^zG(1+\sqrt{2}\alpha_1')G(1+\sqrt{2}\alpha_3)}
\end{align*}
and 
\begin{align*}
h(z) &:= 2^{-z-1}\,{_2F_1}\biggl(1, z-w-1; 1+\sqrt{2}\alpha_1; \frac{1}{2}\biggr) \\
&\qquad - \frac{2^{-z-1}\sqrt{2}\alpha_3\Gamma(z+1+\sqrt{2}\alpha_3)\Gamma(1+\sqrt{2}\alpha_1)\, {_2F_1}(1,z-w-1;z+2+\sqrt{2}\alpha_1;\frac{1}{2})}{\Gamma(z+2+\sqrt{2}\alpha_1)\Gamma(1+\sqrt{2}\alpha_3)}.
\end{align*}
Note that 
\[
w = -2-\sqrt{2}(\alpha_1+\alpha_3) = -1-\sqrt{2}(\alpha_1'+\alpha_3). 
\]
%Take any $z$ such that $\Re(z)>\Re(w)$. 
Take any $z$ as in the statement of the lemma.  Notice that the function $f_0$ defined above is exactly the function $f$ from Lemma~\ref{twoform1} with $\alpha_1'$ in place of $\alpha_1$ and $\alpha_3$ in place of $\alpha_2$ (and $w$ remains the same as above). Thus, by Lemma \ref{fzfinal2} and the assumptions that $\Re(w)>-\frac{1}{2}$ and $\Re(z)>\Re(w)+\delta$, we get the bound
\begin{align}\label{f0upper}
|f_0(z)| &\le  |(z+2+\sqrt{2}\alpha_1)(z+1+\sqrt{2}\alpha_3)|\exp\biggl(\frac{1}{2}\Re(z)\ln |z-w| - \frac{1}{2} \Im(z) \arg(z-w) \notag \\
&\qquad \qquad + C_1\Re(z)  + C_2\log(2+|z|)\biggr).
\end{align}
%the inequality \eqref{fcibd} and the assumption that $\Re(z)> \Re(w)>-\frac{1}{2}$, we get the following bound:
%\begin{align}\label{f0upper}
%|f_0(z)| &\le C_1 (1+|\Im(z)|)^{C_2} e^{-C_3|\Im(z)|}(1+|\Im(z-w)|^{\Re(w)}). % |(z+2+\sqrt{2}\alpha_1)(z+1+\sqrt{2}\alpha_3)|\exp\biggl(\frac{1}{2}\Re(z)\ln |z-w| - \frac{1}{2} \Im(z) \arg(z-w) \notag \\
%&\qquad \qquad + C_1\Re(z)  + C_2\log(2+|z|)\biggr).
%\end{align}
%for some constant $C'$ that depends only on $\alpha_1$, $\alpha_3$, and $\delta$. 
Let $b:= \Re(z-w-1)$, $t:= \Im(z-w+1)$, and $c:= 1+\sqrt{2}\alpha_1$. Then note that $b$ and $c-1$ are not nonpositive integers, since $b>-1$ and $b\ne 0$, and $c-1>-1$ and $c- 1=\sqrt{2}\alpha_1\ne 0$, where the last inequality holds because
\[
\Re(\sqrt{2}\alpha_1) = \Re(-w-2-\sqrt{2}\alpha_3) < \frac{1}{2} - 2 - \sqrt{2}\biggl(-\frac{1}{\sqrt{2}}\biggr)=-\frac{1}{2}.
\]
Also, $\Re(b) > -1$, 
\begin{align*}
\Re(1+b-c) &= \Re(z-w-1-\sqrt{2}\alpha_1)\\
&= \Re(z+1+\sqrt{2}\alpha_3) > \Re(z)\ge \Re(w)+\delta>-1+\delta,
\end{align*}
and by assumption, %since $\Re(1+\sqrt{2}\alpha_1)>2\delta$, 
\[
|\Re(1+b-c)| = |\Re(z+1+\sqrt{2}\alpha_3)|\ge \delta.
\]
Thus, we may apply Lemma \ref{hyperasymp2} to get
\begin{align}\label{f21}
\biggl|{_2F_1}\biggl(1, z-w-1; 1+\sqrt{2}\alpha_1; \frac{1}{2}\biggr)\biggr| &= \biggl|{_2F_1}\biggl(1, b+\i t; c; \frac{1}{2}\biggr)\notag \\
&\le C_1 (|b|+|t|+1)^{C_2} e^{C_3|b|}\notag \\
&\le C_1 (1+|z|)^{C_2} e^{C_3|\Re(z)|}.
\end{align}
Next, let $c:= \Re(z+2+\sqrt{2}\alpha_1)$, $t := \Im(z+2+\sqrt{2}\alpha_1)$, and $b := z-w-1-\i t$. Then for any nonnegative integer $n$,
\begin{align*}
|c+n| &\ge c+n = \Re(z+2+\sqrt{2}\alpha_1+n)\\
&> \Re(z+1+n)\ge \delta+n\ge \delta, 
\end{align*}
and $c\ge \Re(z+1)\ge \delta$. Lastly, note that 
\begin{align*}
|b| &= |z-w-1-\i\Im(z+2+\sqrt{2}\alpha_1)| \le |\Re(z)| + C.
\end{align*}
Thus, by Lemma \ref{hyperasymp},
\begin{align}\label{f22}
\biggl|{_2F_1}\biggl(1,z-w-1;z+2+\sqrt{2}\alpha_1;\frac{1}{2}\biggr)\biggr| &= \biggl|{_2F_1}(1,b+\i t; c+\i t;\frac{1}{2}\biggr)\biggr|\notag \\
&\le C_1 e^{C_2(|b|+|c|)}\notag \\
&\le C_1 e^{C_2|\Re(z)|}. 
\end{align}
Finally, note that for $j=1,3$, 
\[
\Re(z+2+\sqrt{2}\alpha_j) > \Re(w+2+\sqrt{2}\alpha_j)> \Re(w+1)\ge \delta.
\]
Thus, by Lemma \ref{gammalmm3} and the assumption that $|\Re(z+1+\sqrt{2}\alpha_3)|\ge \delta$ (and noting that $\Re(z+2+\sqrt{2}\alpha_j) > \Re(z+1)\ge \delta$ for $j=1,3$), we have
\begin{align}\label{f23}
\biggl|\frac{\Gamma(z+1+\sqrt{2}\alpha_3)}{\Gamma(z+2+\sqrt{2}\alpha_1)}\biggr| &= \biggl|\frac{\Gamma(z+2+\sqrt{2}\alpha_3)}{(z+1+\sqrt{2}\alpha_3)\Gamma(z+2+\sqrt{2}\alpha_1)}\biggr|\notag \\
&\le C_1(1+|z|)^{C_2}. %}{|z+1+\sqrt{2}\alpha_3|}.
\end{align}
By the inequalities \eqref{f21}, \eqref{f22}, and \eqref{f23}, we get
\begin{align*}
|h(z)| &\le C_1(1+|z|)^{C_2} e^{C_3|\Re(z)|}.
\end{align*}
Combining this with equation \eqref{f0upper} completes the proof.

\subsubsection{Proof of Lemma \ref{zerofinal0three}}\label{zerofinal0threepf}
Let $g(z) := \Gamma(-z) f(z) (\mu e^{\sqrt{2}c})^z$ on the domain $\{z:\Re(z)>\Re(w)\}$. Fix some $x_0\in (\Re(w),0)$. Take any $N\ge 1$ and $R\ge 1$. Let $C_R$ be the rectangular contour with vertices $x_0\pm \i R$, $N + x_0\pm \i R$, traversed counter-clockwise (see Figure \ref{fig:CR-with-poles}). Since $f$ is analytic in $\{z:\Re(z) > \Re(w)\}$, the only poles of $g$ are at the nonnegative integers, arising due to the poles of the Gamma function. As before, this gives
\begin{align}\label{cr1three}
\frac{1}{2\pi \i} \oint_{C_R} g(z) dz = -\sum_{n=0}^{N-1}\frac{(-1)^n}{n!}f(n)(\mu e^{\sqrt{2}c})^n.
\end{align}
Now, by Lemma \ref{fzfinal3} and Lemma \ref{gammalmm}, for any $x\ge x_0$ and $R\ge 0$, 
\begin{align}\label{gxirthree}
&|g(x + \i R)| \notag \\
&\le \frac{C_1(|x+\i R|+1)^2\exp(C_2\log(2+R)+ R \arg(\lceil x+1 \rceil -x- \i R))}{\prod_{j=0}^{\lceil x+1\rceil-1 }|j-x+ \i R|}\notag \\
&\quad \cdot (1+|x+\i R|)^{C_3}\exp\biggl(C_4|x|+ \frac{1}{2}x\ln |x+ \i R-w| - \frac{1}{2}R\arg(x+ \i R-w)\biggr),
\end{align}
where $C_1,C_2,C_3,C_4$ depend only on $\alpha_1,\alpha_3, x_0$. Clearly, the bound is decreasing exponentially in $R$ as $R \to \infty$, uniformly over $x$ in any given bounded range.  The proof can now be completed just like the proofs of Lemma \ref{zerofinal0}, Lemma \ref{zerofinal0one}, and Lemma \ref{twofinal0}.

\subsubsection{Proof of Lemma \ref{zerofinalthree}}\label{zerofinalthreepf}
%In this proof, $C, C_1, C_2,\ldots$ will denote arbitrary positive constants depending only on $\alpha_1,\alpha_3, x_0$, whose values may change from line to line.
%By Lemma \ref{zerofinal0three}, it suffices to show that $F(N+x_0)\to 0$ as $N\to \infty$. To prove this,
Take any $N\ge 3$ and let $x:= N+x_0$. 
Using the inequality \eqref{prodjineq} in equation \eqref{gxirthree}, we get
\begin{align*}
|g(x+\i y)|&\le C_1\exp\biggl(C_2x+C_3\log(2+|x+\i y|)+ y \arg(\lceil x+1 \rceil -x- \i y) \notag\\
&\quad - x\ln|x+\i y| + \frac{1}{2}x\ln |x+ \i y-w| - \frac{1}{2}y\arg(x+ \i y-w)\biggr),
\end{align*}
where $C_1,C_2,C_3$ depend only on $\alpha_1,\alpha_3, x_0, \mu$. 
The rest of the proof now goes through exactly as  the proof of Lemma \ref{zerofinalone}.

\end{document}